\documentclass[runningheads,envcountsect]{llncs}
\newif\ifdraft\draftfalse

\newif\iffull\fulltrue %

\iffull
\else
\usepackage{showframe}  %
\fi 
\usepackage{graphicx}
\usepackage{hyperref}

\iffull
\else
\usepackage[firstpage]{draftwatermark}  %

\SetWatermarkAngle{0}

\SetWatermarkText{\raisebox{16cm}{%
  \hspace{0.1cm}%
  \href{https://doi.org/10.6084/m9.figshare.22722151.v4}{\includegraphics{1-available}}%
  \hspace{9cm}%
  \includegraphics{2-functional}%
}}
\fi

\makeatletter
\def\orcidID#1{\kern .08em\href{https://orcid.org/#1}{\includegraphics[keepaspectratio,width=0.9em]{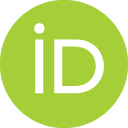}}}
\makeatother

\usepackage{amsmath,amssymb}
\usepackage{tikz}
\usetikzlibrary{calc}
\usetikzlibrary {arrows.meta}
\usepackage{cleveref}
\usepackage{caption}
\usepackage{subcaption}
\usepackage{mathtools}
\usepackage{xcolor}
\usepackage{algorithm}
\usepackage{algpseudocode}
\usepackage{booktabs}
\usepackage{proof}
\usepackage{tikz-cd}
\usepackage{sepnum}
\usepackage{cite}

\usepackage[all]{xy} \SelectTips{cm}{}  %
\newdir{ >}{{}*!/-8pt/@{>}}  
\newdir{|>}{%
	!/1.6pt/@{|}*:(1,-.2)@^{>}*:(1,+.2)@_{>}}

\iffull
\usepackage[appendix=append]{apxproof}
\else
\fi
\usepackage{wrapfig}
\usepackage{comment}
\ifdraft
\specialcomment{auxproof}
{\mbox{}\newline\textbf{BEGIN: AUX-PROOF}\dotfill\newline}
{\mbox{}\newline\textbf{END: AUX-PROOF}\dotfill\newline}
\else
\excludecomment{auxproof}
\fi

\spnewtheorem{notation}[theorem]{Notation}{\bfseries}{\upshape}
\spnewtheorem{mydefinition}[theorem]{Definition}{\bfseries}{\itshape}

\newcommand\myindent{\setlength\parindent{14.5pt}}

\renewcommand{\cref}[1]{\Cref{#1}}
\creflabelformat{enumi}{#2(#1)#3}
\crefname{theorem}{Thm.}{Theorems}
\crefname{mydefinition}{Def.}{Defs}
\crefname{proposition}{Prop.}{Props}
\crefname{lemma}{Lem.}{Lemmas}
\crefname{proof}{Proof.}{Proofs}
\crefname{appendix}{Appendix}{Appendixes}
\crefformat{section}{{\S}#2#1#3}
\crefname{figure}{Fig.}{Figs}
\Crefname{equation}{}{}
\newcommand{\dr}{\mathbf{r}}
\newcommand{\dl}{\mathbf{l}}

\newcommand{\nfl}[2][]{#2^{#1}_{\dl}}
\newcommand{\nfr}[2][]{#2^{#1}_{\dr}}

\newcommand{\bfn}[2][]{\overline{#2}^{#1}}

\newcommand{\nset}[1]{[#1]}
\newcommand{\Nat}{{\mathbb N }}

\newcommand{\Real}{{\mathbb R }}
\newcommand{\Realp}{{\mathbb R_{\geq 0} }}
\newcommand{\defeq}{:=}
\newcommand{\ocf}{\mathrm{CompMDP}}
\newcommand{\pto}{\rightharpoonup}

\newcommand{\kle}{\simeq}
\newcommand{\kli}{\lesssim}
\newcommand{\klg}{\lesssim}
\newcommand{\fary}[1]{{#1}^{\ast}}
\newcommand{\iary}[1]{{#1}^{\omega}}
\newcommand{\cary}[1]{\fary{#1}\cup\iary{#1}}

\newcommand{\dunit}{d}
\newcommand{\dcounit}{e}
\newcommand{\id}{\mathrm{id}}
\newcommand{\ID}{\mathcal{I}}
\newcommand{\SYM}{\mathcal{S}}

\newcommand{\tr}{\mathrm{tr}}
\newcommand{\trace}[4]{\tr^{#4}_{#1;#2,#3}}
\newcommand{\trmdp}[3]{\trace{#1}{#2}{#3}{}}%
\newcommand{\Int}{\mathrm{Int}}
\newcommand{\pMnd}{\mathcal{P}}

\newcommand{\swap}[3][,]{\sigma_{#2 #1 #3}}

\newcommand{\Sets}{\mathbf{Set}}

\newcommand{\mcMnd}{T^{\mathrm{PR}}}

\newcommand{\kleisli}[2]{\mathit{K}\!\ell(#2)}
\newcommand{\rest}[2]{#1 {\upharpoonright}_{#2}}

\newcommand{\mdp}[1]{\mathcal{#1}}

\newcommand{\mc}[1]{\mathcal{#1}}
\newcommand{\indmc}[2]{\mathrm{MC}(#1, #2)}
\newcommand{\mypath}[2]{\pi^{(#1, #2)}}

\newcommand{\Pathstra}[4]{\mathrm{Path}^{#1}(#3, #4)}
\newcommand{\Prob}{\mathrm{Pr}}
\newcommand{\ProPath}[2]{\Prob^{#1, #2}}
\newcommand{\ProPathc}[1]{\Prob^{#1}}

\newcommand{\Reacha}[3]{\mathrm{RPr}^{#3}(#1, #2)}
\newcommand{\Reachstr}[4]{\mathrm{RPr}^{#1, #2}(#3, #4)}
\newcommand{\Rew}{\mathrm{Rw}}
\newcommand{\PReward}[1]{\Rew^{#1}}

\newcommand{\ETRa}[3]{\mathrm{ERw}^{#3}(#1, #2)}
\newcommand{\ETRstr}[4]{\mathrm{ERw}^{#1,#2}(#3, #4)}

\newcommand{\totalAction}{E}
\newcommand{\place}{\underline{\phantom{n}}\,} %
\newcommand{\seqcomp}{\mathbin{;}}
\newcommand{\semFunctor}{\mathcal{S}}
\newcommand{\semFunctorr}{\semFunctor_{\dr}}
\newcommand{\semFunctorrmc}{\semFunctor^{\mathrm{MC}}_{\dr}}

\newcommand{\oMDP}{\mathbf{oMDP}}
\newcommand{\roMDP}{\mathbf{roMDP}}
\newcommand{\roMC}{\mathbf{roMC}}
\newcommand{\semCat}{\mathbb{S}}
\newcommand{\semCatr}{\semCat_{\dr}}

\newcommand{\semCatrmc}{\semCat^{\mathrm{MC}}_{\dr}}

\DeclareMathOperator{\exits}{exits}
\DeclareMathOperator{\precedent}{prec}

\newcommand{\Estar}[2]{\ensuremath{E^{\star,{#1}}_{#2}}}
\newcommand{\kron}[2]{\ensuremath{\delta_{{#1},{#2}}}}

\newcommand{\myparagraph}[1]{\vspace*{.3em}\noindent\textbf{#1}\;}

\ifdraft
\newcommand\asd[1]{{\footnotesize \color[RGB]{105,10,130}[#1 -asd]}}
\newcommand\kzk[1]{{\footnotesize \color{blue}[#1 -kzk]}}
\newcommand\clovis[1]{{\footnotesize \color[RGB]{0,200,30}[#1 -clovis]}}
\newcommand\ichiro[1]{{\footnotesize \color[RGB]{105,10,10}[#1 -ichiro]}}

\newcommand\del[2]{\modi{#1}{#2}{}{}}
\usepackage{todonotes}
\else %
\newcommand{\todo}[2][]{}
\newcommand\asd[1]{}
\newcommand\kzk[1]{}
\newcommand\clovis[1]{}
\newcommand\ichiro[1]{}

\newcommand\del[2]{}
\fi

\begin{document}
\title{
Compositional Probabilistic Model Checking with String Diagrams of MDPs
\thanks{The authors are supported by ERATO HASUO Metamathematics for Systems Design Project (No.\ JPMJER1603), JST. K.W.\ is supported by the JST grant No.\ JPMJFS2136.
}
}
\author{
Kazuki Watanabe\inst{1,2}\orcidID{0000-0002-4167-3370}\and
Clovis Eberhart\inst{1,3}\orcidID{0000-0003-3009-6747} \and\\
Kazuyuki Asada\inst{4}\orcidID{0000-0001-8782-2119} \and
Ichiro Hasuo\inst{1,2}\orcidID{0000-0002-8300-4650}
}
\authorrunning{
K. Watanabe et al.
}
\institute{
National Institute of Informatics, Tokyo, Japan\\
\email{\{kazukiwatanabe,eberhart,hasuo\}@nii.ac.jp} \and
The Graduate University for Advanced Studies (SOKENDAI), Hayama, Japan \and
Japanese-French Laboratory of Informatics, IRL 3527, CNRS, Tokyo, Japan \and 
Tohoku University, Sendai, Japan
\\
\email{kazuyuki.asada.b6@tohoku.ac.jp}\\
}
\maketitle              %
\begin{abstract}
We present a compositional model checking algorithm for Markov decision processes, in which they are composed in the categorical graphical language of \emph{string diagrams}. The algorithm  computes optimal expected rewards. Our theoretical development of the algorithm is supported by category theory, 
while what we call decomposition equalities for expected rewards act as a key enabler.  Experimental evaluation demonstrates its performance advantages.
\end{abstract}

\section{Introduction}
\label{sec:intro}

\emph{Probabilistic model checking} is a topic that attracts both theoretical and practical interest. On the practical side, probabilistic system models can naturally accommodate uncertainties inherent in many real-world systems; moreover, probabilistic model checking can give quantitative answers, enabling more fine-grained assessment than qualitative verification. 
Model checking of Markov decision processes (MDPs)---the target problem of this paper---has additional practical values since it not only verifies a specification but also synthesizes an optimal control strategy. 
On the theoretical side, it is notable that probabilistic model checking has a number of efficient algorithms, despite the challenge that the problem involves continuous quantities (namely probabilities).  See e.g.~\cite{BaierKatoen08}.

However, even those efficient algorithms can struggle when a model is enormous. Models can easily become enormous---the so-called \emph{state-space explosion problem}---due to the growing complexity of modern verification targets. Models that exceed the memory size of a machine for verification are common.

Among possible countermeasures to state-space explosion, one with both mathematical blessings and a proven track record is \emph{compositionality}. It takes as input a model with a compositional structure---where smaller \emph{component} models are combined, sometimes with many layers---and processes the model in a divide-and-conquer manner. In particular, when there is repetition among components, compositional methods can exploit the repetition and reuse intermediate results, leading to a clear performance advantage.

Focusing our attention to MDP model checking, 
 there have been many compositional  methods proposed for various settings. One example is~\cite{KwiatkowskaNPQ13}: 
it studies probabilistic automata (they are only slightly different from MDPs) and in particular their \emph{parallel composition}; the proposed method is a compositional framework, in an assume-guarantee style, based on multi-objective probabilistic model checking. Here, \emph{contracts} among parallel components are not always automatically obtained.  Another example is~\cite{JungesS22}, where the so-called \emph{hierarchical model checking} method for MDPs is introduced. It deals with \emph{sequential composition} rather than parallel composition; assuming what can be called \emph{parametric homogeneity} of components---they must be of the same shape while parameter values may vary---they present a model-checking algorithm that computes a guaranteed interval for the optimal expected reward. 

In this work, inspired by these works and technically building on another recent work of ours~\cite{Watanabe21}, we present another compositional MDP model checking algorithm. We compose MDPs in \emph{string diagrams}---a graphical language of category theory~\cite[Chap.~XI]{MacLane2} that has found applications in computer science~\cite{heunen2019categories,PiedeleuKCS15,BonchiHPSZ19}---that are more sequential than parallel. Our algorithm computes the optimal expected reward, unlike~\cite{JungesS22}. 

One key ingredient of the algorithm is the identification of compositionality as the \emph{preservation of algebraic structures}; more specifically, we identify a compositional solution as a ``homomorphism'' of suitable \emph{monoidal categories}. This identification guided us in our development, explicating requirements of a desired compositional semantic domain (\cref{sec:compositionalMDPs}). 

Another key ingredient is a couple of \emph{decomposition equalities} for reachability probabilities, extended to expected rewards (\cref{sec:compositionalAnalyssiOfOpenMarkovChains}). Those for reachability probabilities are well-known---one of them is \emph{Girard's execution formula}~\cite{girard1989geometry} in linear logic---but our extension to expected rewards seems new. 

The last two key ingredients are combined in \cref{sec:FatSemanticCategories} to formulate a compositional solution. Here we benefit from general categorical constructions, namely the \emph{$\Int$ construction}~\cite{joyal1996} and \emph{change of base}~\cite{Eilenberg65,cruttwell2008normed}.

We implemented the algorithm (it is called $\ocf$) and present its experimental evaluation. Using the benchmarks inspired by real-world problems, we show that 1) $\ocf$ can solve huge models in realistic time (e.g.\ $10^8$ positions, in 6--130 seconds); 2) compositionality does boost performance (in some ablation experiments); and 3) the choice of the degree of compositionality is important. The last is enabled in $\ocf$ by the operator we call \emph{freeze}.

\begin{figure}[tb]
    \centering
     \begin{subfigure}[b]{0.22\columnwidth}
        \centering
        \begin{tikzpicture}[
              innode/.style={draw, rectangle, minimum size=1cm},
              innodemini/.style={draw, rectangle, minimum size=0.5cm},
              interface/.style={inner sep=0},
              innodepos/.style={draw, circle, minimum size=0.2cm},
              ]
              \fill[orange](0.1cm, -0.5cm)--(0.1cm, 1.1cm)--(1.85cm, 1.1cm)--(1.85cm, -0.5cm)--cycle;
              \node[interface] (rdo1) at (0cm, 0cm) {};
              \node[innodepos, fill=white] (pos1) at (1cm, 0cm) {$4$} edge[loop above] node {$0.2$} (pos1);
              \node[interface] (rcdo1) at (2cm, 0cm) {};
              \draw[->] (rdo1) to (pos1);
              \draw[->] (pos1) to node[above] {$0.8$} (rcdo1); 
        \end{tikzpicture}
        \caption{A \emph{task} $\mdp{A}_{i}^{\mathrm{task}}$.}
    \end{subfigure}    
    \hfill
    \begin{subfigure}[b]{0.31\columnwidth}
        \centering
        \begin{tikzpicture}[
              innode/.style={draw, rectangle, minimum size=0.5cm},
              innodemini/.style={draw, rectangle, minimum size=0.5cm},
              interface/.style={inner sep=0},
              innodepos/.style={draw, circle, minimum size=0.2cm},
              ]
              \fill[cyan](0.1cm, -0.5cm)--(0.1cm, 0.5cm)--(3.85cm, 0.5cm)--(3.85cm, -0.5cm)--cycle;
              \node[interface] (rdo1) at (0cm, 0cm) {};
              \node[innode, fill=orange] (pos1) at (1cm, 0cm) {\scalebox{0.8}{$\mdp{A}^{\mathrm{task}}_1$}};
              \node[innode, fill=orange] (pos2) at (3cm, 0cm) {\scalebox{0.8}{$\mdp{A}^{\mathrm{task}}_{m_1}$}};
              \node[interface] (mid1) at (1.8cm, 0cm) {};
              \node[interface] (mid1) at (2cm, 0cm) {$\cdots$};
              \node[interface] (mid2) at (2.2cm, 0cm) {};
              \node[interface] (rcdo1) at (4cm, 0cm) {};
              \draw[->] (rdo1) to (pos1);
              \draw[->] (pos1) to (mid1); 
              \draw[->] (mid2) to (pos2);
              \draw[->] (pos2) to (rcdo1);
        \end{tikzpicture}
        \caption{A \emph{room} $\mdp{A}_{i}^{\mathrm{room}}$ combines tasks.}
    \end{subfigure}
    \hfill
    \begin{subfigure}[b]{0.42\columnwidth}
        \centering
        \begin{tikzpicture}[
              innode/.style={draw, rectangle, minimum size=0.4cm},
              innodemini/.style={draw, rectangle, minimum size=0.5cm},
              interface/.style={inner sep=0},
              innodepos/.style={draw, circle, minimum size=0.2cm},
              ]
              \fill[pink](-0.1cm, -1.8cm)--(-0.1cm, 2cm)--(3.85cm, 2cm)--(3.85cm, -1.8cm)--cycle;
              \node[interface] (rdo1) at (-0.3cm, -1.1cm) {};
              \node[interface] (ldo1) at (-0.3cm, -0.9cm) {};
              \node[innodepos, fill=white] (pos1) at (0.5cm, -1cm) {\scalebox{0.5}{$0$}};
              \node[innode, fill=cyan] (room1) at (1.5cm, -1cm) {\scalebox{0.6}{$\mdp{A}^{\mathrm{room}}_{1}$}};
              \node[innodepos, fill=white] (pos2) at (2.5cm, -1cm) {\scalebox{0.5}{$0$}};
              \node[innodepos, fill=white] (pos3) at (1cm, 0cm) {\scalebox{0.5}{$0$}};
              \node[interface,rotate=70] (mid1) at (1.15cm, 0.5cm) {$\cdots$};
              \node[interface,rotate=70] (mid2) at (1.3cm, 1cm) {$\cdots$};
              \node[innode, fill=cyan] (room2) at (2cm, 0cm) {\scalebox{0.6}{$\mdp{A}^{\mathrm{room}}_{2}$}};
              \node[innodepos, fill=white] (pos4) at (3cm, 0cm) {\scalebox{0.5}{$0$}};
              \node[interface,rotate=70] (mid3) at (3.15cm, 0.5cm) {$\cdots$};
              \node[innodepos, fill=white] (pos5) at (1.5cm, 1.5cm) {\scalebox{0.5}{$0$}};
              \node[innode, fill=cyan] (room3) at (2.5cm, 1.5cm) {\scalebox{0.6}{$\mdp{A}^{\mathrm{room}}_{m_2}$}};
              \node[innodepos, fill=white] (pos6) at (3.5cm, 1.5cm) {\scalebox{0.5}{$0$}};
              \node[interface] (rcdo1) at (4.0cm, 1.4cm) {};
              \node[interface] (lcdo1) at (4.0cm, 1.6cm) {};
              \node[interface] (imid1)  at (0.8cm, 0.5cm) {};
              \node[interface] (imid2)  at (1.4cm, 0.3cm) {};
              \node[interface] (imid3)  at (1.1cm, 1.1cm) {};
              \node[interface] (imid4)  at (1.5cm, 1cm) {};
              \node[interface] (imid5)  at (3.4cm, 0.2cm) {};
              \node[interface] (imid6)  at (3.6cm, 1cm) {};
              \draw[->] (rdo1) to (pos1);
              \draw[<-] (ldo1) to node[above] {\scalebox{0.5}{$\alpha, .9$}} (pos1);
              \draw[->] (pos1) to node[below] {\scalebox{0.5}{$\beta,.8$}} (room1);
              \draw[->] (room1) to (pos2);
              \draw[->] (pos1) to [bend right=70] node[above] {\scalebox{0.5}{$\alpha,.1,\, \beta, .2,\, \gamma,.3$}}(pos2);
              \draw[->] (pos1) to [bend right=30] node[left] {\scalebox{0.5}{$\gamma, .7$}} (pos3);
              \draw[<-] (pos1) to [bend left=30] node[left] {\scalebox{0.5}{$\alpha, .9$}} (pos3);
              \draw[->] (pos3) to [bend right=70] node[above] {\scalebox{0.5}{$\alpha,.1,\, \beta, .2,\, \gamma,.3$}}(pos4);
              \draw[->] (pos3) to node[below] {\scalebox{0.5}{$\beta,.8$}} (room2);
              \draw[->] (room2) to (pos4);
              \draw[->] (pos2) to [bend right=30] (pos4);
              \draw[->] (pos5) to [bend right=70] node[above] {\scalebox{0.5}{$\alpha,.1,\, \beta, .2,\, \gamma,1$}}(pos6);
              \draw[->] (pos5) to node[below] {\scalebox{0.5}{$\beta,.8$}} (room3);
              \draw[->] (room3) to (pos6);
              \draw[->] (pos6) to (rcdo1);
              \draw[->] (lcdo1) to [bend right=30] (pos5);
              \draw[->] (imid1) to (pos3);
              \draw[->] (pos3) to (imid2);
              \draw[->] (pos5) to (imid3);
              \draw[->] (imid4) to (pos5);
              \draw[->] (pos4) to (imid5);
              \draw[->] (imid6) to (pos6);
        \end{tikzpicture}
        \caption{A \emph{floor} $\mdp{A}_{i}^{\mathrm{floor}}$  combines rooms.}
    \end{subfigure}
    \hfill
    \begin{subfigure}[b]{0.42\columnwidth}
        \centering
        \begin{tikzpicture}[
              innode/.style={draw, rectangle, minimum size=0.5cm},
              innodemini/.style={draw, rectangle, minimum size=0.5cm},
              interface/.style={inner sep=0},
              innodepos/.style={draw, circle, minimum size=0.2cm},
              ]
              \fill[lime](0.1cm, -0.5cm)--(0.1cm, 0.8cm)--(3.85cm, 0.8cm)--(3.85cm, -0.5cm)--cycle;
              \node[interface] (rdo1) at (0cm, -0.125cm) {};
              \node[innode, fill=pink] (floor1) at (1cm, 0cm) {\scalebox{0.8}{$\mdp{A}^{\mathrm{floor}}_1$}};
              \node[innode, fill=pink] (floor2) at (3cm, 0cm) {\scalebox{0.8}{$\mdp{A}^{\mathrm{floor}}_{m_3}$}};
              \node[interface] (mid1) at (2cm, 0cm) {$\cdots$};
              \node[interface] (mid2) at (2.2cm, 0cm) {};
              \node[interface] (rcdo1) at (4cm, -0.125cm) {};
              \draw[->] (rdo1) to ($(floor1.west)!0.5!(floor1.south west)$);
              \draw[->] ($(floor1.east)!0.5!(floor1.south east)$) to (1.8cm, -0.125cm);
              \draw[->] (2.2cm, -0.125cm) to ($(floor2.west)!0.5!(floor2.south west)$); 
              \draw[<-] ($(floor1.east)!0.5!(floor1.north east)$) to (1.8cm, 0.125cm);
              \draw[<-] (2.2cm, 0.125cm) to ($(floor2.west)!0.5!(floor2.north west)$); 
              \draw[->] ($(floor2.east)!0.5!(floor2.south east)$) to (rcdo1);
              \draw[-] (0.55cm, 0.725cm) arc [radius=0.3, start angle = 90, end angle=270];
              \draw[<-] (3.45cm, 0.125cm) arc [radius=0.3, start angle = 270, end angle=450];
              \draw[-] (0.55cm, 0.725cm) to (3.45cm, 0.725cm);
        \end{tikzpicture}
        \caption{A \emph{building} $\mdp{A}_{i}^{\mathrm{bldg}}$  combines floors.}
    \end{subfigure}
    \hfill
     \begin{subfigure}[b]{0.42\columnwidth}
        \centering
        \begin{tikzpicture}[
              innode/.style={draw, rectangle, minimum size=0.5cm},
              innodemini/.style={draw, rectangle, minimum size=0.5cm},
              interface/.style={inner sep=0},
              innodepos/.style={draw, circle, minimum size=0.2cm},
              ]
              \node[innode, fill=lime] (bldg1) at (1cm, 0cm) {\scalebox{0.8}{$\mdp{A}^{\mathrm{bldg}}_1$}};
              \node[interface] (mid1) at (2cm, 0cm) {$\cdots$};
              \node[innode, fill=lime] (bldg2) at (3cm, 0cm) {\scalebox{0.8}{$\mdp{A}^{\mathrm{bldg}}_{m_4}$}};
              \draw[->] (0cm, 0cm) to (bldg1);
              \draw[->] (bldg1) to (1.8cm, 0cm);
              \draw[->] (2.2cm, 0cm) to (bldg2);
              \draw[->] (bldg2) to (4cm, 0cm);
        \end{tikzpicture}
        \caption{A \emph{neighborhood} $\mdp{A}^{\mathrm{nbd}}$  combines buildings.}
    \end{subfigure}

\hfill
    \caption{String diagrams of MDPs, an example (the Patrol benchmark in \cref{sec:impl}).}
    \label{fig:string_diagrams_of_mdps}

 \vspace{.5em}
\centering
\begin{math}
   \begin{aligned}
   \scalebox{0.7}{
   \begin{tikzpicture}[
              innode/.style={draw, rectangle, minimum size=1cm},
              innodemini/.style={draw, rectangle, minimum size=0.5cm},
              interface/.style={inner sep=0},
              innodeeve/.style={draw, circle, minimum size=0.2cm},
              innodeadam/.style={draw, rectangle, minimum size=0.2cm},
              ]
              \node[interface] (rdo1) at (0cm, 0.25cm) {};
              \node[interface] (rdo2) at (0cm, 0cm) {};
              \node[interface] (rdo3) at (0cm, -0.25cm) {};
               \node[interface] (rcdo1) at (2cm, 0.25cm) {};
              \node[interface] (rcdo2) at (2cm, 0cm) {};
              \node[interface] (rcdo3) at (2cm, -0.25cm) {};
              \node[innode] (game1) at (1cm, 0cm) {\scalebox{2}{$\mdp{A}$}};
              \draw[->] (rdo1) to ($(game1.north west)!0.5!(game1.west)$);
              \draw[->] (rdo2) to (game1);
              \draw[<-] (rdo3) to ($(game1.south west)!0.5!(game1.west)$);
             \draw[->] ($(game1.north east)!0.5!(game1.east)$) to (rcdo1);
             \draw[<-] (game1) to (rcdo2);
             \draw[<-] ($(game1.south east)!0.5!(game1.east)$) to (rcdo3);
             \node[interface] (seqcomp) at (2.5cm, 0cm) {\scalebox{2}{$\seqcomp$}};
             \node[interface] (rdo1b) at (3cm, 0.25cm) {};
             \node[interface] (rdo2b) at (3cm, 0cm) {};
             \node[interface] (rdo3b) at (3cm, -0.25cm) {};
             \node[interface] (rcdo1b) at (5cm, 0.25cm) {};
             \node[interface] (rcdo2b) at (5cm, -0.25cm) {};
             \node[innode] (game2) at (4cm, 0cm) {\scalebox{2}{$\mdp{B}$}};
              \draw[->] (rdo1b) to ($(game2.north west)!0.5!(game2.west)$);
              \draw[<-] (rdo2b) to (game2);
              \draw[<-] (rdo3b) to ($(game2.south west)!0.5!(game2.west)$);
             \draw[->] ($(game2.north east)!0.5!(game2.east)$) to (rcdo1b);
             \draw[<-] ($(game2.south east)!0.5!(game2.east)$) to (rcdo2b);

             \node[interface] (equal) at (6.5cm, 0cm) {\scalebox{3}{$=$}};

             \node[interface] (rdo1ab) at (8cm, 0.25cm) {};
             \node[interface] (rdo2ab) at (8cm, 0cm) {};
             \node[interface] (rdo3ab) at (8cm, -0.25cm) {};
             \node[innode] (game1ab) at (9cm, 0cm) {\scalebox{2}{$\mdp{A}$}};
             \node[innode] (game2ab) at (11cm, 0cm) {\scalebox{2}{$\mdp{B}$}};
             \node[interface] (rcdo1ab) at (12cm, 0.25cm) {};
             \node[interface] (rcdo2ab) at (12cm, -0.25cm) {};
             \draw[->] (rdo1ab) to ($(game1ab.north west)!0.5!(game1ab.west)$);
             \draw[->] (rdo2ab) to (game1ab);
             \draw[<-] (rdo3ab) to ($(game1ab.south west)!0.5!
             (game1ab.west)$);
             \draw[->] ($(game1ab.north east)!0.5!
             (game1ab.east)$) to ($(game2ab.north west)!0.5!
             (game2ab.west)$);
             \draw[<-] (game1ab) to (game2ab);
             \draw[<-] ($(game1ab.south east)!0.5!
             (game1ab.east)$) to ($(game2ab.south west)!0.5!
             (game2ab.west)$);
             \draw[->] ($(game2ab.north east)!0.5!(game2ab.east)$) to (rcdo1ab);
             \draw[<-] ($(game2ab.south east)!0.5!(game2ab.east)$) to (rcdo2ab);
          \end{tikzpicture}
    }
    \\
    \scalebox{0.7}{
    \begin{tikzpicture}[
              innode/.style={draw, rectangle, minimum size=1cm},
              innodemini/.style={draw, rectangle, minimum size=0.5cm},
              interface/.style={inner sep=0},
              innodeeve/.style={draw, circle, minimum size=0.2cm},
              innodeadam/.style={draw, rectangle, minimum size=0.2cm},
              ]
              \node[interface] (rdo1) at (0cm, 0.25cm) {};
              \node[interface] (rdo2) at (0cm, 0cm) {};
              \node[interface] (rdo3) at (0cm, -0.25cm) {};
               \node[interface] (rcdo1) at (2cm, 0.25cm) {};
              \node[interface] (rcdo2) at (2cm, 0cm) {};
              \node[interface] (rcdo3) at (2cm, -0.25cm) {};
              \node[innode] (game1) at (1cm, 0cm) {\scalebox{2}{$\mdp{A}$}};
              \draw[->] (rdo1) to ($(game1.north west)!0.5!(game1.west)$);
              \draw[->] (rdo2) to (game1);
              \draw[<-] (rdo3) to ($(game1.south west)!0.5!(game1.west)$);
             \draw[->] ($(game1.north east)!0.5!(game1.east)$) to (rcdo1);
             \draw[<-] (game1) to (rcdo2);
             \draw[<-] ($(game1.south east)!0.5!(game1.east)$) to (rcdo3);
             \node[interface] (sum) at (2.5cm, 0cm) {\scalebox{2}{$\oplus$}};
             \node[interface] (rdo1b) at (3cm, 0.25cm) {};
             \node[interface] (rdo2b) at (3cm, 0cm) {};
             \node[interface] (rdo3b) at (3cm, -0.25cm) {};
             \node[interface] (rcdo1b) at (5cm, 0.25cm) {};
             \node[interface] (rcdo2b) at (5cm, -0.25cm) {};
             \node[innode] (game2) at (4cm, 0cm) {\scalebox{2}{$\mdp{B}$}};
              \draw[->] (rdo1b) to ($(game2.north west)!0.5!(game2.west)$);
              \draw[<-] (rdo2b) to (game2);
              \draw[<-] (rdo3b) to ($(game2.south west)!0.5!(game2.west)$);
             \draw[->] ($(game2.north east)!0.5!(game2.east)$) to (rcdo1b);
             \draw[<-] ($(game2.south east)!0.5!(game2.east)$) to (rcdo2b);

             \node[interface] (equal) at (6.5cm, 0cm) {\scalebox{3}{$=$}};

             \node[interface] (rdo1aba) at (8cm, 1.25cm) {};
             \node[interface] (rdo2aba) at (8cm, 1cm) {};
             \node[interface] (rdo3aba) at (8cm, 0.75cm) {};
             \node[interface] (rdo1abb) at (8cm, -0.75cm) {};
             \node[interface] (rdo2abb) at (8cm, -1cm) {};
             \node[interface] (rdo3abb) at (8cm, -1.25cm) {};
             \node[interface] (rcdo1aba) at (12cm, 1.25cm) {};
             \node[interface] (rcdo2aba) at (12cm, 1cm) {};
             \node[interface] (rcdo3aba) at (12cm, 0.75cm) {};
             \node[interface] (rcdo1abb) at (12cm, -0.75cm) {};
             \node[interface] (rcdo2abb) at (12cm, -1cm) {};
             \node[interface] (rcdo3abb) at (12cm, -1.25cm) {};
             \node[innode] (game1ab) at (10cm, 1cm) {\scalebox{2}{$\mdp{A}$}};
             \node[innode] (game2ab) at (10cm, -1cm) {\scalebox{2}{$\mdp{B}$}};
             \draw[->] (rdo1aba) to ($(game1ab.north west)!0.5!(game1ab.west)$);
             \draw[->] (rdo2aba) to (game1ab);
             \draw[->] (rdo3aba) -- (8.5cm, 0.75cm) -- (9cm, -0.75cm) -- ($(game2ab.north west)!0.5!(game2ab.west)$);
             \draw[<-] (rdo1abb) -- (8.5cm, -0.75cm) -- (9cm, -1cm) -- (game2ab);
             \draw[<-] (rdo2abb) -- (8.5cm, -1cm) -- (9cm, -1.25cm) -- ($(game2ab.south west)!0.5!(game2ab.west)$);
             \draw[<-] (rdo3abb) -- (8.5cm, -1.25cm) -- (9cm, 0.75cm) -- ($(game1ab.south west)!0.5!(game1ab.west)$);
             \draw[->] ($(game1ab.north east)!0.5!(game1ab.east)$) to (rcdo1aba);
             \draw[->] ($(game2ab.north east)!0.5!(game2ab.east)$) -- (11cm, -0.75cm) -- (11.5cm, 0.75cm) -- (rcdo3aba);
             \draw[<-] ($(game2ab.south east)!0.5!(game2ab.east)$) -- (11cm, -1.25cm) -- (11.5cm, -0.75cm) -- (rcdo1abb);
             \draw[<-] (game1ab) -- (11.1cm, 1cm) -- (11.6cm, -1cm) -- (rcdo2abb);
             \draw[<-] ($(game1ab.south east)!0.5!(game1ab.east)$) -- (11cm, 0.75cm) -- (11.5cm, -1.25cm) -- (rcdo3abb);
          \end{tikzpicture}
          }\\
          \scalebox{0.7}{
   \begin{tikzpicture}[
              innode/.style={draw, rectangle, minimum size=1cm},
              innodemini/.style={draw, rectangle, minimum size=0.5cm},
              interface/.style={inner sep=0},
              innodeeve/.style={draw, circle, minimum size=0.2cm},
              innodeadam/.style={draw, rectangle, minimum size=0.2cm},
              ]
              \fill[lightgray](1.1cm, -0.6cm)--(1.1cm, 1.2cm)--(3.85cm, 1.2cm)--(3.85cm, -0.6cm)--cycle;
              \node[interface] (rdo1) at (1cm, 0cm) {};
              \node[interface] (ldo1) at (1cm, -0.25cm) {};
              \node[innode,fill=white] (game1) at (2.5cm, 0cm) {\scalebox{2}{$\mdp{A}$}};
              \draw[->] (rdo1) to (game1);
              \draw[<-] (ldo1) to ($(game1.south west)!0.5!(game1.west)$);
              \draw[-] (2.0cm, 1.05cm) arc [radius=0.4, start angle = 90, end angle=270];
              \draw[<-] (3cm, 0.25cm) arc [radius=0.4, start angle = 270, end angle=450];
              \draw[-] (2.0cm, 1.05cm) to (3cm, 1.05cm);
              \draw[->] (3cm, 0cm) to (4cm, 0cm); 
             \node[interface] (equal) at (5.5cm, 0cm) {\scalebox{3}{$=$}};
             \node[interface] (int1) at (7cm, 0cm) {};
             \node[interface] (int2) at (8cm, 0cm) {};
             \node[interface] (int3) at (7cm, -0.25cm) {};
             \node[interface] (int4) at (8cm, -0.25cm) {};
             \draw[<-] (8cm, 1.05cm) arc [radius=0.4, start angle = 90, end angle=270];
             \draw[->] (int1) to (int2);
             \draw[<-] (int3) to (int4);
             \node[interface] (seqcomp) at (8.5cm, 0cm) {\scalebox{2}{$\seqcomp$}};
             \node[interface] (int5) at (9cm, 1.05cm) {};
             \node[interface] (int6) at (9cm, 0.25cm) {};
             \node[interface] (int7) at (9cm, 0cm) {};
             \node[interface] (int8) at (9cm, -0.25cm) {};
             \node[interface] (oplus) at (10cm, 0.75cm) {\scalebox{2}{$\oplus$}};
             \node[interface] (int9) at (11cm, 1.05cm) {};
             \node[interface] (int10) at (11cm, 0.25cm) {};
             \node[interface] (int11) at (11cm, 0cm) {};
             \node[innode] (game1ab) at (10cm, 0cm) {\scalebox{2}{$\mdp{A}$}};
             \draw[->] (int5) to (int9);
             \draw[<-] (int6) to ($(game1ab.north west)!0.5!(game1ab.west)$);
             \draw[->] (int7) to (game1ab);
             \draw[<-] (int8) to ($(game1ab.south west)!0.5!(game1ab.west)$);
             \draw[<-] ($(game1ab.north east)!0.5!(game1ab.east)$) to (int10);
             \draw[->] (game1ab) to (int11);
             \node[interface] (seqcomp) at (11.5cm, 0cm) {\scalebox{2}{$\seqcomp$}};
             \draw[<-] (12cm, 0.25cm) arc [radius=0.4, start angle = 270, end angle=450];
             \draw[->] (12cm, 0cm) to (13cm, 0cm);
          \end{tikzpicture}
    }
  \end{aligned}
\end{math}  
\caption{Sequential composition $\seqcomp$, sum $\oplus$, and loops of MDPs, illustrated.}
   \label{fig:seqCompOplusIllustrated}

\vspace{-2em}
\end{figure}

\myparagraph{Compositional Description of MDPs by String Diagrams}
The calculus we use for composing MDPs is that of \emph{string diagrams}.  \cref{fig:string_diagrams_of_mdps} shows an example used in experiments.
String diagrams offer two basic composition operations, \emph{sequential composition} $\seqcomp$ and \emph{sum} $\oplus$, illustrated in \cref{fig:seqCompOplusIllustrated}.
The rearrangement of wires in $\mdp{A}\oplus\mdp{B}$ is for bundling up wires of the same direction. It is not essential.

We note that \emph{loops} in MDPs can be described using these algebraic operations, as shown in \cref{fig:seqCompOplusIllustrated}.
We extend MDPs with open ends so that they allow such composition; they are  called \emph{open MDPs}. 

The formalism of string diagrams
originates from category theory, specifically from the theory of \emph{monoidal categories} (see e.g.~\cite[Chap.~XI]{MacLane2}). Capturing the mathematical essence of the algebraic structure of arrow composition $\circ$ and tensor product $\otimes$---they correspond to $;$ and $\oplus$ in this work, respectively---monoidal categories and string diagrams have found their application in a wide variety of scientific disciplines, such as  quantum field theory~\cite{khovanov2002functor}, quantum mechanics and computation~\cite{heunen2019categories}, linguistics~\cite{PiedeleuKCS15}, signal flow diagrams~\cite{BonchiHPSZ19}, and so on.

Our reason for using string diagrams to compose MDPs is twofold. Firstly, string diagrams 
offer a rich metatheory---developed over the years together with its various applications---that we can readily exploit. Specifically, the theory covers \emph{functors}, which are (structure-preserving) homomorphisms between monoidal categories. We introduce a \emph{solution functor} $\semFunctor\colon \oMDP\to\semCat$ from a category $\oMDP$ of open MDPs to a semantic category $\mathbb{S}$ that consists of solutions. We show that the  functor $\semFunctor$ preserves two composition operations, that is, 
\begin{equation}\label{eq:introCompSol}\small
 \semFunctor(\mdp{A}\seqcomp\mdp{B})=
 \semFunctor(\mdp{A})\seqcomp\semFunctor(\mdp{B}),
\quad
 \semFunctor(\mdp{A}\oplus\mdp{B})=
 \semFunctor(\mdp{A})\oplus\semFunctor(\mdp{B}),
\end{equation}
where $;$ and $\oplus$ on the right-hand sides are \emph{semantic composition} operations on $\semCat$.
The equalities~\cref{eq:introCompSol} are nothing but \emph{compositionality}: the solution of the whole (on the left) is computed from the solutions of its parts (on the right). 

The second reason for using string diagrams is that they offer an expressive language for composing MDPs---one that enables an efficient description of a number of realistic system models---as we demonstrate with  benchmarks in \cref{sec:impl}. 

\myparagraph{Granularity of Semantics: a Challenge towards Compositionality}
Now the main technical challenge is the design of a semantic domain  $\semCat$  (it is a category in our framework). We shall call it the challenge of \emph{granularity of semantics}; it is encountered generally when one aims at compositional solutions. 
\begin{itemize}
 \item The coarsest candidate for $\semCat$ is the original semantic domain; it consists of solutions and nothing else. This coarsest candidate is not enough most of the time: when components are composed, they may interact with each other via a richer interface than mere solutions. (Consider a team of two people. Its performance is usually not the sum of each member's, since there are other affecting factors such as work style, personal character, etc.)
 \item Therefore one would need to use a finer-grained semantic domain as $\semCat$, which, however, comes with a computational cost:  in~\cref{eq:introCompSol}, one will have to carry around bigger data as intermediate solutions $\semFunctor(\mdp{A})$ and $\semFunctor(\mathcal{B})$; their semantic composition will become more costly, too.
\end{itemize}
Therefore, in choosing $\semCat$, one should find the smallest enrichment\footnote{\emph{Enrichment} here is in the natural language sense; it has nothing to do with the technical notion of \emph{enriched category}.} of the original
semantic domain that addresses all relevant interactions between components and thus enables compositional solutions. This is a theoretical  challenge.

In this work, following our recent work~\cite{Watanabe21} that pursued a compositional solution of parity games, we use category theory as guidance in tackling the above challenge.  Our goal is to obtain a solution functor $\semFunctor\colon \oMDP\to\semCat$ that preserves suitable algebraic structures
 (see~\cref{eq:introCompSol}); the specific notion of algebra of our interest is that of \emph{compact closed categories (compCC)}. 
\begin{itemize}
 \item The category $\oMDP$ organizes open MDPs as a category.
It is a compCC, and its algebraic operations are defined  as in~\cref{fig:seqCompOplusIllustrated}.
 \item For the solution functor  $\semFunctor$  to be compositional, the semantic category $\semCat$ \emph{must itself be a compCC}, that is, $\semCat$ has to be enriched so that the compCC operations  ($\seqcomp$ and $\oplus$) are well-defined.
 \item Once such a semantic domain $\semCat$ is obtained, choosing $\semFunctor$ and showing that it preserves the algebraic operations are straightforward.
\end{itemize}
 Specifically, we find that $\semCat$ must be enriched with \emph{reachability probabilities}, in addition to the desired solutions (namely expected rewards), to be a compCC. This enrichment is based on the \emph{decomposition equalities} we observe in \cref{sec:compositionalAnalyssiOfOpenMarkovChains}. 

After all, our semantic category $\semCat$ is as follows: 
1) an object is a pair of natural numbers describing an interface (how many entrances and exits);
2) an arrow is a collection of ``semantics,'' collected over all possible (memoryless) schedulers $\tau$, which records the expected reward that the scheduler $\tau$ yields when it traverses from each entrance to each exit. The last ``semantics'' is enriched so that it records the reachability probability, too, for the sake of compositionality.
\ifdraft
\color{red}
Let's \emph{depict} some arrows of each semantic category in \cref{fig:catsFunctors}.
\color{black}
\fi

\myparagraph{Related Work}
Compositional model checking is studied e.g.\ in~\cite{DBLP:conf/csl/TsukadaO14,ClarkeLM89,Watanabe21}. Besides,
probabilistic model checking is an actively studied topic; see~\cite[Chap.~10]{BaierKatoen08} for a comprehensive account. 
We shall make a detailed comparison with the works~\cite{JungesS22,KwiatkowskaNPQ13} that study compositional probabilistic model checking.

The work~\cite{KwiatkowskaNPQ13} introduces an assume-guarantee reasoning framework for parallel composition $\parallel$, as we already discussed. Parallel composition is out of our current scope; in fact, we believe that compositionality with respect to $\parallel$ requires a much bigger enrichment of a semantic domain $\semCat$ than mere reachability probabilities as in our work. The work~\cite{KwiatkowskaNPQ13} is remarkable in that its solution to this granularity problem---namely by assume-guarantee reasoning---is practically sensible (domain experts often have ideas about what contract to impose) and comes with automata-theoretic automation. That said, such contracts are not always automatically synthesized in~\cite{KwiatkowskaNPQ13}, while our algorithm is fully automatic.

The work~\cite{JungesS22} is probably the closest to ours in the type of composition (sequential rather than parallel) and automation. However, the technical bases of the two works are quite different: theirs is the theory of \emph{parametric MDPs}~\cite{QuatmannD0JK16}, which is why their emphasis is on parametrized components and interval solutions; ours is monoidal categories and some decomposition equalities (\cref{sec:compositionalAnalyssiOfOpenMarkovChains}). 

We note that the work~\cite{JungesS22} and ours are not strictly comparable. On the one hand, we do not need a crucial assumption in~\cite{JungesS22}, namely that a locally optimal scheduler in each component is part of a globally optimal scheduler. The assumption limits the applicability of \cite{JungesS22}---it practically forces each component to have only one exit. The assumption does not hold in our benchmarks Patrol and Wholesale  (see \cref{sec:impl}). Our algorithm does not need the assumption since it collects the semantics of all relevant memoryless schedulers.

On the other hand, unlike~\cite{JungesS22}, our  algorithm is not parametric, so it cannot exploit the similarity of components if they only differ in parameter values. Note that the target problems are different, too (interval~\cite{JungesS22} vs.\ exact here).

\begin{auxproof}
 \myparagraph{Contributions}
 We present a compositional model checking algorithm for MDPs. MDPs are composed in string diagrams, and the algorithm exactly computes optimal expected rewards. Our theoretical development of the algorithm is supported by category theory, allowing automation of some constructions, while so-called decomposition equalities for expected rewards act as a key enabler. The experimeal evaluation of its implementaion demonstrates performance advantages of the algorithm.
\end{auxproof}

\myparagraph{Notations}
For natural numbers $m$ and $n$, we let $[m, n]\coloneqq\{m,m+1,\dotsc, n-1, n\}$; as a special case, we let $\nset{m}\coloneqq \{1,2,\dotsc, m\}$ (we let $\nset{0}=\emptyset$ by convention).
 The disjoint union of two sets $X,Y$ is denoted by $X+Y$. 

\section{String Diagrams of MDPs}
\label{sec:compositionalMDPs}

We introduce our calculus for composing MDPs, namely  \emph{string diagrams of MDPs}. Our formal definition is via their \emph{unidirectional} and \emph{Markov chain (MC)} restrictions. This apparent detour simplifies the theoretical development, allowing us to exploit the existing categorical infrastructure on (monoidal) categories.

\subsection{Outline }\label{subsec:stringDiagramOfMDPsOutline}

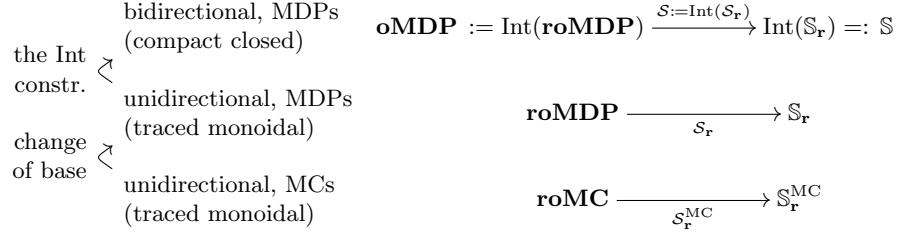
\begin{figure}[tbp]\centering
\begin{math}
  \hspace{-0.6cm}
  \vcenter{\xymatrix@C+1.6em@R=0.5em{
  **[r]{\begin{array}[c]{l}
	    \text{bidirectional, MDPs}\qquad\qquad\qquad\qquad\qquad\\
	\text{(compact closed)}\end{array}
   }
  &
  {\mathllap{\oMDP\, :=\;} 
  \Int(\roMDP)}
     \ar[r]^-{\semFunctor\coloneqq\Int(\semFunctorr)}
   &
  {\Int(\semCatr)\mathrlap{\;=:\,\semCat }}
  \\
  **[r]{\footnotesize\begin{array}[c]{l}\footnotesize
   \text{unidirectional, MDPs}\\
   \footnotesize
   \text{(traced monoidal)}
\end{array}
   }
   \ar@(lu,ld)[u]^{
    \begin{array}{r}
     \text{\small the Int}\\
     \text{\small constr.}
    \end{array}}
  &
  {\roMDP}
      \ar[r]_-{\semFunctorr }
  &
  {\semCatr}
  \\
  **[r]{\footnotesize\begin{array}[c]{l}\footnotesize
   \text{unidirectional, MCs}\\
   \footnotesize
   \text{(traced monoidal)}
\end{array}
   }
   \ar@(lu,ld)[u]^{
    \begin{array}{r}
     \text{\small change}\\
     \text{\small of base}
    \end{array}}
  &
  {\roMC}
      \ar[r]_-{\semFunctorrmc}
  &
  {\semCatrmc}
 }}
\end{math}
\caption{Categories of MDPs/MCs, semantic categories, and solution functors.}
\label{fig:catsFunctors}
\end{figure}
We first make an overview of our technical development.
Although we use some categorical terminologies, prior knowledge of them is not needed in this outline.

\ifdraft
 \begin{figure}[tbp]\centering
 \begin{math}
  \vcenter{\xymatrix@C+1.6em@R=0.5em{
  **[r]{\begin{array}[c]{l}
	    \text{bidirectional, MDPs}\qquad\qquad\qquad\qquad\qquad\\
	\text{(compact closed)}\end{array}
   }
  &
  {\mathllap{\mathbf{oMDP}\, :=\;} 
  \Int(\roMDP)}
     \ar[rr]^-{\semFunctor\coloneqq\Int(\semFunctorr)}
   &&
  {\Int(\semCatr)\mathrlap{\;=:\,\semCat }}
  \\
  **[r]{\footnotesize\begin{array}[c]{l}\footnotesize
   \text{unidirectional, MDPs}\\
   \footnotesize
   \text{(traced monoidal)}
 \end{array}
   }
   \ar@(lu,ld)[u]^{
    \begin{array}{r}
     \text{\small the Int}\\
     \text{\small constr.}
    \end{array}}
  &
  {\roMDP}
      \ar@{ >->}[r]
      \ar@/^2ex/[rr]^{\semFunctorr}
  &
  {\mathrm{CoB}(\roMC)}
      \ar[r]_-{\mathrm{CoB}(\semFunctorrmc)}
  &
  {\mathrm{CoB}(\semCatrmc)\mathrlap{\;=:\,\semCatr }}
  \\
  **[r]{\footnotesize\begin{array}[c]{l}\footnotesize
   \text{unidirectional, MCs}\\
   \footnotesize
   \text{(traced monoidal)}
 \end{array}
   }
   \ar@(lu,ld)[u]^{
    \begin{array}{r}
     \text{\small change}\\
     \text{\small of base}
    \end{array}}
  &&
  {\roMC}
      \ar[r]_-{\semFunctorrmc}
  &
  {\semCatrmc}
 }}
 \end{math}
 \caption{Categories of MDPs/MCs, semantic categories, and solution functors.}
 \label{fig:catsFunctors}
 \end{figure}
\fi

\cref{fig:catsFunctors} is an overview of relevant categories and functors.
The verification targets---\emph{open MDPs}---are  arrows in the compact closed category (compCC) $\oMDP$. The operations $\seqcomp, \oplus$ of compCCs compose MDPs, as shown in~\cref{fig:seqCompOplusIllustrated}. 
Our semantic category is denoted by $\semCat$, and our goal is to define a solution functor $\oMDP\to \semCat$ that is compositional. Mathematically, such a functor with  the desired compositionality (cf.\ \cref{eq:introCompSol}) is called a \emph{compact closed functor}. 

Since its direct definition is tedious, our strategy is to obtain it from a unidirectional \emph{rightward} framework $\semFunctorr\colon\roMDP\to \semCatr$, which canonically induces the desired bidirectional framework via the celebrated \emph{$\Int$ construction}~\cite{joyal1996}. In particular, the category $\oMDP$ is defined by $\oMDP=\Int(\roMDP)$; so are the semantic category and the solution functor ($\semCat=\Int(\semCatr), \semFunctor=\Int(\semFunctorr)$). 

Going this way, 
a complication that one would encounter in a direct definition of $\oMDP$ (namely potential loops of transitions) is nicely taken care of by the $\Int$ construction. Another benefit is that some natural equational axioms in $\oMDP$---such as the associativity of sequential composition $\seqcomp$---follow automatically from those in $\roMDP$, which are much easier to verify.

Mathematically, 
the unidirectional framework $\semFunctorr\colon\roMDP\to \semCatr$ consists of \emph{traced symmetric monoidal categories (TSMCs)} and 
\emph{traced symmetric monoidal functors}; these are ``algebras'' of unidirectional graphs. The $\Int$ construction turns TSMCs into compCCs, which are ``algebras'' of bidirectional graphs.

Yet another  restriction is given by  \emph{(rightward open) Markov chains (MCs)}. See the bottom row of \cref{fig:catsFunctors}.
This MDP-to-MC restriction greatly simplifies our semantic development, freeing us from the bookkeeping of  different schedulers. In fact, we can introduce (optimal memoryless) schedulers systematically by the categorical construction called \emph{change of base}~\cite{Eilenberg65,cruttwell2008normed}; this way we obtain the semantic category $\semCatr$ from $\semCatrmc$. %

\subsection{Open MDPs}
\begin{wrapfigure}[4]{r}{0pt}
  \begin{minipage}[b]{13em}
    \begin{equation}\label{eq:figOpenMDP}
      \begin{aligned}
        \begin{tikzpicture}[
              innode/.style={draw, rectangle, minimum size=1cm},
              innodemini/.style={draw, rectangle, minimum size=0.5cm},
              interface/.style={inner sep=0},
              innodepos/.style={draw, circle, minimum size=0.2cm},
              ]
              \node[interface] (rdo0) at (0cm, 0.25cm) {};
              \node[interface] (rdo1) at (0cm, 0cm) {};
              \node[interface] (ldo1) at (0cm, -0.25cm) {};
              \node[innode, fill=white] (pos1) at (1cm, 0cm) {$\mdp{A}$};
              \node[interface] (rcdo1) at (2cm, 0.38cm) {};
              \node[interface] (lcdo1) at (2cm, 0.125cm) {};
              \node[interface] (lcdo2) at (2cm, -0.38cm) {};
              \node[interface] (lcdo3) at (2cm, -0.125cm) {};
              \draw[->] (rdo0) to ($(pos1.north west)!0.5!(pos1.west)$);
              \draw[->] (rdo1) to (pos1);
              \draw[<-] (ldo1) to ($(pos1.south west)!0.5!(pos1.west)$);
              \draw[->] ($(pos1.north east)!0.25!(pos1.east)$) to (rcdo1); 
              \draw[<-] ($(pos1.north east)!0.75!(pos1.east)$) to (lcdo1); 
              \draw[<-] ($(pos1.south east)!0.75!(pos1.east)$) to (lcdo3);
              \draw[<-] ($(pos1.south east)!0.25!(pos1.east)$) to (lcdo2); 
              \node[interface] (kacco1) at (-0.2cm, 0.1cm) {\scalebox{1.3}{$\{$}};
              \node[interface] (kacco2) at (2.2cm, -0.125cm) {\scalebox{1.8}{$\}$}};
              \node[interface] (rarity1) at (-0.6cm, 0.1cm) {$\nfr{m}$};
              \node[interface] (larity1) at (-0.6cm, -0.25cm) {$\nfl{m}$};
              \node[interface] (rarity2) at (2.6cm, 0.38cm) {$\nfr{n}$};
              \node[interface] (larity2) at (2.6cm, -0.125cm) {$\nfl{n}$};
        \end{tikzpicture}
      \end{aligned}
    \end{equation}
    \end{minipage}
  \end{wrapfigure}
We first introduce \emph{open MDPs}; they have open ends via which they compose. They come with a notion of \emph{arity}---the numbers of open ends on their left and right, distinguishing leftward and rightward ones. For example, the one on the right is from $(2,1)$ to $(1,3)$.

\begin{definition}[open MDP (oMDP)]
\label{def:openMDP}
Let $A$ be a non-empty finite set, whose elements are called \emph{actions}. An \emph{open MDP} 
 $\mdp{A}$ (\emph{over} the action set $A$) is the tuple $(\bfn{m},\bfn{n},Q, A,\allowbreak E,P,R)$ of the following data. We say that it is \emph{from $\bfn{m}$ to $\bfn{n}$}. 
\begin{enumerate}
\item $\bfn{m} = (\nfr{m},\ \nfl{m})$ and $\bfn{n} = (\nfr{n},\ \nfl{n})$ are pairs of natural numbers; they are called the \emph{left-arity} and the \emph{right-arity}, respectively. Moreover (see~\cref{eq:figOpenMDP}),
 elements of $\nset{\nfr{m}+\nfl{n}}$ are called  \emph{entrances}, and those of $\nset{\nfr{n}+\nfl{m}}$ are called \emph{exits}. 

\item $Q$ is a finite set of
 \emph{positions}.
\item $E:\nset{\nfr{m}+\nfl{n}} \to Q + \nset{\nfr{n}+\nfl{m}}$ is an \emph{entry function}, which maps each entrance
to either a position (in $Q$) or an exit (in $\nset{\nfr{n}+\nfl{m}}$).
\item $P:Q\times A\times (Q + \nset{\nfr{n}+\nfl{m}}) \rightarrow \Realp$
        determines \emph{transition probabilities}, where we require $\sum_{s'\in Q+\nset{\nfr{n}+\nfl{m}}}P(s, a, s') \in \{0, 1\}$ for each $s\in Q$ and $a\in A$.
\item $R$ is a \emph{reward function} $R:Q\rightarrow \Realp$.
\item  We impose the following ``unique access to each exit'' condition. 
 Let $\exits: (\nset{\nfr{m}+\nfl{n}} + Q) \to \mathcal{P}(\nset{\nfr{n}+\nfl{m}})$ be the \emph{exit function} that collects all immediately reachable exits, that is, 1) for each $s\in Q$, $\exits(s) = \{ t \in \nset{\nfr{n}+\nfl{m}}\,|\,\exists a \in A. P(s,a,t) > 0 \}$, and 2) for each entrance $s \in \nset{\nfr{m}+\nfl{n}}$,  $\exits(s) = \{E(s)\}$ if $E(s)$ is an exit and $\exits(s) = \emptyset$ otherwise.
\begin{itemize}
 \item For all $s,s' \in \nset{\nfr{m}+\nfl{n}} + Q$, if $\exits(s) \cap \exits(s') \neq \emptyset$, then $s = s'$.
 \item We further require that each exit is reached from an identical position by at most one action. That is, 
for each exit $t \in \nset{\nfr{n}+\nfl{m}}$,  $s \in Q$, and $a,b \in A$, if both $P(s,a,t) > 0$ and $P(s,b,t) > 0$, then $a = b$.

\end{itemize}

\end{enumerate}

\end{definition}

Note that the unique access to each exit condition is for technical convenience; this can be easily enforced by adding an extra ``access'' position to an exit. 

We define the semantics of open MDPs, which is essentially the  standard semantics of MDPs given by expected cumulative rewards. In this paper, it suffices to consider memoryless schedulers (see~\cref{rem:suf_memoryless}).

\begin{definition}[path and scheduler]\label{def:path}
Let $\mdp{A} = (\bfn{m},\bfn{n},Q, A, E,P,R)$ be an open MDP.
A (finite) \emph{path} $\mypath{i}{j}$ in $\mdp{A}$ from an entrance $i\in \nset{\nfr{m}+\nfl{n}}$ 
 to an exit $j\in \nset{\nfr{n}+\nfl{m}}$ is a finite sequence $i, s_1,\dotsc, s_n, j$ such that $E(i) = s_1$ and for all $k \in \nset{n}$, $s_k \in Q$. For each $k\in \nset{n}$, $\mypath{i}{j}_k$ denotes $s_k$, and $\mypath{i}{j}_{n+1}$ denotes $j$. The set of all paths in $\mdp{A}$ from $i$ to $j$ is 
denoted by  $\Pathstra{\mdp{A}}{\tau}{i}{j}$.

A \emph{(memoryless) scheduler} $\tau$ of $\mdp{A}$ is a function $\tau:Q\rightarrow  A$.
\end{definition}

\begin{remark}
\label{rem:suf_memoryless}
It is well-known (as hinted in~\cite{Baier0KW17}) that we can restrict to memoryless schedulers for optimal expected rewards, \emph{assuming that} the MDP in question is almost surely terminating under any scheduler $(\dagger)$. We require the assumption $(\dagger)$ in our compositional framework, too, and it is true in all benchmarks in this paper.  The assumption $(\dagger)$ must be checked only for the top-level (composed) MDP; $(\dagger)$ for its components can then be deduced.

\end{remark}

\begin{definition}[probability and reward of a path]\label{def:probRewPath}
Let $\mdp{A} = (\bfn{m},\bfn{n},Q, A, E,\allowbreak P,R)$ be an open MDP,
$\tau:Q\rightarrow  A$ be a scheduler of $\mdp{A}$,
and  $\mypath{i}{j}$ be a path in $\mdp{A}$.
The \emph{probability} $\ProPath{\mdp{A}}{\tau}(\mypath{i}{j} )$ of 
 $\mypath{i}{j}$ under $\tau$ is
\begin{math}
    \ProPath{\mdp{A}}{\tau}(\mypath{i}{j} )\defeq \textstyle
    \prod_{k=1}^n P\bigl(\,\mypath{i}{j}_k,\tau(\mypath{i}{j}_k),\mypath{i}{j}_{k+1}\,\bigr).
\end{math}
 The \emph{reward} $\PReward{\mdp{A}}(\mypath{i}{j})$ along the path $\mypath{i}{j}$ is the sum of the position rewards, that is,  $ \PReward{\mdp{A}}(\mypath{i}{j}) \defeq \sum_{k\in \nset{n}} R(\mypath{i}{j}_k)$.

\end{definition}

\vspace{.3em}
Our target problem on open MDPs is to compute
 the \emph{expected cumulative reward}  collected in a passage from a specified entrance $i$ to a specified exit $j$.
This is defined below, together with reachability probability, in the usual manner.

 \begin{definition}[reachability probability and expected (cumulative) reward of open MDPs]\label{def:openMDPReachProbExpCumRew} 
 Let $\mdp{A}$ be an open MDP and $\tau$ be a scheduler, as in \cref{def:path}. 
 Let $i$ be an entrance and $j$ be an exit.

 The \emph{reachability probability} $\Reachstr{\mdp{A}}{\tau}{i}{j}$ from $i$ to $j$, in $\mdp{A}$ under $\tau$, is defined by $\Reachstr{\mdp{A}}{\tau}{i}{j} \defeq \sum_{\mypath{i}{j}\in  \Pathstra{\mdp{A}}{\tau}{i}{j}} \ProPath{\mdp{A}}{\tau}(\mypath{i}{j})$.

 The \emph{expected (cumulative) reward $\ETRstr{\mdp{A}}{\tau}{i}{j}$ 
 from $i$ to $j$}, in $\mdp{A}$ under $\tau$, is defined by $\ETRstr{\mdp{A}}{\tau}{i}{j}\defeq \sum_{\mypath{i}{j}\in  \Pathstra{\mdp{A}}{\tau}{i}{j}}\ProPath{\mdp{A}}{\tau}(\mypath{i}{j})\cdot \PReward{\mdp{A}}(\mypath{i}{j})$. Note that the infinite sum here always converges to a finite value; this is because there are only finitely many positions in $\mdp{A}$. See e.g.~\cite{BaierKatoen08}.
 \end{definition}

\begin{remark}\label{rem:notAlmostSureTerm}
In standard definitions such as \cref{def:openMDPReachProbExpCumRew}, it is common to either 1)  assume $\Reachstr{\mdp{A}}{\tau}{i}{j}=1$ for technical convenience~\cite{JungesS22}, or 2) allow $\Reachstr{\mdp{A}}{\tau}{i}{j}<1$, but in that case define $\ETRstr{\mdp{A}}{\tau}{i}{j}\defeq \infty$ \cite{BaierKatoen08}.
These definitions are not suited for our purpose (and  for compositional model checking in general), since we take into account multiple exits, to each of which the reachability probability is typically $<1$, 
and we need non-$\infty$ expected rewards over those exits 
for compositionality.
 Note that our definition of expected reward is not conditional (unlike~\cite[Rem.~10.74]{BaierKatoen08}): when the reachability probability from $i$ to $j$ is small, it makes the expected reward small as well.  Our notion of expected reward can be thought of as a ``weighted sum'' of rewards.

\end{remark}

\subsection{
Rightward Open MDPs and Traced Monoidal String Diagrams
}\label{subsec:roMDPsAndTSMCStringDiagrams}
Following the outline (\cref{subsec:stringDiagramOfMDPsOutline}),  in this section we focus on (unidirectional) \emph{rightward} open MDPs and introduce the ``algebra'' $\roMDP$ of them. The operations  $\seqcomp, \oplus, \tr$ of \emph{traced symmetric monoidal categories (TSMCs)} compose rightward open MDPs in string diagrams.

\begin{definition}[rightward open MDP (roMDP)]\label{def:rightwardOpenMDP}
  An open MDP $\mdp{A} = (\bfn{m}, \bfn{n}, Q, A, E, P, R)$ is \emph{rightward} if all its entrances are on the left and all its exits are on the right, that is, $\bfn{m} = (\nfr{m}, \nfl{0})$
    and $\bfn{n}= (\nfr{n}, \nfl{0})$
   for some \(\nfr{m}\) and \(\nfr{n}\).  We write $\mdp{A} = (\nfr{m}, \nfr{n}, Q, A, E, P, R)$, dropping $0$ from the arities.

We say that a rightward open MDP $\mdp{A}$ is \emph{from $m$ to $n$}, writing $\mdp{A}:m\rightarrow n$, if it is from $(m,0)$ to $(n,0)$ as an open MDP.
\end{definition}

\iffull
 We use a natural equivalence relation by \emph{roMDP isomorphism} (\cref{def:isomorphismRoMDPs}) so that roMDPs satisfy TSMC axioms
 given in \cref{subsec:eqAxiomOfTSMCs}. These axioms do not hold up-to equality of MDPs because, e.g.,  two sets $Q_{1}+ (Q_{2}+ Q_{3})$ and  $(Q_{1}+ Q_{2})+ Q_{3}$ are not equal but only isomorphic.
 \else
 We use an equivalence relation by \emph{roMDP isomorphism} so that roMDPs satisfy TSMC axioms
 given in~\cref{subsec:eqAxiomOfTSMCs}. See~\cite[Appendix A]{Watanabe23long} for details. 
 \fi

We move on to introduce algebraic operations for composing rightward open MDPs. Two of them, namely \emph{sequential composition} $\seqcomp$ and \emph{sum} $\oplus$, look like~\cref{fig:seqCompOplusIllustrated} except that all wires are rightward. The other major operation is the \emph{trace operator} $\tr$ that realizes (unidirectional) loops, as illustrated in \cref{fig:trace_of_rmdp}. 

\begin{figure}[t]
  \centering
\begin{math}
 \begin{array}{l}
 \bigl(\mdp{A}:l+m\rightarrow l + n\bigr)
 \\
 \longmapsto
 \bigl(\,\trmdp{l}{m}{n}(\mdp{A}):m\rightarrow n\,\bigr),
 \end{array}
\end{math}
\quad\text{ as in}
    \begin{subfigure}[c]{0.55\columnwidth}
    \centering
    \scalebox{0.85}{
    \begin{tikzpicture}
    \node[draw, rectangle, minimum size=1cm] (mdpA) at (1cm, 0cm) {\scalebox{1.5}{$\mdp{A}$}};
    \draw[->] (0cm, 0.25cm) to ($(mdpA.north west)!0.5!(mdpA.west)$); 
    \draw[->] (0cm, -0.25cm) to ($(mdpA.south west)!0.5!(mdpA.west)$); 
    \draw[->] ($(mdpA.south east)!0.5!(mdpA.east)$) to (2cm, -0.25cm);
    \draw[->] ($(mdpA.north east)!0.5!(mdpA.east)$) to (2cm, 0.25cm);
    \node[] (doml) at (0.2cm, 0.4cm) {$l$};
    \node[] (domm) at (0.2cm, -0.1cm) {$m$};
    \node[] (codoml) at (1.8cm, 0.4cm) {$l$};
    \node[] (codomn) at (1.8cm, -0.1cm) {$n$};
    \node[] (translate) at (3cm, 0cm) {\large $\longmapsto$}; 
    \node[draw, rectangle, minimum size=1cm] (mdpA2) at (5cm, 0cm) {\scalebox{1.5}{$\mdp{A}$}};
    \draw[<-] ($(mdpA2.north west)!0.5!(mdpA2.west)$) to [out=165,in=195] (4.5cm,0.7cm);
    \draw[-] (4.5cm, 0.7cm) to (5.5cm, 0.7cm);
    \draw ($(mdpA2.north east)!0.5!(mdpA2.east)$) to [out=15,in=345] (5.5cm,0.7cm);
    \draw[->] (4cm, -0.25cm) to  ($(mdpA2.south west)!0.5!(mdpA2.west)$);
    \draw[->] ($(mdpA2.south east)!0.5!(mdpA2.east)$) to (6cm, -0.25cm);
    \node[] (domm2) at (4.2cm, -0.1cm) {$m$};
    \node[] (codomn2) at (5.8cm, -0.1cm) {$n$};
    \end{tikzpicture}   
    }
\end{subfigure}
\caption{The trace operator.}
\label{fig:trace_of_rmdp}
\end{figure}

\begin{definition}[sequential composition $\seqcomp$ of roMDPs]
\label{def:seqCompRightwardOpenMDPs}
  Let $\mdp{A}\colon m\to k$ and $\mdp{B}\colon k\to n$ be rightward open MDPs
  with the same action  set $A$ and with matching arities.
 Their \emph{sequential composition} 
$\mdp{A} \seqcomp \mdp{B}\colon m\to n$ 
is given by
  \begin{math}
      \mdp{A} \seqcomp \mdp{B} \defeq
     \bigl(m,  n, Q^{\mdp{A}}+Q^{\mdp{B}},
              A,
               E^{\mdp{A} \seqcomp \mdp{B}},
                P^{\mdp{A} \seqcomp \mdp{B}}, [R^{\mdp{A}}, R^{\mdp{B}}]\bigr)
     \end{math}, where
\begin{itemize}
 \item $E^{\mdp{A} \seqcomp \mdp{B}}(i)\defeq E^{\mdp{A}}(i)$ if $E^{\mdp{A}}(i) \in Q^{\mdp{A}}$, and 
$E^{\mdp{A} \seqcomp \mdp{B}}(i)\defeq E^{\mdp{B}}(E^{\mdp{A}}(i))$ otherwise (if the $\mdp{A}$-entrance $i$ goes to an $\mdp{A}$-exit
which is identified with a $\mdp{B}$-entrance);
 \item the transition probabilities are defined in the following natural manner
     \begin{align*}\small
\begin{aligned}
          P^{\mdp{A} \seqcomp \mdp{B}}(s^{\mdp{A}}, a, s') &\defeq \begin{cases}
            P^{\mdp{A}}(s^{\mdp{A}}, a, s') &\text{if $s' \in Q^{\mdp{A}}$,}\\
            \sum_{i \in \nset{k}} P^{\mdp{A}}(s^{\mdp{A}}, a, i) \cdot \delta_{E^{\mdp{B}}(i) = s'} &\text{otherwise (i.e.\ $s'\in Q^{\mdp{B}}+\nset{n}$),}
         \end{cases}\\
        P^{\mdp{A} \seqcomp \mdp{B}}(s^{\mdp{B}}, a, s') &\defeq
         \begin{cases}
            P^{\mdp{B}}(s^{\mdp{B}}, a, s') &\text{if $s' \in Q^{\mdp{B}} + \nset{n}$,}\\
            0 &\text{otherwise,}
         \end{cases}
\end{aligned}     
\end{align*}
       where $\delta$ is a characteristic function (returning $1$ if the condition is true);

 \item and $[R^{\mdp{A}}, R^{\mdp{B}}]\colon Q^{\mdp{A}}+Q^{\mdp{B}}\to \Realp$ combines $R^{\mdp{A}}, R^{\mdp{B}}$ by case distinction.
\end{itemize}
      
\end{definition}

\begin{auxproof}
 Next, we define the sum $\mdp{A}\oplus\mdp{B}$ of open MDPs, which express the disjoint union of $\mdp{A}$ and $\mdp{B}$ with respecting the order of open ends.
\end{auxproof}

Defining sum $\oplus$ of roMDPs is straightforward, following \cref{fig:seqCompOplusIllustrated}.
\iffull
See \cref{def:sumRightwardOpenMDPs}. 
\else 
See~\cite[Appendix A]{Watanabe23long} for details. 
\fi

The trace operator $\tr$ is primitive in the TSMC $\roMDP$; it is crucial in defining  bidirectional sequential composition shown in~\cref{fig:seqCompOplusIllustrated} (cf.\ \cref{def:seqCompOpenMDPs}). 

\begin{definition}[the trace operator \(\trmdp{l}{m}{n}\) over roMDPs]
\label{def:trace_tmdp}
Let $\mdp{A}:l+m\rightarrow l+n$ be a rightward open MDP.
The \emph{trace} $\trmdp{l}{m}{n}(\mdp{A}):m\rightarrow n$ of $\mdp{A}$ with respect to  $l$ is the roMDP  $\trmdp{l}{m}{n}(\mdp{A})\defeq  
  \bigl(m,n, Q^{\mdp{A}},A,E,P,R^{\mdp{A}}\bigr)$ (cf.\ \cref{fig:trace_of_rmdp}), where
  \begin{itemize}
    \item The entry function $E$ is defined naturally, using a sequence $i_{0},\dotsc, i_{k-1}$ of intermediate open ends (in $\nset{l}$) until reaching a destination $i_{k}$. 

Precisely, we let $i_{0}\defeq i+l$ and $i_{j}=E^{\mdp{A}}(i_{j-1})$ for each $j$. We let $k$ to be the first index at which $i_{k}$ comes out of the loop, that is, 1) $i_{j}\in \nset{l}$ for each $j\in \nset{k-1}$, and 2)  $i_{k}\in [l+1,l+n]+Q^{\mdp{A}}$. 
Then we define $E(i)$ by the following:  $E(i)\defeq i_{k}-l$ if $i_{k}\in [l+1,l+n]$; and $E(i)\defeq i_{k}$ otherwise. 

    \item  The transition probabilities $P$ are defined as follows.  We let $\precedent(t)$ be the set of open ends in $\nset{l}$---those which are in the loop---that eventually enter $\mdp{A}$ at $t\in [l+1, n]+Q^{\mdp{A}}$. Precisely,
$\precedent(t) \defeq \{i\in\nset{l}\mid \exists i_{0},\dotsc, i_{k}.\, i_{0}=i,   i_{j+1} = E(i_j) \text{(for each $j$)}, i_{k}=t, i_{0},\dotsc, i_{k-1}\in [1,l], i_{k}\in [l+1, n]+Q^{\mdp{A}}
\}$. Using this, 
\\[-\baselineskip]
    \begin{align*}\small
\begin{aligned}
           P(q, a, q') &\defeq \begin{cases}
    P^{\mdp{A}}(q, a, q'+l)+ \sum_{i\in\precedent(q'+l) }P^{\mdp{A}}(q, a, i)  \text{ if }q'\in \nset{n}, \\
    P^{\mdp{A}}(q, a, q') + \sum_{i\in\precedent(q') }P^{\mdp{A}}(q, a, i)\text{ otherwise, i.e.\ if } q'\in Q^{\mdp{A}}.
  \end{cases} 
\end{aligned}    
\end{align*}
  \end{itemize}
Here $Q^{\mdp{A}}$ and $\nset{l}$ are assumed to be disjoint without loss of generality. 

\end{definition}

\vspace{-.2em}
\noindent
\begin{minipage}{\textwidth}
\begin{remark}
 In  string diagrams, it is common to annotate a wire with its type, such as
 $\overset{n}{\longrightarrow}$
 for $\id_{n}\colon n\to n$.  It is also common to separate a wire for a sum type into wires of its component types, such as below on the left. Therefore the two diagrams below on the right designate the same mathematical entity. Note that, on its right-hand side, the type annotation $1$ to each wire is omitted.
  \begin{displaymath}\footnotesize
 \begin{array}{ll}
\vcenter{\hbox{\scalebox{.8}{
    \begin{tikzpicture}
        \draw[->] (0cm, 0cm) to (2cm, 0cm);
        \node[] (dom) at (1cm, 0.2cm) {${\Huge m+n}$};
        \node[] (eq) at (2.5cm, 0cm) {$=$};
        \draw[->] (3cm, 0.3cm) to (5cm, 0.3cm);
        \node[] (dom2) at (4cm, 0.5cm) {$m$};
        \draw[->] (3cm, -0.3cm) to (5cm, -0.3cm);
        \node[] (dom3) at (4cm, -0.1cm) {$n$};
    \end{tikzpicture}   
    }}}
    \hspace{30pt}
    &\vcenter{\hbox{\scalebox{.8}{
    \begin{tikzpicture}
    \node[draw, rectangle, minimum size=1cm,fill=white] (mdpA) at (1cm, 0cm) {\scalebox{1.5}{$\mdp{A}$}};
    \draw[->] (0cm, 0cm) to (mdpA.west);
    \draw[->] (mdpA.east) to (2cm, 0cm);
    \node[] (dom) at (0.2cm, 0.2cm) {$3$};
    \node[] (dom) at (1.8cm, 0.2cm) {$2$};
    \node[] (eq) at (2.5cm, 0cm) {$=$};
    \node[draw, rectangle, minimum size=1cm] (mdpA2) at (4cm, 0cm) {\scalebox{1.5}{$\mdp{A}$}};
     \draw[->] (3cm, 0cm) to (mdpA2.west);
     \draw[->] (3cm, 0.25cm) to ($(mdpA2.north west)!0.5!(mdpA2.west)$);
     \draw[->] (3cm, -0.25cm) to ($(mdpA2.south west)!0.5!(mdpA2.west)$);
     \draw[->] ($(mdpA2.north east)!0.5!(mdpA2.east)$) to (5cm, 0.25cm);
    \draw[->] ($(mdpA2.south east)!0.5!(mdpA2.east)$) to (5cm, -0.25cm);
    \end{tikzpicture}   
    }}}
 \end{array} 
\end{displaymath}

\end{remark}
\end{minipage}

\subsection{TSMC Equations between roMDPs}\label{subsec:eqAxiomOfTSMCs}
Here we show that the three operations $\seqcomp, \oplus, \tr$ on roMDPs satisfy the equational axioms of TSMCs~\cite{joyal1996},  shown in \cref{fig:TSMCAxioms,fig:TSMCAxiomsInStringDiagrams}. These equational axioms are not directly needed for compositional model checking.  We nevertheless study them because 1) they validate some natural bookkeeping equivalences of roMDPs needed for their efficient handling, and 2) they act as a sanity check of  the mathematical authenticity of our compositional framework. For example, the handling of open ends is subtle in~\cref{subsec:roMDPsAndTSMCStringDiagrams}---e.g.\ whether they should be positions or not---and the TSMC equational axioms led us to our current definitions.

\begin{figure}[t]
 \begin{displaymath}\scriptsize
 \begin{array}{lllll}
  \text{($;$-Unit)}\quad &\ID_{m}\seqcomp\mdp{A} = \mdp{A} = \mdp{A}\seqcomp\ID_{n} 
&\quad&
 \text{(Vanishing1)}
 \quad
 & \trmdp{0}{m}{m}(\ID_m) = \ID_m
\\
\text{($;$-Assoc)}\quad& \mdp{A}\seqcomp(\mdp{B} \seqcomp\mdp{C}) = 
     (\mdp{A}\seqcomp\mdp{B}) \seqcomp\mdp{C}
 &\quad&
 \text{(Vanishing2)}
 \quad
 & \text{(see below)}
\\
\text{($\oplus$-Assoc)}\quad&
   (\mdp{A}\oplus\mdp{B})\oplus\mdp{C}
   = \mdp{A}\oplus(\mdp{B}\oplus \mdp{C})
 &\quad&
 \text{(Superposing)}
 \quad
 & \text{(see below)}
\\
 \text{(Bifunc1)}\quad
 &
   \ID_{m}\oplus\ID_{n}
   = \ID_{m+n}
 &\quad&
\text{(Yanking)}
 \quad
 & \trmdp{m}{m}{m}(\SYM_{m, m}) = \ID_m
\\
  \text{(Bifunc2)}\quad
 &
   (\mdp{A}\oplus\mdp{B})\seqcomp
   (\mdp{C}\oplus\mdp{D})
   =
   (\mdp{A}\seqcomp\mdp{C})\oplus
   (\mdp{B}\seqcomp\mdp{D})
 &\quad&
 \text{(Naturality1)}
 \quad
 & \text{(see below)}
\\
&&\quad&
 \text{(Naturality2)}
 \quad
 & \trmdp{l}{m}{n}(\mdp{A} \seqcomp (\ID_l\oplus \mdp{B})) 
\\
  \text{(Swap1)}\quad
 &
   \SYM_{m,  0} = \ID_{m}
 &\quad&
 \quad
 &
   = \trmdp{l}{m}{k}(\mdp{A}) \seqcomp \mdp{B}
\\
 \text{(Swap2)}\quad
  &
    \SYM_{l, m+n}
    =
    (\SYM_{l, m}\oplus\ID_n)\seqcomp
    (\ID_m\oplus\SYM_{l, n})
 &\quad&
 \text{(Dinaturality)}
 \quad
 & \text{(see below)}
\\
  \text{(Swap3)}\quad
 &
   \SYM_{m,  n}\seqcomp \SYM_{n, m} = \ID_{m+n}
 &\quad&
 &
 \end{array} 
 \end{displaymath}

\vspace{-2em}
  \begin{displaymath}\scriptsize
 \begin{array}{ll}
\vcenter{\hbox{\scalebox{0.7}{
    \begin{tikzpicture}
    \node[] (name) at (1cm, 2cm) {(Vanishing2)};
    \node[draw, rectangle, minimum size=1cm] (mdpA) at (1cm, 0cm) {\scalebox{1.5}{$\mdp{A}$}};
    \draw[<-] ($(mdpA.north west)!0.5!(mdpA.west)$) to [out=165,in=195] (0.5cm,0.7cm);
    \draw ($(mdpA.north east)!0.5!(mdpA.east)$) to [out=15,in=345] (1.5cm,0.7cm);
    \draw[] (0.5cm, 0.7cm) to (1.5cm, 0.7cm);
    \node[] (doml) at (1cm, 0.9cm) {$l_1+l_2$}; 
    \draw[->] (0cm, -0.25cm) to ($(mdpA.south west)!0.5!(mdpA.west)$); 
    \draw[->] ($(mdpA.south east)!0.5!(mdpA.east)$) to (2cm, -0.25cm);
    \node[] (domm) at (0.2cm, -0.1cm) {$m$};
    \node[] (codomn) at (1.8cm, -0.1cm) {$n$};
    \node[] (translate) at (2.7cm, 0cm) {\large $=$}; 
    \node[draw, rectangle, minimum size=1cm] (mdpA2) at (4.5cm, 0cm) {\scalebox{1.5}{$\mdp{A}$}};
    \draw[<-] ($(mdpA2.north west)!0.5!(mdpA2.west)$) to [out=165,in=195] (4cm,0.7cm);
    \draw (4cm, 0.7cm) to (5cm, 0.7cm);
    \draw ($(mdpA2.north east)!0.5!(mdpA2.east)$) to [out=15,in=345] (5cm,0.7cm);
    \draw[->] (3.5cm, -0.25cm) to  ($(mdpA2.south west)!0.5!(mdpA2.west)$);
    \draw[->] ($(mdpA2.south east)!0.5!(mdpA2.east)$) to (5.5cm, -0.25cm);
    \node[] (doml2) at (4.5cm, 0.85cm) {$l_1$}; 
    \draw[<-] (mdpA2.west) to [out=165,in=195] (4cm,1.05cm);
    \draw[] (4cm, 1.05cm) to (5cm, 1.05cm);
    \draw (mdpA2.east) to [out=15,in=345] (5cm,1.05cm);
     \node[] (doml3) at (4.5cm, 1.2cm) {$l_2$};
    \node[] (domm2) at (3.7cm, -0.1cm) {$m$};
    \node[] (codomn2) at (5.3cm, -0.1cm) {$n$};
    \end{tikzpicture}   
    }}}
    &\vcenter{\hbox{\scalebox{0.7}{
    \begin{tikzpicture}
    \node[] (name) at (1cm, 1.3cm) {(Superposing)};
    \node[draw, rectangle, minimum size=1cm] (mdpA) at (1cm, 0cm) {\scalebox{1.5}{$\mdp{A}$}};
    \draw[<-] ($(mdpA.north west)!0.5!(mdpA.west)$) to [out=165,in=195] (0.5cm,0.7cm);
    \draw ($(mdpA.north east)!0.5!(mdpA.east)$) to [out=15,in=345] (1.5cm,0.7cm);
    \draw[] (0.5cm, 0.7cm) to (1.5cm, 0.7cm);
    \node[] (doml) at (1cm, 0.9cm) {$l$}; 
    \draw[->] (0cm, -0.25cm) to ($(mdpA.south west)!0.5!(mdpA.west)$); 
    \draw[->] ($(mdpA.south east)!0.5!(mdpA.east)$) to (2cm, -0.25cm);
    \node[] (domm) at (0.2cm, -0.1cm) {$m_1$};
    \node[] (codomn) at (1.8cm, -0.1cm) {$n_1$};
    \node[draw, rectangle, minimum size=0.5cm] (mdpB) at (1cm, -1cm) {\scalebox{1.5}{$\mdp{B}$}};
    \draw[->] (0cm, -1cm) to (mdpB.west);
    \draw[->] (mdpB.east) to (2cm, -1cm);
    \node[] (domm2) at (0.4cm, -0.85cm) {$m_2$};
    \node[] (codomn2) at (1.6cm, -0.85cm) {$n_2$};
    \node[] (translate) at (2.7cm, 0cm) {\large $=$}; 
    \node[draw, rectangle, minimum size=1cm] (mdpA2) at (4.5cm, 0cm) {\scalebox{1.5}{$\mdp{A}$}};
    \node[] (doml) at (4.5cm, 0.9cm) {$l$};
     \node[draw, rectangle, minimum size=0.5cm] (mdpB) at (4.5cm, -1cm) {\scalebox{1.5}{$\mdp{B}$}};
    \draw (3.5cm, 0.25cm) to [out=165,in=195] (3.5cm,0.7cm);
    \draw[] (3.5cm, 0.7cm) to (5.5cm, 0.7cm);
    \draw (5.5cm, 0.25cm) to [out=15,in=345] (5.5cm,0.7cm);
    \draw[->] (3.5cm, 0.25cm) to  ($(mdpA2.north west)!0.5!(mdpA2.west)$);
    \draw[] ($(mdpA2.north east)!0.5!(mdpA2.east)$) to (5.5cm, 0.25cm);
    \draw[->] (3.5cm, -0.25cm) to  ($(mdpA2.south west)!0.5!(mdpA2.west)$);
    \draw[->] ($(mdpA2.south east)!0.5!(mdpA2.east)$) to (5.5cm, -0.25cm);
    \draw[->] (3.5cm, -1cm) to (mdpB.west);
    \draw[->] (mdpB.east) to (5.5cm, -1cm);
        \node[] (domm2) at (3.9cm, -0.85cm) {$m_2$};
    \node[] (codomn2) at (5.1cm, -0.85cm) {$n_2$};
    \node[] (domm2) at (3.7cm, -0.1cm) {$m_1$};
    \node[] (codomn2) at (5.3cm, -0.1cm) {$n_1$};
    \draw[dashed] (3.7cm, -1.35cm) to (5.3cm, -1.35cm);
    \draw[dashed] (3.7cm, -1.35cm) to (3.7cm, 0.6cm);
    \draw[dashed] (5.3cm, 0.6cm) to (3.7cm, 0.6cm);
    \draw[dashed] (5.3cm, 0.6cm) to (5.3cm, -1.35cm);
    \end{tikzpicture}   
    }}}
    \\
    &\\
\vcenter{\hbox{\scalebox{0.7}{
    \begin{tikzpicture}
    \node[] (name) at (-0.25cm, 1.2cm) {(Naturality1)};
    \node[draw, rectangle, minimum size=0.5cm] (mdpB) at (-0.25cm, -0.25cm) {\scalebox{1.5}{$\mdp{B}$}};
    \node[draw, rectangle, minimum size=1cm] (mdpA) at (1cm, 0cm) {\scalebox{1.5}{$\mdp{A}$}};
    \draw[->] (-1cm, -0.25cm) to (mdpB.west);
    \draw[->] (0cm, -0.25cm) to ($(mdpA.south west)!0.5!(mdpA.west)$); 
    \draw[->] (-1cm, 0.25cm) to ($(mdpA.north west)!0.5!(mdpA.west)$);
     \node[] (doml) at (0.5cm, 0.85cm) {$l$};
    \draw (-1cm, 0.25cm) to [out=165,in=195] (-1cm,0.7cm);
    \draw (2cm, 0.25cm) to [out=15,in=345] (2cm, 0.7cm);
    \draw[] (-1cm,0.7cm) to (2cm, 0.7cm);
    \draw[] ($(mdpA.north east)!0.5!(mdpA.east)$) to (2cm, 0.25cm);
    \draw[] ($(mdpA.south east)!0.5!(mdpA.east)$) to (2cm, -0.25cm);
    \node[] (domk) at (0.2cm, -0.1cm) {$k$};
    \node[] (domm) at (-0.75cm, -0.1cm) {$m$};
    \node[] (codomn) at (1.8cm, -0.1cm) {$n$};
    \node[] (translate) at (2.8cm, 0cm) {\large $=$};
    \node[draw, rectangle, minimum size=0.5cm] (mdpB2) at (4.15cm, -0.25cm) {\scalebox{1.5}{$\mdp{B}$}};
    \node[] (doml4) at (3.6cm, -0.1cm) {$m$};
    \draw[->] (3.4cm, -0.25cm) to (mdpB2.west);
    \node[draw, rectangle, minimum size=1cm] (mdpA2) at (5.4cm, 0cm) {\scalebox{1.5}{$\mdp{A}$}};
    \draw[<-] ($(mdpA2.north west)!0.5!(mdpA2.west)$) to [out=165,in=195] (4.9cm,0.7cm);
    \draw[] (4.9cm, 0.7cm) to (5.9cm, 0.7cm);
    \draw ($(mdpA2.north east)!0.5!(mdpA2.east)$) to [out=15,in=345] (5.9cm,0.7cm);
    \draw[->] (4.4cm, -0.25cm) to  ($(mdpA2.south west)!0.5!(mdpA2.west)$);
    \draw[->] ($(mdpA2.south east)!0.5!(mdpA2.east)$) to (6.4cm, -0.25cm);
     \node[] (doml3) at (5.4cm, 0.85cm) {$l$};
    \node[] (domm2) at (4.6cm, -0.1cm) {$k$};
    \node[] (codomn2) at (6.2cm, -0.1cm) {$n$};
    \end{tikzpicture}   
    }}}
    \hspace{30pt}
     &\vcenter{\hbox{\scalebox{0.7}{
    \begin{tikzpicture}
        \node[] (name) at (-0.25cm, 1.2cm) {(Dinaturality)};
    \node[draw, rectangle, minimum size=0.5cm] (mdpB) at (-0.25cm, 0.25cm) {\scalebox{1.5}{$\mdp{B}$}};
    \node[draw, rectangle, minimum size=1cm] (mdpA) at (1cm, 0cm) {\scalebox{1.5}{$\mdp{A}$}};
    \draw[->] (-1cm, 0.25cm) to (mdpB.west);
    \draw[->] (mdpB.east) to ($(mdpA.north west)!0.5!(mdpA.west)$);
    \node[] (domk) at (0.2cm, 0.4cm) {$k$};
    \draw[] ($(mdpA.north east)!0.5!(mdpA.east)$) to (2cm, 0.25cm);
    \draw (-1cm, 0.25cm) to [out=165,in=195] (-1cm,0.7cm);
    \draw (2cm, 0.25cm) to [out=15,in=345] (2cm,0.7cm);
    \draw[] (-1cm, 0.7cm) to (2cm, 0.7cm);
    \node[] (doml) at (0.5cm, 0.85cm) {$l$};
    \draw[->] (-1cm, -0.25cm) to ($(mdpA.south west)!0.5!(mdpA.west)$);
    \draw[->] ($(mdpA.south east)!0.5!(mdpA.east)$) to (2cm, -0.25cm);
    \node[] (domm) at (-0.25cm, -0.15cm) {$m$};
    \node[] (codomn) at (1.8cm, -0.15cm) {$n$};
    \node[] (translate) at (2.8cm, 0cm) {\large $=$}; 
    \node[draw, rectangle, minimum size=1cm] (mdpA2) at (4.4cm, 0cm) {\scalebox{1.5}{$\mdp{A}$}};
    \draw[->] (3.4cm, 0.25cm) to ($(mdpA2.north west)!0.5!(mdpA2.west)$);
    \draw[->] (3.4cm, -0.25cm) to ($(mdpA2.south west)!0.5!(mdpA2.west)$);
    \node[draw, rectangle, minimum size=0.5cm] (mdpB2) at (5.7cm, 0.25cm) {\scalebox{1.5}{$\mdp{B}$}};
    \draw[->] ($(mdpA2.north east)!0.5!(mdpA2.east)$) to (mdpB2.west);
    \draw[] (mdpB2.east) to (6.4cm, 0.25cm);
    \draw[->] ($(mdpA2.south east)!0.5!(mdpA2.east)$) to (6.4cm, -0.25cm);
    \draw (3.4cm, 0.25cm) to [out=165,in=195] (3.4cm,0.7cm);
    \draw (6.4cm, 0.25cm) to [out=15,in=345] (6.4cm,0.7cm);
    \draw[] (3.4cm, 0.7cm) to (6.4cm, 0.7cm);
    \node[] (doml2) at (5cm, 0.85cm) {$k$};
    \node[] (domk2) at (5.15cm, 0.4cm) {$l$};
    \node[] (domm) at (3.6cm, -0.15cm) {$m$};
    \node[] (codomn) at (5.7cm, -0.15cm) {$n$};
    \end{tikzpicture}   
    }}}
 \end{array} 
\end{displaymath}

\caption{The equational axioms of TSMCs, expressed for roMDPs, with some string diagram illustrations. Here we omit types of roMDPs; see~\cite{joyal1996} for details.  
}
\label{fig:TSMCAxioms}
\label{fig:TSMCAxiomsInStringDiagrams}
\end{figure}
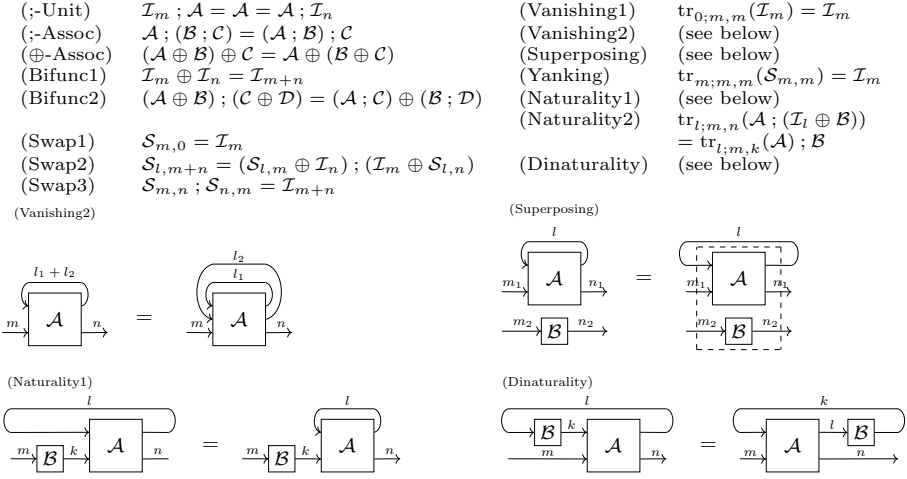

 The TSMC axioms  use 
 some ``positionless'' roMDPs as wires, such as  \emph{identities} $\ID_{m}$ ($\stackrel{m}{\text{---}\!\text{---}}$ in string diagrams) and  \emph{swaps} $\SYM_{m,  n}$ ($\times$).
 \iffull
 See \cref{def:idSwapRoMDPs}.
 \else
 See~\cite[Appendix A]{Watanabe23long} for details. 
 \fi
\begin{auxproof}
 \begin{definition}[identity, swap]\myindent\label{def:idSwapRoMDPs}
 Let $m, n$ be natural numbers. The \emph{identity} $\ID_{m}$ on $m$ (over $A$) is given by \(\ID_{m}=(m,\ m,\allowbreak\ \emptyset,\ A,\ E,\ !,\ !)\), where $E(i) = i$ for each $i\in \nset{m}$. The \emph{swap} $\SYM_{m,  n}$ on $m$ and $n$ (\emph{over} A) is given by $\SYM_{m,  n}=(m+n,\, n+m,\,\emptyset, A, E, !, !)$, where (1) $E(i) = i+n$ if $i\in \nset{m}$, and (2) $E(i) = i-n$  if $i\in [m+1, m+n]$.
 \end{definition}
\end{auxproof}
The proof of the following is routine. 
\iffull
For details, see~\cref{sec:proof_roMDPsTSMCEqAx}.
\else 
For details, see~\cite[Appendix B]{Watanabe23long}.
\fi

\begin{theorem}\label{thm:roMDPsTSMCEqAx}
 The three operations $\seqcomp, \oplus, \tr$ on roMDPs, defined in~\cref{subsec:roMDPsAndTSMCStringDiagrams}, satisfy the equational axioms in~\cref{fig:TSMCAxioms} up-to  isomorphisms 
 \iffull
 (\cref{def:isomorphismRoMDPs}).
 \else 
 (see~\cite[Appendix A]{Watanabe23long} for details).
 \fi \qed
\end{theorem}

\begin{corollary}[a TSMC $\roMDP$]\label{cor:roMDPsTSMC}
 Let $\roMDP$ be the category whose objects are natural numbers and whose arrows are roMDPs over the action set $A$ modulo 
 \iffull 
 isomorphisms (\cref{def:isomorphismRoMDPs}).
 \else 
 isomorphisms. 
 \fi Then the  operations $\seqcomp,\oplus,\tr,\ID,\SYM$ make $\roMDP$ a traced symmetric monoidal category (TSMC). \qed
\end{corollary}

\subsection{Open MDPs and  ``Compact Closed'' String Diagrams}
\label{subsec:string_diagrams_openMDPs}
Following the outline  in~\cref{subsec:stringDiagramOfMDPsOutline}, we now introduce a bidirectional ``compact closed'' calculus of open MDPs (oMDPs),  using the $\Int$ construction~\cite{joyal1996} that  turns  TSMCs in general into compact closed categories (compCCs). 

 The following definition simply says  $\oMDP\defeq\Int(\roMDP)$, although it uses concrete terms adapted to the current context. 
\begin{auxproof}
 The construction is called Int since it resembles the construction of integers by pairs of natural numbers modulo a suitable equivalence (namely $(n,m)=n-m$) . 
\end{auxproof}

\begin{definition}[the category $\oMDP$]\label{def:oMDPCat}
 The \emph{category $\oMDP$ of open MDPs} is defined as follows. Its objects are pairs $(\nfr{m},\nfl{m})$ of natural numbers. Its arrows are defined by rightward open MDPs as follows:
\begin{equation}\label{eq:arrowsOfOMDPCat}
 \vcenter{\infer={ 
   \text{an arrow } \mdp{A}\colon \nfr{m}+\nfl{n}\longrightarrow  \nfr{n}+\nfl{m}
   \text{ in $\roMDP$, i.e.\ an roMDP}
   }{
   \text{an arrow } (\nfr{m}, \nfl{m}) \longrightarrow (\nfr{n}, \nfl{n}) 
   \text{ in $\oMDP$}
   }}
\end{equation}
where the double lines $=\joinrel=$ mean ``is the same thing as.'' 
\end{definition}

The definition may not immediately justify its name: no open MDPs appear there; only roMDPs do. The point is that we identify the roMDP $\mdp{A}$ in~\cref{eq:arrowsOfOMDPCat} with the oMDP $\Psi(\mdp{A})$ of the designated type, using ``twists'' in \cref{fig:twist}. 
\iffull 
See \cref{prop:twist}.
\else 
See \cite[Appendix A]{Watanabe23long} for details.
\fi
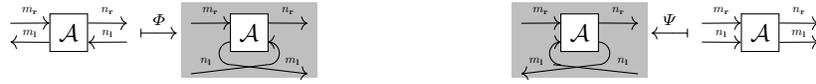
\begin{figure}[tbp]
 \centering
    \begin{subfigure}[b]{0.45\columnwidth}
        \centering
        \begin{tikzpicture}[
              innode/.style={draw, rectangle, minimum size=0.5cm},
              innodemini/.style={draw, rectangle, minimum size=0.5cm},
              interface/.style={inner sep=0},
              innodepos/.style={draw, circle, minimum size=0.2cm},
              ]
              \node[interface] (rdo1) at (0cm, 0.125cm) {};
              \node[innode, fill=white] (pos1) at (0.8cm, 0cm) {$\mdp{A}$};
              \node[interface] (ldo1) at (0cm, -0.125cm) {};
              \draw[->] (rdo1) to node[above] {\scalebox{0.5}{$\nfr{m}$}} ($(pos1.north west)!0.5!(pos1.west)$);
              \draw[<-] (ldo1) to node[above,yshift=-0.05cm] {\scalebox{0.5}{$\nfl{m}$}} ($(pos1.south west)!0.5!(pos1.west)$);
              \node[interface] (rcdo1) at (1.6cm, 0.125cm) {};
              \node[interface] (lcdo1) at (1.6cm, -0.125cm) {};
              \draw[<-] (rcdo1) to node[above] {\scalebox{0.5}{$\nfr{n}$}} ($(pos1.north east)!0.5!(pos1.east)$);
              \draw[->] (lcdo1) to node[above,yshift=-0.05cm] {\scalebox{0.5}{$\nfl{n}$}} ($(pos1.south east)!0.5!(pos1.east)$);
              \node[interface] (maps1) at (2cm, 0.1cm) {$\overset{\Phi}{\longmapsto}$};
              \fill[lightgray](2.3cm, 0.4cm)--(2.3cm, -0.6cm)--(4.1cm, -0.6cm)--(4.1cm, 0.4cm)--cycle;
              \node[interface] (rdo11) at (2.4cm, 0.125cm) {};
              \node[interface] (ldo11) at (2.4cm, -0.55cm) {};
              \node[interface] (rcdo11) at (4cm, 0.125cm) {};
              \node[interface] (lcdo11) at (4cm, -0.55cm) {};
              \node[innode, fill=white] (pos2) at (3.2cm, 0cm) {$\mdp{A}$};
              \draw[->] (rdo11) to node[above] {\scalebox{0.5}{$\nfr{m}$}} ($(pos2.north west)!0.5!(pos2.west)$);
              \draw[-] (2.95cm, -0.125cm) arc [radius=0.15, start angle = 90, end angle=270];
              \draw[->] ($(pos2.north east)!0.5!(pos2.east)$) to node[above] {\scalebox{0.5}{$\nfr{n}$}} (rcdo11);
              \draw[->] (2.95cm, -0.425cm) to node[at end, above, xshift=-0.2cm] {\scalebox{0.5}{$\nfl{m}$}} (lcdo11);
              \draw[->] (3.45cm, -0.425cm) arc [radius=0.15, start angle = 270, end angle=450];
              \draw[-] (3.45cm, -0.425cm) to node[at end, above, xshift=0.2cm] {\scalebox{0.5}{$\nfl{n}$}} (ldo11);
        \end{tikzpicture}
    \end{subfigure}
    \hfill
    \begin{subfigure}[b]{0.45\columnwidth}
        \centering
        \begin{tikzpicture}[
              innode/.style={draw, rectangle, minimum size=0.5cm},
              innodemini/.style={draw, rectangle, minimum size=0.5cm},
              interface/.style={inner sep=0},
              innodepos/.style={draw, circle, minimum size=0.2cm},
              ]
              \fill[lightgray](-0.1cm, 0.4cm)--(-0.1cm, -0.6cm)--(1.7cm, -0.6cm)--(1.7cm, 0.4cm)--cycle;
              \node[interface] (rdo11) at (0cm, 0.125cm) {};
              \node[interface] (ldo11) at (0cm, -0.55cm) {};
              \node[interface] (rcdo11) at (1.6cm, 0.125cm) {};
              \node[interface] (lcdo11) at (1.6cm, -0.55cm) {};
              \node[innode, fill=white] (pos2) at (0.8cm, 0cm) {$\mdp{A}$};
              \draw[->] (rdo11) to node[above] {\scalebox{0.5}{$\nfr{m}$}} ($(pos2.north west)!0.5!(pos2.west)$);
              \draw[<-] (0.55cm, -0.125cm) arc [radius=0.15, start angle = 90, end angle=270];
              \draw[->] ($(pos2.north east)!0.5!(pos2.east)$) to node[above] {\scalebox{0.5}{$\nfr{n}$}} (rcdo11);
              \draw[-] (0.45cm, -0.425cm) to node[at end, above, xshift=-0.2cm] {\scalebox{0.5}{$\nfl{n}$}} (lcdo11);
              \draw[-] (1.05cm, -0.425cm) arc [radius=0.15, start angle = 270, end angle=450];
              \draw[->] (1.05cm, -0.425cm) to node[at end, above, xshift=0.2cm] {\scalebox{0.5}{$\nfl{m}$}} (ldo11);
              \node[interface] (maps1) at (2cm, 0.1cm) {$\overset{\Psi}{\reflectbox{$\longmapsto$}}$};
              \node[interface] (rdo1) at (2.4cm, 0.125cm) {};
              \node[innode, fill=white] (pos1) at (3.2cm, 0cm) {$\mdp{A}$};
              \node[interface] (ldo1) at (2.4cm, -0.125cm) {};
              \draw[->] (rdo1) to node[above] {\scalebox{0.5}{$\nfr{m}$}} ($(pos1.north west)!0.5!(pos1.west)$);
              \draw[->] (ldo1) to node[above,yshift=-0.05cm] {\scalebox{0.5}{$\nfl{n}$}} ($(pos1.south west)!0.5!(pos1.west)$);
              \node[interface] (rcdo1) at (4cm, 0.125cm) {};
              \node[interface] (lcdo1) at (4cm, -0.125cm) {};
              \draw[<-] (rcdo1) to node[above] {\scalebox{0.5}{$\nfr{n}$}} ($(pos1.north east)!0.5!(pos1.east)$);
              \draw[<-] (lcdo1) to node[above,yshift=-0.05cm] {\scalebox{0.5}{$\nfl{m}$}} ($(pos1.south east)!0.5!(pos1.east)$);
        \end{tikzpicture}
    \end{subfigure}
    \caption{Turning oMDPs to roMDPs, and vice versa, via twists.}
     \label{fig:twist}
\end{figure}

We move on to describe algebraic operations for composing oMDPs. These operations come from the structure of $\oMDP$ as a compCC; the latter, in turn, arises canonically from the $\Int$ construction. 

\begin{definition}[$\seqcomp$ of oMDPs]
\label{def:seqCompOpenMDPs}
Let $\mdp{A}:(\nfr{m},\nfl{m})\rightarrow (\nfr{l},\nfl{l})$ and $\mdp{B}:(\nfr{l},\nfl{l})\rightarrow (\nfr{n},\nfl{n})$ be arrows in $\oMDP$ with the same action set $A$. Their \emph{sequential composition} $\mdp{A}\seqcomp\mdp{B}:(\nfr{m},\nfl{m})\rightarrow (\nfr{n},\nfl{n})$ is
defined by the string diagram  in~\cref{fig:seqCompIntOMDP}, formulated  in $\roMDP$. Textually the definition is
\begin{math}
  \mdp{A}\seqcomp\mdp{B}
 \defeq \trmdp{\nfl{l}}{\nfr{m} + \nfl{n}}{\nfr{n}+ \nfl{m}}\allowbreak\big((\SYM_{\nfl{l}, \nfr{m}}\oplus\ID_{\nfl{n}}) \seqcomp (\mdp{A}\oplus \ID_{\nfl{n}}) \seqcomp (\ID_{\nfr{l}}\oplus \SYM_{\nfl{m},\nfl{n}}) \seqcomp (\mdp{B}\oplus\ID_{\nfl{m}}) \seqcomp (\SYM_{\nfr{n},\nfl{l}}\oplus\ID_{\nfl{m}})\big)
\end{math}.
\end{definition}

\begin{figure}[tbp]
  \centering
    \begin{subfigure}[b]{0.45\columnwidth}
    \centering
    \scalebox{0.5}{
    \begin{tikzpicture}
    \draw (0cm, 1cm) to [out=165,in=195] (0cm,1.5cm);
    \draw[-] (0cm, 1.5cm) to (5.25cm, 1.5cm);
    \draw[arrows = {<[length=2mm]-}] (5.25cm, 1.5cm) to (10.5cm, 1.5cm);
    \draw (10.5cm, 1cm) to [out=15,in=345] (10.5cm,1.5cm);
    \node[] (doml) at (0.25cm, 1.2cm) {$\nfl{l}$};
    \node[] (domm) at (0.25cm, 0.2cm) {$\nfr{m}$};
    \node[] (domn) at (0.25cm, -0.8cm) {$\nfl{n}$};
    \draw[-] (0cm, 0cm) to (0.5cm, 0cm);
    \draw[-] (0.5cm, 0cm) to (1.5cm, 1cm);
    \draw[-] (1.5cm, 1cm) to (2cm, 1cm);
    \draw[-] (0cm, 1cm) to (0.5cm, 1cm);
    \draw[-] (0.5cm, 1cm) to (1.5cm, 0cm);
    \draw[-] (1.5cm, 0cm) to (2cm, 0cm);
    \draw[-] (0cm, -1cm) to (4.5cm, -1cm);
    \node[draw, rectangle, minimum size=1.5cm] (mdpA) at (3cm, 0.5cm) {\scalebox{2}{$\mdp{A}$}};
    \draw[arrows = {->[length=2mm]}] (2cm, 1cm) to ($(mdpA.north west)!0.33!(mdpA.west)$);
    \draw[arrows = {->[length=2mm]}] (2cm, 0cm) to ($(mdpA.south west)!0.33!(mdpA.west)$);
    \draw[-] ($(mdpA.north east)!0.33!(mdpA.east)$) to (6cm, 1cm);
    \draw[-] ($(mdpA.south east)!0.33!(mdpA.east)$) to (4.5cm, 0cm);
    \node[] (codoml) at (4.25cm, 1.2cm) {$\nfr{l}$};
    \node[] (codomm) at (4.25cm, 0.2cm) {$\nfl{m}$};
    \draw[-] (4.5cm, 0cm) to (5.5cm, -1cm);
    \draw[-] (4.5cm, -1cm) to (5.5cm, 0cm);
    \draw[-] (5.5cm, -1cm) to (6.5cm, -1cm);
    \draw[-] (5.5cm, 0cm) to (6cm, 0cm);
    \node[draw, rectangle, minimum size=1.5cm] (mdpB) at (7cm, 0.5cm) {\scalebox{2}{$\mdp{B}$}};
    \draw[arrows = {->[length=2mm]}] (6cm, 1cm) to ($(mdpB.north west)!0.33!(mdpB.west)$);
    \draw[arrows = {->[length=2mm]}] (6cm, 0cm) to ($(mdpB.south west)!0.33!(mdpB.west)$);
    \draw[arrows = {->[length=2mm]}] (6.5cm, -1cm) to (10.5cm, -1cm);
    \draw[-] ($(mdpB.south east)!0.33!(mdpB.east)$) to (9cm, 0cm);
    \draw[-] ($(mdpB.north east)!0.33!(mdpB.east)$) to (9cm, 1cm);
    \node[] (codoml2) at (8.75cm, 1.2cm) {$\nfr{n}$};
    \node[] (codomm2) at (8.75cm, 0.2cm) {$\nfl{l}$};
    \draw[-] (9cm, 0cm) to (10cm, 1cm);
    \draw[-] (9cm, 1cm) to (10cm, 0cm);
    \draw[arrows = {->[length=2mm]}] (10cm, 1cm) to (10.5cm, 1cm);
    \draw[arrows = {->[length=2mm]}] (10cm, 0cm) to (10.5cm, 0cm);
    \end{tikzpicture}
    }
\end{subfigure}
    \begin{subfigure}[b]{0.45\columnwidth}
    \centering
    \scalebox{0.5}{
    \begin{tikzpicture}
    \draw[-] (0cm, 1cm) to (0.5cm, 1cm);
    \node[] (dom1) at (0.25cm, 1.2cm) {$\nfr{m}$};
    \draw[-] (0cm, 0cm) to (0.5cm, 0cm);
    \node[] (dom2) at (0.25cm, 0.2cm) {$\nfr{k}$};
    \draw[-] (0cm, -1cm) to (0.5cm, -1cm);
    \node[] (don1) at (0.25cm, -0.8cm) {$\nfl{l}$};
    \draw[-] (0cm, -2cm) to (0.5cm, -2cm);
    \node[] (don2) at (0.25cm, -1.8cm) {$\nfl{n}$};
    \draw[-] (0.5cm, 1cm) to (1.5cm, 0cm);
    \draw[-] (0.5cm, 0cm) to (1.5cm, 1cm);
    \draw[-] (0.5cm, -1cm) to (1.5cm, -2cm);
    \draw[-] (0.5cm, -2cm) to (1.5cm, -1cm);
    \draw[-] (1.5cm, 1cm) to (4.5cm, 1cm);
    \draw[arrows = {->[length=2mm]}] (1.5cm, 0cm) to (2.25cm, 0cm);
    \draw[arrows = {->[length=2mm]}] (1.5cm, -1cm) to (2.25cm, -1cm);
    \draw[-] (1.5cm, -2cm) to (4.5cm, -2cm);
    \node[draw, rectangle, minimum size=1.5cm] (mdpA) at (3cm, -0.5cm) {\scalebox{2}{$\mdp{A}$}};
    \draw[-] ($(mdpA.north east)!0.33!(mdpA.east)$) to (4.5cm, 0cm);
    \node[] (codomAr) at (4.25cm, 0.2cm) {$\nfr{n}$};
    \draw[-] ($(mdpA.south east)!0.33!(mdpA.east)$) to (4.5cm, -1cm);
    \node[] (codomAl) at (4.25cm, -0.8cm) {$\nfl{m}$};
    \draw[-] (4.5cm, 1cm) to (5.5cm, 0cm);
    \draw[-] (4.5cm, 0cm) to (5.5cm, 1cm);
    \draw[-] (4.5cm, -1cm) to (5.5cm, -2cm);
    \draw[-] (4.5cm, -2cm) to (5.5cm, -1cm);
    \draw[arrows = {->[length=2mm]}] (5.5cm, 1cm) to (8.5cm, 1cm);
    \draw[arrows = {->[length=2mm]}] (5.5cm, 0cm) to (6.25cm, 0cm);
    \draw[arrows = {->[length=2mm]}] (5.5cm, -1cm) to (6.25cm, -1cm);
    \draw[arrows = {->[length=2mm]}] (5.5cm, -2cm) to (8.5cm, -2cm);
    \node[draw, rectangle, minimum size=1.5cm] (mdpB) at (7cm, -0.5cm) {\scalebox{2}{$\mdp{B}$}};
    \draw[arrows = {->[length=2mm]}] ($(mdpB.north east)!0.33!(mdpB.east)$) to (8.5cm, 0cm);
    \node[] (codomBr) at (8.1cm, 0.2cm) {$\nfr{l}$};
    \draw[arrows = {->[length=2mm]}] ($(mdpB.south east)!0.33!(mdpB.east)$) to (8.5cm, -1cm);
    \node[] (codomBl) at (8.1cm, -0.8cm) {$\nfl{k}$};
    \end{tikzpicture}
    }
\end{subfigure}
 \caption{String diagrams in $\roMDP$ for $\mdp{A}\seqcomp\mdp{B}$, $\mdp{A}\oplus\mdp{B}$ in $\oMDP$.}
\label{fig:seqCompIntOMDP}
\end{figure}

The definition of \emph{sum} $\oplus$ of oMDPs is similarly shown in the string diagram in \cref{fig:seqCompIntOMDP}, formulated in $\roMDP$. Definition of ``wires'' such as identities, swaps, \emph{units} ($\subset$ in string diagrams) and \emph{counits} ($\supset$) is easy, too.

\begin{auxproof}
 \begin{definition}[sum $\oplus$ of oMDPs]
 \label{def:sumOpenMDPs}
 Let $\mdp{A}:(\nfr{m},\nfl{m})\rightarrow (\nfr{n},\nfl{n})$ and $\mdp{B}:(\nfr{k},\nfl{k})\rightarrow (\nfr{l},\nfl{l})$ be arrows with the same action set $A$. Their sum $\mdp{A}\oplus\mdp{B}:(\nfr{m}+\nfr{k},\nfl{k}+\nfl{m})\rightarrow (\nfr{n}+\nfr{l},\nfl{l}+\nfl{n})$ is given by $\mdp{A}\oplus\mdp{B}\defeq (\SYM_{\nfr{m}, \nfr{k} }\oplus \SYM_{\nfl{l}, \nfl{n}}) \seqcomp (\ID_{\nfr{k}}\oplus \mdp{A}\oplus\ID_{\nfl{l}}) \seqcomp (\SYM_{\nfr{k}, \nfr{n}}\oplus \SYM_{\nfl{m}, \nfl{l}}) \seqcomp (\ID_{\nfr{n}}\oplus\mdp{B} \oplus \ID_{\nfl{m}})$.
 \end{definition}
 We emphasize that the above mathematical definitions are rigorous embodiments of the graphical intuition in~\cref{fig:seqCompOplusIllustrated}. 
\end{auxproof}

\begin{auxproof}
 In compCC string diagrams for oMDPs, we additionally use the following wires.
 \begin{definition}[identity, swap, unit, and counit]
 \label{def:seqCompSumOpenMDPs}
 Let $(\nfr{m}, \nfl{m}), (\nfr{n}, \nfl{n})\in \Nat\times \Nat$. The \emph{identity} $\ID_{(\nfr{m}, \nfl{m})}:(\nfr{m}, \nfl{m})\rightarrow (\nfr{m}, \nfl{m})$ on $(\nfr{m}, \nfl{m})$ is given by $\ID_{(\nfr{m}, \nfl{m})}\defeq\ID_{\nfr{m}+\nfl{m}}$. The \emph{swap} $\SYM_{(\nfr{m}, \nfl{m}), (\nfr{n}, \nfl{n})}:(\nfr{m}+\nfr{n}, \nfl{n}+\nfl{m})\rightarrow (\nfr{n}+\nfr{m}, \nfl{m}+\nfl{n})$ on $(
 \nfr{m}, \nfl{m}), (\nfr{n}, \nfl{n})$ is given by $\SYM_{(\nfr{m}, \nfl{m}), (\nfr{n}, \nfl{n})}\defeq \SYM_{\nfr{m},\nfr{n}}\oplus \SYM_{\nfl{n},\nfl{m}}$. The \emph{unit} $\dunit_{(\nfr{m}, \nfl{m})}:(0, 0)\rightarrow (\nfr{m}+\nfl{m}, \nfr{m}+\nfl{m})$ on $(\nfr{m}, \nfl{m})$ is given by $\dunit_{(\nfr{m}, \nfl{m})}\defeq \ID_{\nfr{m}+\nfl{m}}$, and the \emph{counit} $\dcounit_{(\nfr{m}, \nfl{m})}:(\nfl{m} + \nfr{m}, \nfl{m}+\nfr{m})\rightarrow (0, 0)$ on $(\nfr{m}, \nfl{m})$ is given by $\dcounit_{(\nfr{m}, \nfl{m})}\defeq \ID_{\nfl{m}+\nfr{m}}$.
 \end{definition}
\end{auxproof}

\begin{theorem}[$\oMDP$ is a compCC]\label{thm:oMDPIsCompactClosed}
 The category $\oMDP$ (\cref{def:oMDPCat}), equipped with the operations $\seqcomp,\oplus$, is a compCC.
\qed
\end{theorem}

\section{Decomposition Equalities for Open Markov Chains}
\label{sec:compositionalAnalyssiOfOpenMarkovChains}
Here we exhibit some basic equalities that decompose the behavior of (rightward open) Markov chains. We start with such equalities on \emph{reachability probabilities} (which are widely known) and extend them to equalities on \emph{expected rewards} (which seem less known). Notably, the latter equalities involve not only expected rewards but also reachability probabilities.

Here we focus on \emph{rightward open Markov chains (roMCs)}, since the extension to richer settings is taken care of by categorical constructions. See \cref{fig:catsFunctors}.
\begin{definition}[roMC]
A \emph{rightward open Markov chain (roMC)} $\mc{C}$ from $m$ to $n$ is an roMDP from $m$ to $n$ over the singleton action set $\{\star\}$.

For an roMC $\mc{C}$, its \emph{reachability probability} $\Reacha{i}{j}{\mc{C}}$ and \emph{expected reward} $\ETRa{i}{j}{\mc{C}}$ are defined as in \cref{def:openMDPReachProbExpCumRew}.  The scheduler $\tau$ is omitted since it is unique.

Rightward open MCs, as a special case of roMDPs, form a TSMC (\cref{cor:roMDPsTSMC}). It is denoted by $\roMC$. 
\end{definition}

\vspace{.1em}
The following equalities are well-known,
although they are  not stated in terms of open MCs. 
Recall that $\Reacha{i}{k}{\mc{C}}$ is the probability of reaching the exit $k$ from the entrance $i$ in $\mc{C}$ (\cref{def:openMDPReachProbExpCumRew}). Recall also the  definitions of $\mc{C}\seqcomp\mc{D}$ (\cref{def:seqCompRightwardOpenMDPs}) and $\trmdp{l}{m}{n}{}(\mc{E})$ (\cref{def:trace_tmdp}), which are essentially as in \cref{fig:seqCompOplusIllustrated} and \cref{fig:trace_of_rmdp}.

\begin{proposition}[decomposition equalities for $\mathrm{RPr}$]
\label{prop:MCeq}
    Let $\mc{C}:m\rightarrow l$, $\mc{D}:l\rightarrow n$ and $\mc{E}:l+m\rightarrow l+n$ be  roMCs. The following matrix equalities hold. 
    \begin{align}
&\small
\begin{array}{l}
         \bigl[\,\Reacha{i}{j}{\mc{C} \seqcomp \mc{D}}\,\bigr]_{i\in\nset{m},j\in\nset{n}}  \;=\; 
     \bigl[\,\Reacha{i}{k}{\mc{C}}\,\bigr]_{i\in\nset{m},k\in\nset{l}} \cdot
          \bigl[\,\Reacha{k}{j}{\mc{D}}\,\bigr]_{k\in\nset{l},j\in\nset{n}}\ ,
\end{array}
   \\
&\small\begin{array}{l}
   \bigl[\,\Reacha{i}{j}{\trmdp{l}{m}{n}{}(\mc{E})} \,\bigr]_{i\in\nset{m},j\in\nset{n}}
    \;=\;    \bigl[\,
    \Reacha{l+i}{l+j}{\mc{E}} 
   \,\bigr]_{i\in\nset{m},j\in\nset{n}}
+ \textstyle\sum_{d\in\Nat} A\cdot B^{d} \cdot C.
\end{array}    \label{eq:infiniteSumTraceRPr}
\end{align}
 Here $\bigl[\,\Reacha{i}{j}{\mc{C} \seqcomp \mc{D}}\,\bigr]_{i\in\nset{m},j\in \nset{n}}$ denotes the $m\times n$ matrix with the designated components; other matrices are similar. The matrices $A, B, C$ are given by $A\defeq  \bigl[\,
     \Reacha{l+i}{k}{\mc{E}}
   \,\bigr]_{i\in\nset{m},k\in\nset{l}}$, $B\defeq  \bigl[\,
    \Reacha{k}{k'}{\mc{E}}
   \,\bigr]_{k\in\nset{l},k'\in\nset{l}}$,
   and $C\defeq\bigl[\,
    \Reacha{k'}{l+j}{\mc{E}}
   \,\bigr]_{k'\in\nset{l},j\in\nset{n}}$. In the last line, note that the matrix in the middle is the $d$-th power.
\qed
\end{proposition}

\vspace{.1em}
The first equality is easy, distinguishing cases  on the intermediate open end $k$ (mutually exclusive since MCs are rightward). The second  says
\begin{displaymath}
 \includegraphics[width=.45\textwidth]{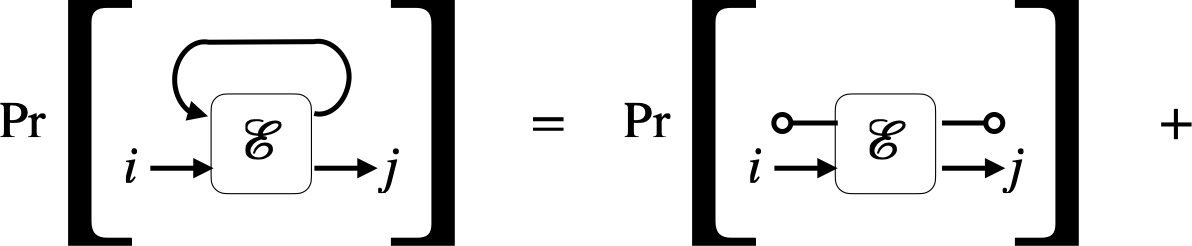}
 \hspace{7.5pt}
 \includegraphics[width=.48\textwidth]{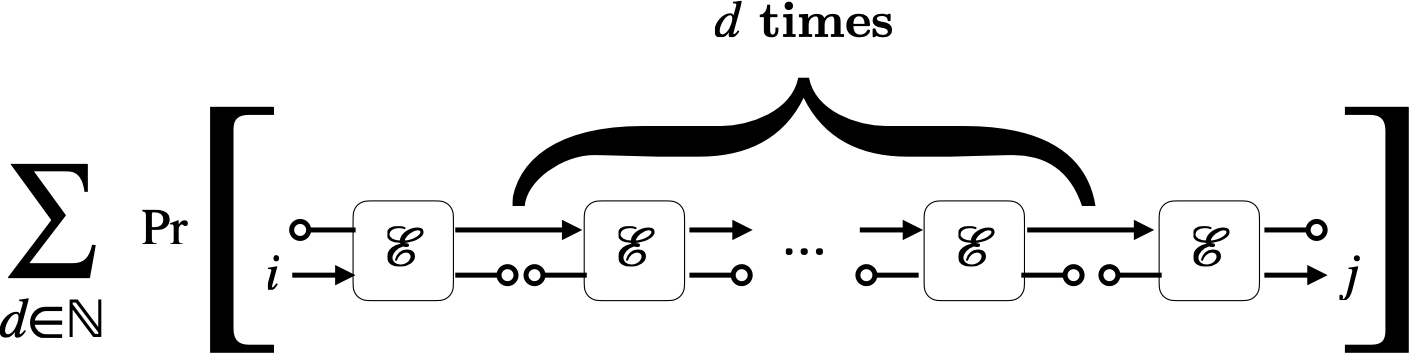}
\end{displaymath}
which is  intuitive. Here, the small circles in the diagram correspond to dead ends.
It is known as \emph{Girard's execution formula}~\cite{girard1989geometry} in linear logic. 

We  now extend \cref{prop:MCeq} to expected rewards $\ETRa{i}{j}{\mc{C}}$.

\vspace{.1em}
\noindent
\begin{proposition}[decomposition eq.\ for $\mathrm{ERw}$]
\label{prop:MCERweq}
    Let $\mc{C}:m\rightarrow l$, $\mc{D}:l\rightarrow n$ and $\mc{E}:l+m\rightarrow l+n$ be  roMCs. The following equalities of matrices hold.
\begin{align}
&\small\begin{array}{rl}
           \bigl[\,\ETRa{i}{j}{\mc{C} \seqcomp \mc{D}}\,\bigr]_{
	       i\in\nset{m}
              ,
              j\in\nset{n}
            }
     \;=\; 
     &\bigl[\,\Reacha{i}{k}{\mc{C}}\,\bigr]_{i\in\nset{m},k\in\nset{l}} \cdot
          \bigl[\,\ETRa{k}{j}{\mc{D}}\,\bigr]_{k\in\nset{l},j\in\nset{n}} 
   \\
   &+ \bigl[\,\ETRa{i}{k}{\mc{C}}\,\bigr]_{i\in\nset{m},k\in\nset{l}} \cdot
          \bigl[\,\Reacha{k}{j}{\mc{D}}\,\bigr]_{k\in\nset{l},j\in\nset{n}}\ ,
\end{array}   
\\
&\small\begin{array}{rl}
    & \bigl[\,\ETRa{i}{j}{\trmdp{l}{m}{n}{}(\mc{E})} \,\bigr]_{i\in\nset{m},j\in\nset{n}}
    \;=\;    \bigl[\,
    \ETRa{l+i}{l+j}{\mc{E}} 
   \,\bigr]_{i\in\nset{m},j\in\nset{n}}
    +  \textstyle\sum_{d\in\Nat} A\cdot B^{d}\cdot C\ .
\end{array}    \label{eq:infiniteSumTraceERw}
\end{align}
Here $A, B,C $ are the following $m\times 2l$, $2l\times 2l$, $2l\times n$ matrices.
\begin{align*}
&\small\begin{array}{ll}
A  &= \begin{pmatrix}
 \bigl[\,
    \Reacha{l+i}{k}{\mc{E}} 
   \,\bigr]_{i\in\nset{m},k\in\nset{l}} & \bigl[\,
    \ETRa{l+i}{k}{\mc{E}} 
   \,\bigr]_{i\in\nset{m},k\in\nset{l}} \\
\end{pmatrix},
\\
B &= \begin{pmatrix}
 \bigl[\,
    \Reacha{k}{k'}{\mc{E}} 
   \,\bigr]_{k\in\nset{l},k'\in\nset{l}} & \bigl[\,
    \ETRa{k}{k'}{\mc{E}} 
   \,\bigr]_{k\in\nset{l},k'\in\nset{l}} \\
   \bigl[\,
    0
   \,\bigr]_{k\in\nset{l},k'\in\nset{l}} & \bigl[\,
    \Reacha{k}{k'}{\mc{E}} 
   \,\bigr]_{k\in\nset{l},k'\in\nset{l}} 
\end{pmatrix},
\\
C &= \begin{pmatrix}
 \bigl[\,
    \ETRa{k'}{l+j}{\mc{E}} 
   \,\bigr]_{k'\in\nset{l},j\in\nset{n}}\\
   \bigl[\,
    \Reacha{k'}{l+j}{\mc{E}} 
   \,\bigr]_{k'\in\nset{l},j\in\nset{n}} \\
\end{pmatrix}.
\end{array}
\end{align*}

\vspace*{-2em}
\qed
\end{proposition}

\vspace{.1em}
\cref{prop:MCERweq} seems new, although proving them is not hard once the statements are given (\iffull
see~\cref{sec:proof_MCERweq} for details\else
see~\cite[Appendix C]{Watanabe23long} for details\fi). 
They enable one to compute the expected rewards of composite roMCs $\mc{C} \seqcomp \mc{D}$ and $\trmdp{l}{m}{n}{\mc{E}}$ from those of component roMCs $\mc{C},\mc{D},\mc{E}$. They also signify the role of reachability probabilities in such computation, suggesting their use in the definition of semantic categories (cf.\ granularity of semantics in~\cref{sec:intro}). 

The last equalities in \cref{prop:MCeq,prop:MCERweq} involve infinite sums $\sum_{d\in\Nat}$, and one may wonder how to compute them. A key is their characterization as  \emph{least fixed points} via the Kleene theorem: the desired quantity on the left side ($\mathrm{RPr}$ or $\mathrm{ERw}$) is a solution of a suitable linear equation; see \cref{prop:linear_equations}. With the given definitions, the proof of \cref{prop:MCeq,prop:MCERweq}   is (lengthy but) routine work (see e.g.~\cite[Thm. 10.15]{BaierKatoen08}).

\begin{proposition}[linear equation characterization for
\cref{eq:infiniteSumTraceRPr} and \cref{eq:infiniteSumTraceERw}]
\label{prop:linear_equations}
    Let $\mc{E}:l+m\rightarrow l+n$ be an  roMC, and $k\in \nset{l+1,l+n}$ be a specified exit of $\mc{E}$. 
 Consider the following linear equation on an unknown vector $[x_i]_{i\in \nset{l+m}}$:
    \begin{align}\label{eq:linearEqRPr}
         \bigl[\,x_i\,\bigr]_{i\in\nset{l+m}} =   \bigl[\,\Reacha{i}{k}{\mc{E}}\,\bigr]_{i\in\nset{l+m}} +  \bigl[\,\Reacha{i}{j}{\mc{E}}\,\bigr]_{i\in\nset{l+m}, j\in \nset{l}}\cdot \bigl[\,x_{j}\,\bigr]_{j\in\nset{l}}.
    \end{align}
Consider the least solution $[\tilde{x}_i]_{i\in \nset{l+m}}$ of the equation. Then its part 
$[\tilde{x}_{i+l}]_{i\in \nset{m}}$ is given by
 the vector  $\big(\Reacha{i}{k-l}{\trmdp{l}{m}{n}{}(\mc{E})}\big)_{i\in\nset{m} }$ of suitable reachability probabilities.

Moreover,
 consider the following linear equation on an unknown vector $[y_i]_{i\in \nset{l+m}}$:
    \begin{align}\label{eq:linearEqERw}
\begin{aligned}
          \bigl[\,y_i\,\bigr]_{i\in\nset{l+m}} =   \bigl[\,&\ETRa{i}{k}{\mc{E}}\,\bigr]_{i\in\nset{l+m}} + \bigl[\,\ETRa{i}{j}{\mc{E}}\,\bigr]_{i\in\nset{l+m}, j\in \nset{l}}\cdot \bigl[\,x_{j}\,\bigr]_{j\in\nset{l}}\\
         &+  \bigl[\,\Reacha{i}{j}{\mc{E}}\,\bigr]_{i\in\nset{l+m}, j\in \nset{l}}\cdot \bigl[\,y_{j}\,\bigr]_{j\in\nset{l}} ,
\end{aligned}    
\end{align}
where the unknown $[x_j]_{j\in\nset{l}}$ is shared with~\cref{eq:linearEqRPr}. 
Consider the least solution $[\tilde{y}_i]_{i\in \nset{l+m}}$ of the equation. Then its part 
$[\tilde{y}_{i+l}]_{i\in \nset{m}}$ is given by
 the vector  $\big(\ETRa{i}{k-l}{\trmdp{l}{m}{n}{}(\mc{E})}\big)_{i\in\nset{m} }$ of suitable expected rewards.

\end{proposition}

We can modify the linear equations~(\ref{eq:linearEqRPr},\ref{eq:linearEqERw})---removing unreachable positions, specifically---so that they have unique solutions
without changing the least ones. One can then solve these linear equations to compute the reachabilities and expected rewards in (\ref{eq:infiniteSumTraceRPr},\ref{eq:infiniteSumTraceERw}).
This is a well-known technique for computing reachability probabilities~\cite[Thm.~10.19]{BaierKatoen08}; it is not hard to confirm the correctness of our current extension to expected rewards.

\section{Semantic Categories and Solution Functors}
\label{sec:FatSemanticCategories}
\begin{wrapfigure}[7]{r}{0pt}
\centering
\includegraphics[width=.25\textwidth]{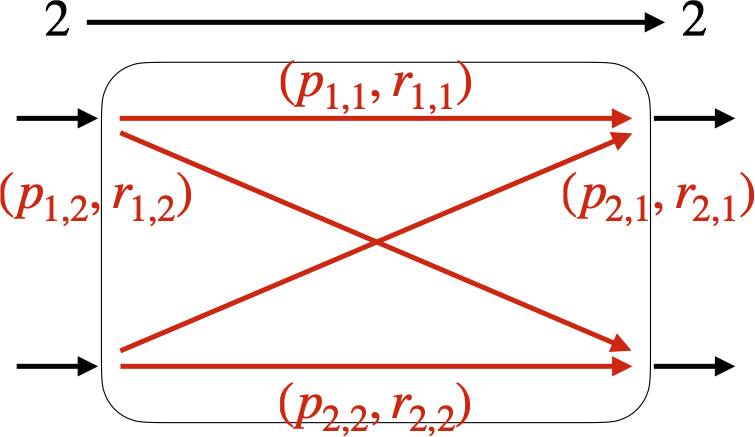}
\caption{An arrow $f\colon 2\to 2$ in $\semCatrmc$.} 
\label{fig:SroMCillust}
\end{wrapfigure}
We  build on the decomposition equalities  (\cref{prop:MCERweq}) and define the semantic category $\semCat$ for compositional model checking. This is the main construct in our  framework. 

Our definitions proceed in three steps, from roMCs to roMDPs to oMDPs (\cref{fig:catsFunctors}). The gaps between them are filled in using general constructions from category theory.

\subsection{Semantic Category for Rightward Open MCs}\label{subsec:semCatROMCs}
We first define the semantic category $\semCatrmc$ for roMCs (\cref{fig:catsFunctors}, bottom right).

\begin{definition}[objects and arrows of $\semCatrmc$]\label{def:semCatrmc}
 The \emph{category $\semCatrmc$} has natural numbers $m$ as objects. Its arrow $f\colon m \to n$ is given by an assignment, for each pair $(i,j)$ of $i\in \nset{m}$ and  $j\in\nset{n}$, of a pair $(p_{i,j}, r_{i,j})$ of nonnegative real numbers. There pairs $(p_{i,j}, r_{i,j})$ are subject to the following conditions.
\begin{itemize}
 \item (Subnormality) $\sum_{j\in\nset{n}}p_{i,j}\le 1$ for each $i\in\nset{m}$.
 \item (Realizability) $p_{i,j}=0$ implies $r_{i,j}=0$.
\end{itemize}
\end{definition}

 An illustration is in \cref{fig:SroMCillust}. For an object $m$, each $i\in \nset{m}$ is identified with an open end, much like in $\roMC$ and $\roMDP$. For an arrow $f\colon m \to n$, the pair $f(i,j)=(p_{i,j}, r_{i,j})$ encodes a reachability probability and an expected reward, from an open end $i$ to $j$; together they represent a possible roMC behavior.

We go on to define the algebraic operations of $\semCatrmc$ as a TSMC. While there is a categorical description of $\semCatrmc$  using  a \emph{monad}~\cite{Moggi91}, we prefer a concrete definition here.
\iffull
See~\cref{sec:monad} for the categorical definition of $\semCatrmc$. 
\else
See~\cite[Appendix D]{Watanabe23long} for the categorical definition of $\semCatrmc$. 
\fi

\begin{definition}[sequential composition $\seqcomp$ of $\semCatrmc$]\label{def:seqCompSemCatrmc}
 Let  $f\colon m\rightarrow l$ and $g\colon l\rightarrow n$ be arrows in $\semCatrmc$. Their \emph{sequential composition} $f\seqcomp g\colon m\to n$ of $f$ and $g$ is defined as follows: letting
$f(i, j) = (p^{f}_{i, j}, r^{f}_{i, j})$ and $g(i, j) = (p^{g}_{i, j}, r^{g}_{i, j})$, then
$f\seqcomp g(i)\defeq(p^{f\seqcomp g}_{i,j},r^{f\seqcomp g}_{i,j})_{j\in [n]}$ is given by
\begin{align*}
         \bigl[\, p^{f\seqcomp g}_{i, j} \,\bigr]_{i\in\nset{m},j\in\nset{n}}  =
     \bigl[\,& p^{f}_{i, k}\,\bigr]_{i\in\nset{m},k\in\nset{l}} \cdot
          \bigl[\,p^{g}_{k, j}\,\bigr]_{k\in\nset{l},j\in\nset{n}},
   \\
   \bigl[\, r^{f\seqcomp g}_{i, j} \,\bigr]_{i\in\nset{m},j\in\nset{n}} =
     \bigl[\,& p^{f}_{i, k}\,\bigr]_{i\in\nset{m},k\in\nset{l}} \cdot
          \bigl[\,r^{g}_{k, j}\,\bigr]_{k\in\nset{l},j\in\nset{n}}\\ 
     &+
     \bigl[\, r^{f}_{i, k}\,\bigr]_{i\in\nset{m},k\in\nset{l}} \cdot
          \bigl[\,p^{g}_{k, j}\,\bigr]_{k\in\nset{l},j\in\nset{n}}. 
\end{align*}
\end{definition}
The sum $\oplus$ and the trace operator $\tr$ of $\semCatrmc$ are defined similarly. 
To define and prove axioms of the trace operator (\cref{fig:TSMCAxioms}), we exploit  the categorical theory of \emph{strong unique decomposition categories}~\cite{hoshino12}. 
\iffull
See~\cref{subsec:execution_formula}. 
\else
See~\cite[Appendix D]{Watanabe23long}. 
\fi
\begin{definition}[$\semCatrmc$ as a TSMC]\label{def:semCatrmcAsTSMC}
$\semCatrmc$ is a TSMC, with its operations 
 $ \seqcomp ,\oplus,\tr$. 
\end{definition}

Once we expand the above definitions to concrete terms, it is evident that they mirror the decomposition equalities. Indeed, the sequential composition $\seqcomp$ mirrors the first equalities in \cref{prop:MCeq,prop:MCERweq}. The same holds for the trace operator, too.
 Therefore, one can think of the above categorical development in \cref{def:seqCompSemCatrmc} and \cref{def:semCatrmcAsTSMC} as  a structured \emph{lifting} of the (local) equalities in \cref{prop:MCeq,prop:MCERweq} to the (global) categorical structures, as shown in \cref{fig:catsFunctors}. 

Once we found the semantic domain $\semCatrmc$, the following definition is easy.
\begin{definition}[$\semFunctorrmc$]
\label{def:semFuncRoMC}
The \emph{solution functor} $\semFunctorrmc\colon \roMC\to \semCatrmc$ is defined as follows. It carries an object $m$ (a natural number) to the same $m$; it carries an arrow $\mc{C}\colon m\to n$ in $\roMC$ to the arrow $\semFunctorrmc(\mc{C})\colon m\to n$ in $\semCatrmc$, defined by
 \begin{equation}\label{eq:defsemFuncRoMC}
     \semFunctorrmc(\mc{C})(i, j) \defeq  \bigl(\,\Reacha{i}{j}{\mc{C}},\,\ETRa{i}{j}{\mc{C}}\,\bigr),
 \end{equation} 
using reachability probabilities and expected rewards (\cref{def:openMDPReachProbExpCumRew}). 
\end{definition}

\begin{theorem}[$\semFunctorrmc$ is compositional]\label{thm:semFunctorrmcIsCompositional}
The correspondence $\semFunctorrmc$, defined in~\cref{eq:defsemFuncRoMC}, is a traced symmetric monoidal functor. That is, 
\begin{math}
 \semFunctorrmc(\mc{C} \seqcomp \mc{D})
=  
\semFunctorrmc(\mc{C}) \seqcomp 
\semFunctorrmc(\mc{D})
\end{math},
\begin{math}
 \semFunctorrmc(\mc{C}\oplus\mc{D})
=  
\semFunctorrmc(\mc{C})\oplus
\semFunctorrmc(\mc{D})
\end{math},
and
\begin{math}
 \semFunctorrmc(
\tr(\mc{E}))
=  
\tr(\semFunctorrmc(\mc{E}))
\end{math}.
Here $ \seqcomp ,\oplus,\tr$ on the left are from \cref{subsec:roMDPsAndTSMCStringDiagrams}; those on the right are from \cref{def:semCatrmcAsTSMC}. \qed
\end{theorem}

\begin{auxproof}
 As a theoretical remark whose understanding is unnecessary to read the paper,
 the semantic category $\semCatrmc$ can be seen as a full subcategory of the \emph{Kleisli category} of $T$~\cite{Moggi91}.
 Furthermore, it is a \emph{strong unique decomposition category} whose monoidal product is the coproduct, 
 and hence is a TSMC whose trace is given by the execution formula~\cite{hoshino12}. %
 \begin{proposition}
 \begin{enumerate}
 \item 
  Let $\kleisli{\Sets}{T}$ be the Kleisli category of $T$. The restriction $\rest{\kleisli{\Sets}{T}}{\Nat}$ to natural numbers as objects is a TSMC.
 \item 
  The restriction $\rest{\kleisli{\Sets}{T}}{\Nat}$ has the same objects and arrows as $\semCatrmc$. 
 \qed
\end{enumerate}
 \end{proposition}

 \begin{definition}
 We equip $\semCatrmc$ with the TSMC structure induced from $\rest{\kleisli{\Sets}{T}}{\Nat}$ by the equality $\rest{\kleisli{\Sets}{T}}{\Nat} = \semCatrmc$ as categories.
 \end{definition}

 The TSMC structure, introduced abstractly above, is illustrated as follows:
 \begin{itemize}
 \item for arrows $f:m\rightarrow l$ and $g:l\rightarrow n$ such that $f(i, j) = (p^{f}_{i, j}, r^{f}_{i, j})$ and $g(i, j) = (p^{g}_{i, j}, r^{g}_{i, j})$,
 their sequential composition $f\seqcomp g:m\rightarrow n$,
 for which let $(f\seqcomp g)(i,j) = (p^{f\seqcomp g}_{i, j}, r^{f\seqcomp g}_{i, j})$, is given by the following equations of matrices:
 \begin{align*}
     &     \bigl[\, p^{f\seqcomp g}_{i, j} \,\bigr]_{i\in\nset{m},j\in\nset{n}}  \;=\; 
     \bigl[\, p^{f}_{i, k}\,\bigr]_{i\in\nset{m},k\in\nset{l}} \cdot
          \bigl[\,p^{g}_{k, j}\,\bigr]_{k\in\nset{l},j\in\nset{n}}\ ,
   \\
   & \bigl[\, r^{f\seqcomp g}_{i, j} \,\bigr]_{i\in\nset{m},j\in\nset{n}}  \;=\; 
     \bigl[\, p^{f}_{i, k}\,\bigr]_{i\in\nset{m},k\in\nset{l}} \cdot
          \bigl[\,r^{g}_{k, j}\,\bigr]_{k\in\nset{l},j\in\nset{n}}\ 
     +
     \bigl[\, r^{f}_{i, k}\,\bigr]_{i\in\nset{m},k\in\nset{l}} \cdot
          \bigl[\,p^{g}_{k, j}\,\bigr]_{k\in\nset{l},j\in\nset{n}}\ 
     ,
 \end{align*}
 \item for arrow $f:l+m\rightarrow l+n$ such that $f(i, j) = (p^{f}_{i, j}, r^{f}_{i, j})$, its trace $\trace{l}{m}{n}{}(f)$, for which let $\trace{l}{m}{n}{}(f) = (p^{\trace{l}{m}{n}{}(f)}_{i, j}, r^{\trace{l}{m}{n}{}(f)}_{i, j})$, is given by the following equations of matrices:
 \begin{align*}
    &
   \bigl[\,p^{\trace{l}{m}{n}{}(f)}_{i, j} \,\bigr]_{i\in\nset{m},j\in\nset{n}}
    \;=\;    \bigl[\,
    p^{f}_{l+i, l+j}
   \,\bigr]_{i\in\nset{m},j\in\nset{n}}
 + 
  \\
  &\qquad
    \textstyle\sum_{d\in\Nat} 
   \bigl[\,
     p^{f}_{l+i, k}
   \,\bigr]_{i\in\nset{m},k\in\nset{l}}
 \cdot
   \bigl[\,
     p^{f}_{k, k'}
   \,\bigr]_{k\in\nset{l},k'\in\nset{l}}^{d}
 \cdot
   \bigl[\,
    p^{f}_{k', l+j}
   \,\bigr]_{k'\in\nset{l},j\in\nset{n}}\ ,\\
   & \bigl[\, r^{\trace{l}{m}{n}{}(f)}_{i, j} \,\bigr]_{i\in\nset{m},j\in\nset{n}}
    \;=\;    \bigl[\,
    r^{f}_{l+i, l+j} 
   \,\bigr]_{i\in\nset{m},j\in\nset{n}}
    +  \textstyle\sum_{d\in\Nat} A\cdot B^{d}\cdot C.
 \end{align*}
 Here $A, B,C $ are the following $m\times 2l$, $2l\times 2l$, $2l\times n$ matrices.
 \begin{align*}
 &\begin{aligned}
 A &= \begin{pmatrix}
 \bigl[\,
    p^{f}_{l+i, k} 
   \,\bigr]_{i\in\nset{m},k\in\nset{l}} & 
   \bigl[\,
    r^{f}_{l+i, k}
   \,\bigr]_{i\in\nset{m},k\in\nset{l}} \\
 \end{pmatrix}
 \\
 B &= \begin{pmatrix}
 \bigl[\,
    p^{f}_{k, k'} 
   \,\bigr]_{k\in\nset{l},k'\in\nset{l}} & \bigl[\,
    r^{f}_{k, k'} 
   \,\bigr]_{k\in\nset{l},k'\in\nset{l}}\\
   \bigl[\,
    0
   \,\bigr]_{k\in\nset{l},k'\in\nset{l}} &  \bigl[\,
    p^{f}_{k, k'} 
   \,\bigr]_{k\in\nset{l},k'\in\nset{l}}
 \end{pmatrix} 
 \\
 C &= \begin{pmatrix}
 \bigl[\,
    r^{f}_{k', l+j}
   \,\bigr]_{k'\in\nset{l},j\in\nset{n}}\\
   \bigl[\,
    p^{f}_{k', l+j}
   \,\bigr]_{k'\in\nset{l},j\in\nset{n}} \\
 \end{pmatrix}
 \end{aligned}
 \end{align*}
 \end{itemize}

 Note that the above formulae resemble those in \cref{prop:MCeq,prop:MCERweq},
 which is a key to show that the following theorem readily follows from the results of \cref{sec:compositionalAnalyssiOfOpenMarkovChains}.

 \begin{theorem}[$\semFunctorrmc$ is traced]
 \label{thm:coherence_cond_trace}
 The functor $\semFunctorrmc:\roMC\rightarrow \semCatrmc$ defined in the following is a traced symmetric monoidal functor:
 for an roMC $\mc{C}:m\rightarrow n$, $i\in \nset{m}$ and $j\in \nset{n}$, 
 \begin{align*}
     \semFunctorrmc(\mc{C})(i, j) \defeq  (\Reacha{i}{j}{\mc{C}},\ETRa{i}{j}{\mc{C}})\ .
 \end{align*}
 \qed
 \end{theorem}

 \ifdraft
 \newpage
 \fi

\end{auxproof}

\subsection{Semantic Category of Rightward Open MDPs}\label{subsec:semCatROMDPs}
We extend the theory in \cref{subsec:semCatROMCs} from MCs to MDPs (\cref{fig:catsFunctors}).
In particular,
on the semantics side, we have to  bundle up all possible behaviors of an MDP under different schedulers. We find that this is done systematically by \emph{change of base}~\cite{Eilenberg65,cruttwell2008normed}.
 We use the following notation for fixing scheduler $\tau$.

\begin{definition}[roMC $\indmc{\mdp{A}}{\tau}$ induced by $\mdp{A},\tau$]
Let $\mdp{A}:m\rightarrow n$ be a rightward open MDP and $\tau:Q^{\mdp{A}}\rightarrow A$ be a memoryless scheduler.  The \emph{rightward open MC} $\indmc{\mdp{A}}{\tau}$ \emph{induced by $\mdp{A}$ and $\tau$} is $(m,n,Q^{\mdp{A}}, \{\star\},E^{A}, P^{\indmc{\mdp{A}}{\tau}},R^{\mdp{A}})$, where for each $s\in Q$ and $t\in (\nset{\nfr{n}+\nfl{m}}+Q)$, $P^{\indmc{\mdp{A}}{\tau}}(s, \star, t) \defeq P^{\mdp{A}}(s, \tau(s), t)$.
\end{definition}

Much like in~\cref{subsec:semCatROMCs}, we first describe the semantic category $\semCatr$ in concrete terms. We later use the categorical machinery to define its algebraic structure.

\begin{definition}[objects and arrows of $\semCatr$]
\label{def:mor_in_sr}
 The \emph{category $\semCatr$} has natural numbers $m$ as objects. Its arrow $F\colon m \to n$ is given by a set  $\{f_{i}\colon m\to n \text{ in $\semCatrmc$}\}_{i\in I}$ of arrows
of the same type in $\semCatrmc$ ($I$ is an arbitrary index set).
\end{definition}

The above definition of arrows---collecting arrows in $\semCatrmc$, each of which corresponds to the behavior of $\indmc{\mdp{A}}{\tau}$ for each $\tau$---follows from the change of base construction (specifically
with the powerset functor $\pMnd$ on the category $\Sets$ of sets). Its general theory
gives sequential composition $ \seqcomp $ for free (concretely described in~\cref{def:seqc_in_sr}), together with equational axioms. 
\iffull
See~\cref{subsec:change_of_base} for details.
\else 
See~\cite[Appendix D]{Watanabe23long}.
\fi
Sum $\oplus$ and trace $\tr$ are not covered by general  theory, but  we can define them analogously to $\seqcomp$ in the current setting. Thus, for $\oplus$ and $\tr$ as well, we are using change of base as an inspiration.

Here is a concrete description of algebraic operations. It applies the corresponding operation of $\semCatrmc$ in the elementwise manner.

\begin{definition}[$\seqcomp,\oplus,\tr$ in $\semCatr$]
\label{def:seqc_in_sr}
Let $F:m\rightarrow l$, $G:l\rightarrow n$, $H:l+m\rightarrow l+n$ be arrows in $\semCatr$.
Their \emph{sequential composition} $F\seqcomp G$ 
of $F$ and $G$ 
is given by $F\seqcomp G\defeq \{ f\seqcomp g\mid f\in F,\  g\in G\}$ where $f\seqcomp g$ is the sequential composition of $f$ and $g$
in $\semCatrmc$. The \emph{trace} $\trace{l}{m}{n}{}(H):m\rightarrow n$ of $H$ with respect to $l$ is given by $\trace{l}{m}{n}{}(H) \defeq \{ \trace{l}{m}{n}{}(h) \mid h\in H\}$ where $ \trace{l}{m}{n}{}(h)$ is the trace of $h$ with respect to $l$ in $\semCatrmc$.

\emph{Sum} $\oplus$ in $\semCatr$ is defined analogously, applying the operation in $\semCatrmc$ elementwise. 
\iffull
See \cref{apnd:completeDefOfSemCatr}.
\else 
See \cite[Appendix A]{Watanabe23long} for details.
\fi
\end{definition}

\begin{auxproof}
 On the equational axioms of TSMCs (\cref{fig:TSMCAxioms}), those regarding $\seqcomp$ follow from the general theory (see e.g.~\cite{laird2017qualitative});
 the rest are routine.\footnote{Special categorical settings under which the remaining axioms are guaranteed are studied in~\cite{laird2017qualitative}.  This is driven towards game semantics and  not suited for our purpose.}
\end{auxproof}
\begin{theorem}
 $\semCatr$ is a TSMC. \qed
\end{theorem}

We now define a solution functor and prove its compositionality.

\begin{definition}[$\semFunctorr$]
\label{def:semFuncR}
 The \emph{solution functor} $\semFunctorr \colon \roMDP\to \semCatr$ is defined as follows. It carries an object $m\in\Nat$ to $m$, and an arrow $\mdp{A}\colon m\to n$ in $\roMDP$ to $\semFunctorr(\mdp{A})\colon m\to n$ in $\semCatr$. The latter is defined in the following elementwise manner, using $\semFunctorrmc$ in \cref{def:semFuncRoMC}. 
 \begin{equation}\label{eq:defsemFuncR}
     \semFunctorr(\mdp{A}) \;\defeq\; \bigl\{ \semFunctorrmc(\indmc{\mdp{A}}{\tau})\,\big|\, \tau:Q^{\mdp{A}}\rightarrow A\text{ a (memoryless) scheduler}\bigr\}.
 \end{equation}
\end{definition}

\begin{theorem}[compositionality]\label{thm:semFunctorrIsCompositional}
The correspondence $\semFunctorr \colon \roMDP\to \semCatr$ 
is a traced symmetric monoidal functor, preserving $\seqcomp,\oplus,\tr$ as in \cref{thm:semFunctorrmcIsCompositional}. 
\qed
\end{theorem}

\begin{remark}[memoryless schedulers]\label{rem:memoryless}
 Our restriction to memoryless schedulers (cf.\ \cref{def:path}) plays a crucial role in the proof of \cref{thm:semFunctorrIsCompositional}, specifially for the trace operator (i.e.\ loops, cf.\ \cref{fig:trace_of_rmdp}). Intuitively, a \emph{memoryful} scheduler for a loop may
act differently in different iterations. Its technical consequence is that the elementwise definition of $\tr$, as in \cref{def:seqc_in_sr}, no longer works for memoryful schedulers. 
\begin{auxproof}
 It is known that memoryless schedulers do not suffice for  other objectives.
 A possible workaround is to enrich the categorical framework with memories, as done e.g.\ in~\cite{HoshinoMH14CSLLICSNoCrossref}. This is future work.
\end{auxproof}
\end{remark}

\begin{auxproof}

\begin{remark}[$\roMC$ and $\roMDP$]\label{rem:roMDPIsNotByChangeOfBase}
Unlike the relationship between the first and second rows in \cref{fig:catsFunctors} (where everything is transferred by Int), in the passage from the third to the second, we use change of base only on the semantic side. 

We can indeed apply change of base to  $\roMC$ on the system side, but this yields a category where an arrow is a bunch of roMCs. Such data is not itself an roMDP $\mdp{A}$ but rather 
$\bigl\{\,\indmc{\mdp{A}}{\tau}\,\big|\, \tau \text{ is a scheduler}\,\}$. 

Nevertheless, the value of change of base as a guiding principle on the semantic side is significant, as we demonstrated above.
\end{remark}

 Here, we define the semantic category $\semCatr$ of rightward open MDPs by exploiting the categorical structure of $\semCatrmc$. The intuition is that for each memoryless scheduler $\tau$ in rightward open MDP $\mdp{A}$, $\tau$ induces rightward open MC $\indmc{\mdp{A}}{\tau}$, and the result of a rightward open MDP is a set of rightward open MCs that are induced by memoryless schedulers.

 MCs by the well-known categorical construction called \emph{change of base}~\cite{Eilenberg65,kelly1982basic}, which is applied for construction for quantitative models of \emph{game semantics}~\cite{journals/iandc/HylandO00} in~\cite{laird2017qualitative}.

 \begin{definition}[arrow in $\semCatr$]
 \myindent\label{def:mor_in_sr}
 Let $m, n\in \Nat$.
 An \emph{arrow} $F$ from $m$ to $n$ (in $\semCatr$), denoted by $F:m\rightarrow n$, is a subset of $\semCatrmc(m, n)$. The set arrows from $m$ to $n$ in $\semCatr$ is denoted by $\semCatr(m, n)$.
 \end{definition}

 \begin{definition}[sequential compositon $\seqcomp$ in $\semCatr$]
 \myindent\label{def:seqc_in_sr}
 Let $F:m\rightarrow l$, $G:l\rightarrow n$ be arrows in $\semCatr$.
 The sequential composition $F\seqcomp G$ of $F$ and $G$ is given by $F\seqcomp G\defeq \{ f\seqcomp g\mid f\in F,\  g\in G\}$ where $f\seqcomp g$ is the sequential composition %
 in $\semCatrmc$.
 \end{definition}

 The sum $\oplus$, trace operator $\tr$, identity $\id$, and swap $\swap[]{}{}$ in $\semCatr$ are defined similarly in the natural manner. The complete definition is found in~\cref{apnd:completeDefOfSemCatr}.
 The construction of~\cref{def:mor_in_sr,def:seqc_in_sr} is known as \emph{change of base}~\cite{Eilenberg65,kelly1982basic}.

 \begin{definition}[induced rightward open MC $\indmc{\mdp{A}}{\tau}$]\myindent
 Let $\mdp{A}:m\rightarrow n$ be a rightward open MDP and $\tau:Q^{\mdp{A}}\rightarrow A$ be a memoryless scheduler. We define the \emph{induced rightward open MC} $\indmc{\mdp{A}}{\tau}$ as $(m,n,Q^{\mdp{A}}, \{\star\},E^{A}, P^{\indmc{\mdp{A}}{\tau}},R^{\mdp{A}})$, where for each $s\in Q$ and $t\in (\nset{\nfr{n}+\nfl{m}}+Q)$, $P^{\indmc{\mdp{A}}{\tau}}(s, \star, t) \defeq P^{\mdp{A}}(s, \tau(s), t)$.
 \end{definition}

 \begin{proposition}[$\semCatr, \semFunctorr$ are traced]
 \label{prop:fsemmdp_TSMC}
 We have a traced symmetric monoidal category $(\semCatr,\seqcomp ,\oplus, \tr, \id, \swap[]{}{})$, whose objects are natural numbers.

 Furthermore, the functor $\semFunctorr:\roMDP\rightarrow \semCatr$ defined in the following is a traced symmetric monoidal functor: for a rightward open MDP $\mdp{A}:m\rightarrow n$, %
 \begin{align*}
     \semFunctorr(\mdp{A}) \defeq \{ \semFunctorrmc(\indmc{\mdp{A}}{\tau})\mid \tau:Q^{\mdp{A}}\rightarrow A\text{ is a memoryless scheduler of }\mdp{A}\}.
 \end{align*}
 \qed
\end{proposition}

\end{auxproof}

\subsection{Semantic Category of MDPs}
Finally, we extend from (unidirectional) roMDPs to (bidirectional) oMDPs (i.e.\ from the second  to the first row in \cref{fig:catsFunctors}). The system-side construction is already presented in~\cref{subsec:string_diagrams_openMDPs}; the semantical side, described here, follows the same $\Int$ construction~\cite{joyal1996}. The common intuition is that of twists, see~\cref{fig:twist}.

\begin{definition}[the semantic category $\semCat$]
\label{def:mor_in_s}
We define $\semCat=\Int(\semCatr)$. Concretely, its objects are pairs $(\nfr{m},\nfl{m})$ of natural numbers. Its arrows are given by arrows of $\semCatr$ as follows:
\begin{equation}\label{eq:arrowsOfSemCat}
 \vcenter{\infer={ 
   \text{an arrow } F\colon \nfr{m}+\nfl{n}\longrightarrow  \nfr{n}+\nfl{m}
   \text{ in $\semCatr$}
   }{
   \text{an arrow } F\colon (\nfr{m}, \nfl{m}) \longrightarrow (\nfr{n}, \nfl{n}) 
   \text{ in $\semCat$}
   }}
\end{equation}

By  general properties of  $\Int$, $\semCat$ is a compact closed category (compCC).
\end{definition}

The $\Int$ construction applies not only to categories but also to functors.
\begin{definition}[$\semFunctor$]
The \emph{solution functor} $\semFunctor \colon \oMDP\to \semCat$ is defined by $\semFunctor=\Int(\semFunctorr)$. 
\end{definition}

The following is our main theorem. 
\begin{theorem}[the solution  $\semFunctor$ is compositional]
 The solution functor $\semFunctor \colon \oMDP\to \semCat$ is a compact closed functor, preserving operations $\seqcomp, \oplus$ as in
\begin{equation}\small
 \begin{array}{l}
  \semFunctor(\mdp{A} \seqcomp \mdp{B})
 =  
 \semFunctor(\mdp{A}) \seqcomp 
 \semFunctor(\mdp{B})
 ,\quad
 \semFunctor(\mdp{A}\oplus\mdp{B})
 =  
 \semFunctor(\mdp{A})\oplus
 \semFunctor(\mdp{B})
 .
\end{array}
\tag*{\qed}
\end{equation}

\end{theorem}

We can  easily confirm, from \cref{def:semFuncRoMC,def:semFuncR}, that   $\semFunctor$  computes the solution  we want.  Given an open MDP $\mathcal{A}$, an entrance $i$ and an exit $j$, $\semFunctor$  returns the set
\begin{equation}\label{eq:concreteOutputOfSemFunc}\small
\begin{array}{l}
 \bigl\{\,\bigl(\,\Reacha{i}{j}{\indmc{\mdp{A}}{\tau}},\,\ETRa{i}{j}{\indmc{\mdp{A}}{\tau}}\,\bigr)
 \,\big|\,
 \text{$\tau$ is a memoryless scheduler} 
 \,\bigr\}
\end{array}
\end{equation} of pairs of a reachability probability and expected reward, under different schedulers, in a passage from $i$ to $j$. 

\begin{remark}[synthesizing an optimal scheduler]\label{rem:synthesis}
The compositional solution functor $\semFunctor$ abstracts away schedulers and only records their results (see~\cref{eq:concreteOutputOfSemFunc} where $\tau$ is not recorded). At the implementation level, we can explicitly record schedulers so that our compositional algorithm also synthesizes an optimal scheduler. We do not do so here for theoretical simplicity.

\end{remark}

\begin{auxproof}

 \begin{definition}[arrow in $\semCat$]
 \myindent\label{def:mor_in_s}
 Let $(\nfr{m},\nfl{m}), (\nfr{n},\nfl{n})\in \Nat\times \Nat$.
 An \emph{arrow} $F$ from $(\nfr{m}, \nfl{m})$ to $(\nfr{n}, \nfl{n})$ (in $\semCat$), denoted by $F:(\nfr{m}, \nfl{m})\rightarrow (\nfr{n}, \nfl{n})$, is an arrow in $\semCatrmc(\nfr{m}+\nfl{n}, \nfr{n}+\nfl{m})$. The set of arrows from $(\nfr{m}, \nfl{m})$ to $(\nfr{n}, \nfl{n})$ in $\semCat$ is denoted by $\semCat\big((\nfr{m}, \nfl{m}), (\nfr{n}, \nfl{n})\big)$.
 \end{definition}

 By using the TSMC structure of $\semCatrmc$, the sequential compositon $\seqcomp$, sum $\oplus$, identity $\id$, swap $\swap[]{}{}$, unit $\dunit$, counit $\dcounit$ in $\semCatr$ are defined  in the same manner as~\cref{subsec:string_diagrams_openMDPs}, namely by the $\Int$ construction. The complete definitions are found in~\cref{apnd:completeDefOfSemCat}.

 \begin{theorem}[$\semCat,\ \semFunctor$ are compact closed]
 \label{thm:coherence_cond_compCC}
 We have a compact closed category
 $(\semCat, \seqcomp, \oplus, \id, \swap[]{}{}, \dunit, \dcounit)$, whose objects are pairs of natural numbers.

 Furthermore, the functor $\semFunctor:\oMDP\rightarrow \semCat$ defined in the following is a compact closed functor: for an arrow $\mdp{A}:(\nfr{m}, \nfl{m})\rightarrow (\nfr{n}, \nfl{n})$ in $\oMDP$, 
 \begin{align*}
     \semFunctor(\mdp{A}) &\defeq  \semFunctorr(\mdp{A}),
 \end{align*}
 where note that by definition $\mdp{A}$ is also an arrow from $\nfr{m}+\nfl{n}$ to $\nfr{n}+\nfl{m}$ in $\roMDP$.
 \qed
 \end{theorem}

\end{auxproof}

\section{Implementation and Experiments}\label{sec:impl}
\paragraph{Meager Semantics}
Since our problem is to compute optimal expected rewards, in our compositional algorithm, we can ignore those intermediate results which are \emph{totally subsumed} by other results (i.e.\ those which come from clearly suboptimal schedulers). This notion of \emph{subsumption} is formalized as an order $\le$ between parallel arrows in $\semCatrmc$ (cf.\ \cref{def:semCatrmc}):  $(p_{i,j},r_{i,j})_{i,j} \le (p'_{i,j},r'_{i,j})_{i,j}$ if $p_{i,j}\le p'_{i,j}$ and  $r_{i,j}\le r'_{i,j}$ for each $i,j$. Our implementation works with this \emph{meager semantics} for better performance; specifically, it removes elements of $\semFunctorr(\mdp{A})$ in~\cref{eq:defsemFuncR} that are subsumed by others. It is possible to formulate this meager semantics as categories and functors, compare it with the semantics in \cref{sec:FatSemanticCategories}, and prove its correctness. We defer it to another venue for lack of space. 
\paragraph{Implementation}
We implemented the compositional solution functor $\semFunctor\colon \oMDP\allowbreak\to \semCat$, using the meager semantics as discussed. This prototype implementation is in Python and called $\ocf$. 

$\ocf$ takes a string diagram $\mdp{A}$ of open MDPs as input;
they are expressed in a textual format that uses operations $\seqcomp,\oplus$ (such as the textual expression in \cref{def:seqCompOpenMDPs}). Note that  we are abusing notations here, identifying a string diagram of oMDPs and the composite oMDP $\mdp{A}$ denoted by it. 

Given such input $\mdp{A}$, $\ocf$ returns the arrow $\semFunctor(\mdp{A})$, which is concretely given by pairs of a reachability probability and expected reward shown in~\cref{eq:concreteOutputOfSemFunc} (we have suboptimal pairs removed, as discussed above). Since different pairs correspond to different schedulers, we choose a pair in which the expected reward is the greatest. This way we answer the optimal expected reward problem.

\paragraph{Freezing} In the input format of $\ocf$, we have an additional \emph{freeze} operator: any expression inside it is considered monolithic, and thus $\ocf$ does not solve it compositionally. Those frozen oMDPs---i.e.,\ those expressed by frozen expressions---are solved by PRISM~\cite{kwiatkowska2011prism} in our implementation.

 Freezing allows us to choose how deep---in the sense of the nesting  of string diagrams---we go compositional. For example, when a component oMDP $\mdp{A}_{0}$ is small but has many loops, fully compositional model checking of $\mdp{A}_{0}$ can be more expensive than (monolithic) PRISM. Freezing is useful in such situations. 

We have found experimentally that the degree of freezing often should not be extremal (i.e.\ none or all). The optimal degree, which should be thus somewhere intermediate, is not known a priori.

 However, there are not too many options (the number of layers in compositional model description), and freezing a half is recommended, both from our experience and for the purpose of binary search.

We require that a frozen oMDP should have a unique exit. Otherwise, an oMDP with a specified exit can have the reachability probability $<1$, in which case PRISM returns $\infty$ as the expected reward. The last is different from our definition of expected reward (\cref{rem:notAlmostSureTerm}).

\paragraph{Research Questions}
We posed the following questions.
\begin{description}
 \item[RQ1] Does the compositionality of $\ocf$ help  improve performance?
 \item[RQ2] How much do we benefit from freezing, i.e.,\ a feature that allows us to choose the degree of compositionality?
 \item[RQ3]  What is the absolute performance of $\ocf$?
 \item[RQ4] Does the formalism of string digrams accommodate real-world models, enabling their compositional model checking?
 \item[RQ5] On which (compositional) models does $\ocf$ work well?
\end{description}

\paragraph{Experiment Setting}
We conducted experiments  on Apple 2.3 GHz Dual-Core Intel Core i5 with 16GB of RAM. 
We designed three benchmarks, called Patrol, Wholesale, and Packets, as string diagrams of MDPs. 
Patrol is sketched in \cref{fig:string_diagrams_of_mdps}; it has layers of \emph{tasks}, \emph{rooms}, \emph{floors}, \emph{buildings} and a \emph{neighborhood}.

Wholesale is similar to Patrol, with four layers (\emph{item}, \emph{dispatch}, \emph{pipeline}, \emph{wholesale}), but their transition structures are more complex: they have more loops, and more actions are enabled in each position, compared to Patrol. The lowest-level component MDP is much larger, too: an \emph{item} in Wholesale has 5000 positions, while a \emph{task} in Patrol has a unique position.

Packets has  two layers: the lower layer models a transmission of 100 packets with probabilistic failure. The upper layer is a sequence of copies of 2--5 variations of the lower layer---in total, we have 50 copies---modeling  50 batches of packets.

For Patrol and Wholesale, we conducted experiments with varying \emph{degree of identification (DI)}; this can be seen as an ablation study. These benchmarks have  identical copies of a component MDP in their string diagrams; high DI  means that these copies are indeed expressed as multiple occurrences of the same variable, informing $\ocf$ to reuse the intermediate solution. As DI goes lower, we introduce new variables for these copies and let them look different to $\ocf$. Specifically, we have twice as many variables for DI-mid, and three (Patrol) or four (Wholesale) times as many for DI-low, as for DI-high.

For Packets, we conducted experiments with different degrees of  freezing (FZ). FZ-none indicates no freezing, where our compositional algorithm digs all the way down to individual positions as component MDPs. FZ-all freezes everything, which means we simply used PRISM (no compositionality). FZ-int.\ (\emph{intermediate}) freezes the lower   of the two layers. Note that this includes the performance comparison between CompMDP and PRISM (i.e.\ FZ-all).

For Patrol and Wholesale, we also compared the performance of CompMDP and PRISM using their simple variations Patrol5 and Wholesale5. We did not use other variations (Patrol/Wholesale1--4) since the translation of the models to the PRISM format blowed up.

\begin{table}[tbp]
    \caption{Experimental results. 
}
    \label{tab:experiment_results}
    \begin{minipage}[t]{.45\textwidth}
    \centering\scriptsize
    \scalebox{0.9}{
    \begin{tabular}[t]{lrrrrr}
        \toprule
         &&&
         \multicolumn{3}{c}{\text{exec.\ time [s]}}
        \\\cmidrule(lr){4-6}
        benchmark & $|Q|$ & $|\totalAction|$  & 
        DI-high &
        DI-mid &
        DI-low 
        \\
        \midrule Patrol1    & $10^8$ &   $10^{8}$  & 21 & 42 & 83 \\   
        Patrol2    & $10^8$  &   $10^8$    & 23 & 48 & 90\\
        Patrol3    & $10^9$  &          $10^9$      &  22  &  43  &  89  \\   
        Patrol4    &   $10^9$        &    $10^9$  &  30   &  60  &   121 \\
        \midrule Wholesale1    &  $10^8$       &  $2\cdot 10^8$             &  130 &  260  &  394 \\   
        Wholesale2    &   $10^8$       &    $2\cdot 10^8$   & 92    & 179   & 274  \\  
        Wholesale3    &   $2\cdot 10^8$    &  $4\cdot 10^8$      & 6    &  12 & 23  \\  
        Wholesale4    &    $2\cdot 10^8$    &  $4\cdot 10^8$               &  129   & 260   & 393 \\
        \bottomrule
    \end{tabular}
    }
    \end{minipage}
    \hfill
    \begin{minipage}[t]{.5\textwidth}\scriptsize
    \scalebox{0.9}{
    \begin{tabular}[t]{lrrrrr}
        \toprule 
             &  &&
         \multicolumn{3}{c}{\text{exec.\ time [s]}}
       \\\cmidrule(lr){4-6}
        benchmark & $|Q|$ & $|\totalAction|$ &
          FZ-none &   FZ-int.  & FZ-all\\
        &&&&&(PRISM)\\
        \midrule Packets1    & $2.5\cdot 10^5$  & $5\cdot 10^5$& TO         &     1       & 65\\
        Packets2    &  $2.5\cdot 10^5$  & $5\cdot 10^5$& TO         &     3       & 64\\
        Packets3    & $2.5\cdot 10^5$  & $5\cdot 10^5$&  TO         &     1       & 56\\
        Packets4    & $2.5\cdot 10^5$  & $5\cdot 10^5$&  TO         &     3       & 56\\
        Patrol5 &    $10^8$ &    $10^8$   &      $22$    &  $22$             &  TO\\
        Wholesale5 & $5\cdot 10^7$  &   $10^8$   &  TO   
            &   14  & TO  \\
        \bottomrule
    \end{tabular}
    }
\begin{minipage}{\textwidth}
  $|Q|$ is the number of positions; $|\totalAction|$ is the number of transitions (only counting action branching,  not probabilistic branching);  execution time is the average of five runs, in sec.;  timeout (TO) is 1200 sec.
\end{minipage}

\end{minipage}

\end{table}

\paragraph{Results and Discussion}
\cref{tab:experiment_results} summarizes the experiment results. 

\textbf{RQ1}. 
A big advantage of compositional verification is that it can reuse intermediate results. This advantage is clearly observed in the ablation experiments with the benchmarks Patrol1--4 and Wholesale1--4: as the degree of reuse goes 1/2 and 1/3--1/4 (see above), the execution time grew inverse-proportionally. Moreover, with the benchmarks Packets1--4, Patrol5 and Wholesale5, we see that 
compositionality greatly improves performance, compared to PRISM (FZ-all). Overall, we can say that compositionality has clear performance advantages in probabilistic model checking.

\textbf{RQ2}. 
The Packets experiments show that controlling the degree of compositionality is important. Packet's lower layer (frozen in FZ-int.) is a large and complex model, without a clear compositional structure;  its fully compositional treatment turned out to be prohibitively expensive. The performance advantage of FZ-int.\ compared to PRISM (FZ-all) is encouraging. The Patrol5 and Wholesale5 experiments also show the advantage of compositionality.

\begin{auxproof}
 Nice to explore symbolic algorithms, as an enemy and as a friend (can we improve $\ocf$ by symbolic representation?)
\end{auxproof}

\textbf{RQ3}. 
We find the absolute performance of $\ocf$ quite satisfactory. The Patrol and Wholesale benchmarks are huge models, with so many positions that fitting their explicit state representation in memory is already nontrivial.   $\ocf$, exploiting their succinct presentation by string diagrams, successfully model-checked them in realistic time (6--130 seconds with DI-high).

\textbf{RQ4}. 
The experiments suggest that string diagrams are a practical modeling formalism, allowing faster solutions of realistic benchmarks. It seems likely that the formalism is more suited for \emph{task compositionality} (where components are sub-\emph{tasks} and they are sequentially composed with possible fallbacks and loops) rather than \emph{system compositionality} (where components are sub-\emph{systems} and they are parallelly composed).

\textbf{RQ5}. 
It seems that  the number of locally optimal schedulers is an important factor: if there are many of them, then we have to record more in the intermediate solutions of the meager semantics. This number typically increases when more actions are available, as the comparison between Patrol and Wholesale.

\begin{auxproof}

 \subsection{Evaluation}
 We evaluate $\ocf$ with some research questions.
 \subsubsection{RQ1: Does compositionality improve the performance of probabilistic model checking?}
 \cref{tab:experiment_results_patrol} and~\cref{tab:experiment_results_wholesale} show that $\ocf$ solves the benchmarks faster as the number of occurrences of copies becomes larger. We observe that the compositionality of $\ocf$ successfully exploits the occurrence of copies in the benchmarks.
 \subsubsection{RQ2: How much do we benefit from the flexibility of our framework, which allows us to choose the level of compositionality we exploit?}
 \cref{tab:experiment_results_packets} shows that the hybrid $\ocf$ is the fastest approach in pure $\ocf$, hybrid $\ocf$ and PRISM. We observe that the flexibility of our framework is effective if we carefully design the level of compositionality.

 \subsubsection{RQ3: Overall, what is the absolute performance of our compositional probabilistic model checking?}

 \subsubsection{RQ4: What about practical applicability? Do we have many compositional models to which our method will apply?}

 Our benchmarks are realistic and we believe there are many others. That said, a modeling language that allows users to express their problems in a way compositionality is explicit and exploitable, is desired. It is our future work.

 \subsubsection{RQ5: On which (compositional) models does our method perform well?}
 One of the deciding factors is the number of locally optimal schedulers. This is because an arrow in the semantic category becomes larger in proportion to the number of optimal schedulers. We see that the number of optimal schedulers tends to increase as (1) the number of actions increases, and (2) the size of domains and codomains increases.
\end{auxproof}

\bibliographystyle{splncs04}
\bibliography{main}

\begin{thebibliography}{10}
\providecommand{\url}[1]{\texttt{#1}}
\providecommand{\urlprefix}{URL }
\providecommand{\doi}[1]{https://doi.org/#1}

\bibitem{BaierKatoen08}
Baier, C., Katoen, J.: Principles of Model Checking. {MIT} Press (2008)

\bibitem{Baier0KW17}
Baier, C., Klein, J., Kl{\"{u}}ppelholz, S., Wunderlich, S.: Maximizing the
  conditional expected reward for reaching the goal. In: Legay, A., Margaria,
  T. (eds.) Tools and Algorithms for the Construction and Analysis of Systems -
  23rd International Conference, {TACAS} 2017, Held as Part of the European
  Joint Conferences on Theory and Practice of Software, {ETAPS} 2017, Uppsala,
  Sweden, April 22-29, 2017, Proceedings, Part {II}. Lecture Notes in Computer
  Science, vol. 10206, pp. 269--285 (2017). \doi{10.1007/978-3-662-54580-5\_16}

\bibitem{BonchiHPSZ19}
Bonchi, F., Holland, J., Piedeleu, R., Sobocinski, P., Zanasi, F.: Diagrammatic
  algebra: from linear to concurrent systems. Proc. {ACM} Program. Lang.
  \textbf{3}({POPL}),  25:1--25:28 (2019). \doi{10.1145/3290338}

\bibitem{ClarkeLM89}
Clarke, E.M., Long, D.E., McMillan, K.L.: Compositional {M}odel {C}hecking. In:
  Proceedings of the Fourth Annual Symposium on Logic in Computer Science
  {(LICS} '89), Pacific Grove, California, USA, June 5-8, 1989. pp. 353--362.
  {IEEE} Computer Society (1989). \doi{10.1109/LICS.1989.39190}

\bibitem{cruttwell2008normed}
Cruttwell, G.S.: Normed spaces and the change of base for enriched categories.
  Ph.D. thesis, Dalhousie University (2008)

\bibitem{Eilenberg65}
Eilenberg, S., Kelly, G.M.: Closed categories. In: Proceedings of the
  Conference on Categorical Algebra: La Jolla 1965. pp. 421--562. Springer
  (1966)

\bibitem{girard1989geometry}
Girard, J.Y.: Geometry of interaction {I}: Interpretation of {S}ystem {F}. In:
  Studies in Logic and the Foundations of Mathematics, vol.~127, pp. 221--260.
  Elsevier (1989)

\bibitem{haghverdiS06}
Haghverdi, E., Scott, P.J.: A categorical model for the geometry of
  interaction. Theor. Comput. Sci.  \textbf{350}(2-3),  252--274 (2006).
  \doi{10.1016/j.tcs.2005.10.028}

\bibitem{heunen2019categories}
Heunen, C., Vicary, J.: Categories for Quantum Theory: An Introduction. Oxford
  University Press (2019)

\bibitem{hoshino12}
Hoshino, N.: A representation theorem for unique decomposition categories. In:
  Berger, U., Mislove, M.W. (eds.) Proceedings of the 28th Conference on the
  Mathematical Foundations of Programming Semantics, {MFPS} 2012, Bath, UK,
  June 6-9, 2012. Electronic Notes in Theoretical Computer Science, vol.~286,
  pp. 213--227. Elsevier (2012). \doi{10.1016/j.entcs.2012.08.014}

\bibitem{joyal1996}
Joyal, A., Street, R., Verity, D.: Traced monoidal categories. Mathematical
  Proceedings of the Cambridge Philosophical Society  \textbf{119}(3),
  447--468 (1996)

\bibitem{JungesS22}
Junges, S., Spaan, M.T.J.: Abstraction-refinement for hierarchical
  probabilistic models. In: Shoham, S., Vizel, Y. (eds.) Computer Aided
  Verification - 34th International Conference, {CAV} 2022, Haifa, Israel,
  August 7-10, 2022, Proceedings, Part {I}. Lecture Notes in Computer Science,
  vol. 13371, pp. 102--123. Springer (2022). \doi{10.1007/978-3-031-13185-1\_6}

\bibitem{kelly1982basic}
Kelly, M.: Basic concepts of enriched category theory, vol.~64. CUP Archive
  (1982)

\bibitem{khovanov2002functor}
Khovanov, M.: A functor-valued invariant of tangles. Algebraic \& Geometric
  Topology  \textbf{2}(2),  665--741 (2002)

\bibitem{kwiatkowska2011prism}
Kwiatkowska, M.Z., Norman, G., Parker, D.: {PRISM} 4.0: Verification of
  probabilistic real-time systems. In: Gopalakrishnan, G., Qadeer, S. (eds.)
  Computer Aided Verification - 23rd International Conference, {CAV} 2011,
  Snowbird, UT, USA, July 14-20, 2011. Proceedings. Lecture Notes in Computer
  Science, vol.~6806, pp. 585--591. Springer (2011).
  \doi{10.1007/978-3-642-22110-1\_47}

\bibitem{KwiatkowskaNPQ13}
Kwiatkowska, M.Z., Norman, G., Parker, D., Qu, H.: Compositional probabilistic
  verification through multi-objective model checking. Inf. Comput.
  \textbf{232},  38--65 (2013). \doi{10.1016/j.ic.2013.10.001}

\bibitem{MacLane2}
Mac~Lane, S.: Categories for the Working Mathematician. Springer, second edn.
  (1998)

\bibitem{Moggi91}
Moggi, E.: Notions of computation and monads. Inf. Comput.  \textbf{93}(1),
  55--92 (1991). \doi{10.1016/0890-5401(91)90052-4}

\bibitem{PiedeleuKCS15}
Piedeleu, R., Kartsaklis, D., Coecke, B., Sadrzadeh, M.: Open system
  categorical quantum semantics in natural language processing. In: Moss, L.S.,
  Sobocinski, P. (eds.) 6th Conference on Algebra and Coalgebra in Computer
  Science, {CALCO} 2015, June 24-26, 2015, Nijmegen, The Netherlands. LIPIcs,
  vol.~35, pp. 270--289. Schloss Dagstuhl - Leibniz-Zentrum f{\"{u}}r
  Informatik (2015). \doi{10.4230/LIPIcs.CALCO.2015.270}

\bibitem{QuatmannD0JK16}
Quatmann, T., Dehnert, C., Jansen, N., Junges, S., Katoen, J.: Parameter
  synthesis for {M}arkov models: Faster than ever. In: Artho, C., Legay, A.,
  Peled, D. (eds.) Automated Technology for Verification and Analysis - 14th
  International Symposium, {ATVA} 2016, Chiba, Japan, October 17-20, 2016,
  Proceedings. Lecture Notes in Computer Science, vol.~9938, pp. 50--67 (2016).
  \doi{10.1007/978-3-319-46520-3\_4}

\bibitem{DBLP:conf/csl/TsukadaO14}
Tsukada, T., Ong, C.L.: Compositional higher-order model checking via
  \emph{{\(\omega\)}}-regular games over {B}{\"{o}}hm trees. In: Joint Meeting
  of the Twenty-Third {EACSL} Annual Conference on Computer Science Logic
  {(CSL)} and the Twenty-Ninth Annual {ACM/IEEE} Symposium on Logic in Computer
  Science (LICS), {CSL-LICS} '14, Vienna, Austria, July 14 - 18, 2014. pp.
  78:1--78:10. {ACM} (2014)

\bibitem{Watanabe21}
Watanabe, K., Eberhart, C., Asada, K., Hasuo, I.: A compositional approach to
  parity games. In: Sokolova, A. (ed.) Proceedings 37th Conference on
  Mathematical Foundations of Programming Semantics, {MFPS} 2021, Hybrid:
  Salzburg, Austria and Online, 30th August - 2nd September, 2021. {EPTCS},
  vol.~351, pp. 278--295 (2021). \doi{10.4204/EPTCS.351.17}

\end{thebibliography}
\iffull
\appendix

\section{Omitted Definitions and Propositions}
\subsection{Definitions on $\roMDP$}
\begin{definition}[isomorphism of roMDPs]\label{def:isomorphismRoMDPs}
 Let 
 $\mdp{A} = \bigl(m,n,Q^{\mdp{A}}, 
 A,E^{\mdp{A}},P^{\mdp{A}},R^{\mdp{A}}\bigr)$ and 
 $\mdp{B} = \bigl(m,n,Q^{\mdp{B}}, 
 A,E^{\mdp{B}},P^{\mdp{B}},R^{\mdp{B}}\bigr)$ be  rightward open MDPs; we assume they have the same action set $A$ and the same arity $(m,n)$. 
 An \emph{isomorphism} from $\mdp{A}$ to $\mdp{B}$ is a bijection $\eta:Q^{\mdp{A}}\rightarrow Q^{\mdp{B}}$ that preserves the MDP structure, that is,  
  (i) for each $i
 $, $\eta(E^{\mdp{A}}(i)) = E^{\mdp{B}}(i)$,
  (ii) for each $(s, a ,t) 
 $, $P^{\mdp{A}}(s,a,t)=P^{\mdp{B}}(\eta(s), a,\eta(t))$, and 
  (iii) for each $s\in Q^{\mdp{A}}$, $R^{\mdp{A}}(s) =R^{\mdp{B}}(\eta(s))$. Here we abused notation and let $\eta$ also denote the identity function on exits.

 \end{definition}

 \begin{definition}[sum $\oplus$ of roMDPs]
 \label{def:sumRightwardOpenMDPs}
  Let $\mdp{A}\colon m\to n$ and $\mdp{B}\colon k\to l$ be rightward open MDPs with the same action set $A$. Their \emph{sum} $\mdp{A} \oplus \mdp{B}\colon m+k\to n+l$ is given by
  \begin{math}
      \mdp{A} \oplus \mdp{B} \defeq  
 \bigl(  m+k,  n+l,
      Q^{\mdp{A}}+Q^{\mdp{B}}, A, E^{\mdp{A}}+E^{\mdp{B}} , P^{\mdp{A}} + P^{\mdp{B}}, R^{\mdp{A}}+R^{\mdp{B}}\bigr)
  \end{math}. Here the function $E^{\mdp{A}}+E^{\mdp{B}}$ naturally combines the two functions by case distinction: for each $i\in \nset{m+k}$, 1)
 \begin{math}
 i \mapsto E^{\mdp{A}}(i)
 \end{math} if $i\le m$,  2)
 \begin{math}
 i \mapsto E^{\mdp{B}}(i-m)
 \end{math} if $m < i$ and $E^{\mdp{B}}(i-m)\in Q^{\mdp{B}}$, and 3) 
 \begin{math}
 i\mapsto E^{\mdp{B}}(i-m) + n
 \end{math} if $m < i$ and $E^{\mdp{B}}(i-m)\in \nset{l}$.
 The definitions of $P^{\mdp{A}} + P^{\mdp{B}}$ and $R^{\mdp{A}}+R^{\mdp{B}}$ are similar.
 \end{definition}

  \begin{definition}[identity, swap]\label{def:idSwapRoMDPs}
 Let $m, n$ be natural numbers. The \emph{identity} $\ID_{m}$ on $m$ (over the action set $A$) is given by \(\ID_{m}=(m, m,\allowbreak \emptyset, A, E, !, !)\), where $E(i) = i$ for each $i\in \nset{m}$. The \emph{swap} $\SYM_{m,  n}$ on $m$ and $n$ (\emph{over} A) is given by $\SYM_{m,  n}=(m+n,\, n+m,\,\emptyset, A, E, !, !)$, where (1) $E(i) = i+n$ if $i\in \nset{m}$, and (2) $E(i) = i-n$  if $i\in [m+1, m+n]$.
 \end{definition}

 \subsection{Twists between oMDPs and roMDPs}
  \begin{proposition}\label{prop:twist}
 Let $\nfr{m},\nfl{m},\nfr{n}, \nfl{n}$ be natural numbers, $S$ be the collection of all open MDPs from $(\nfr{m},\nfl{m})$ to $(\nfr{n},\nfl{n})$, and 
 $
 T
 $ be the collection of all roMDPs from 
 $\nfr{m}+\nfl{n}$ to
 $\nfr{n}+\nfl{m}$. Between these two collections, the following two mappings
 $\Phi \colon S\;\rightarrow \; T$ and $\Psi \colon T\;\rightarrow \; S$
 establish a bijective correspondence.
 \begin{equation}\label{eq:twist}
  \begin{aligned}
    \begin{tikzpicture}[
              innode/.style={draw, rectangle, minimum size=0.5cm},
              innodemini/.style={draw, rectangle, minimum size=0.5cm},
              interface/.style={inner sep=0},
              innodepos/.style={draw, circle, minimum size=0.2cm},
              ]
              \node[interface] (rdo1) at (0cm, 0.125cm) {};
              \node[innode, fill=white] (pos1) at (0.8cm, 0cm) {$\mdp{A}$};
              \node[interface] (ldo1) at (0cm, -0.125cm) {};
              \draw[->] (rdo1) to node[above] {\scalebox{0.5}{$\nfr{m}$}} ($(pos1.north west)!0.5!(pos1.west)$);
              \draw[<-] (ldo1) to node[above,yshift=-0.05cm] {\scalebox{0.5}{$\nfl{m}$}} ($(pos1.south west)!0.5!(pos1.west)$);
              \node[interface] (rcdo1) at (1.6cm, 0.125cm) {};
              \node[interface] (lcdo1) at (1.6cm, -0.125cm) {};
              \draw[<-] (rcdo1) to node[above] {\scalebox{0.5}{$\nfr{n}$}} ($(pos1.north east)!0.5!(pos1.east)$);
              \draw[->] (lcdo1) to node[above,yshift=-0.05cm] {\scalebox{0.5}{$\nfl{n}$}} ($(pos1.south east)!0.5!(pos1.east)$);
              \node[interface] (maps1) at (2cm, 0.1cm) {$\overset{\Phi}{\longmapsto}$};
              \fill[lightgray](2.3cm, 0.4cm)--(2.3cm, -0.6cm)--(4.1cm, -0.6cm)--(4.1cm, 0.4cm)--cycle;
              \node[interface] (rdo11) at (2.4cm, 0.125cm) {};
              \node[interface] (ldo11) at (2.4cm, -0.55cm) {};
              \node[interface] (rcdo11) at (4cm, 0.125cm) {};
              \node[interface] (lcdo11) at (4cm, -0.55cm) {};
              \node[innode, fill=white] (pos2) at (3.2cm, 0cm) {$\mdp{A}$};
              \draw[->] (rdo11) to node[above] {\scalebox{0.5}{$\nfr{m}$}} ($(pos2.north west)!0.5!(pos2.west)$);
              \draw[-] (2.95cm, -0.125cm) arc [radius=0.15, start angle = 90, end angle=270];
              \draw[->] ($(pos2.north east)!0.5!(pos2.east)$) to node[above] {\scalebox{0.5}{$\nfr{n}$}} (rcdo11);
              \draw[->] (2.95cm, -0.425cm) to node[at end, above, xshift=-0.2cm] {\scalebox{0.5}{$\nfl{m}$}} (lcdo11);
              \draw[->] (3.45cm, -0.425cm) arc [radius=0.15, start angle = 270, end angle=450];
              \draw[-] (3.45cm, -0.425cm) to node[at end, above, xshift=0.2cm] {\scalebox{0.5}{$\nfl{n}$}} (ldo11);
        \end{tikzpicture},\qquad
    \begin{tikzpicture}[
              innode/.style={draw, rectangle, minimum size=0.5cm},
              innodemini/.style={draw, rectangle, minimum size=0.5cm},
              interface/.style={inner sep=0},
              innodepos/.style={draw, circle, minimum size=0.2cm},
              ]
              \fill[lightgray](-0.1cm, 0.4cm)--(-0.1cm, -0.6cm)--(1.7cm, -0.6cm)--(1.7cm, 0.4cm)--cycle;
              \node[interface] (rdo11) at (0cm, 0.125cm) {};
              \node[interface] (ldo11) at (0cm, -0.55cm) {};
              \node[interface] (rcdo11) at (1.6cm, 0.125cm) {};
              \node[interface] (lcdo11) at (1.6cm, -0.55cm) {};
              \node[innode, fill=white] (pos2) at (0.8cm, 0cm) {$\mdp{A}$};
              \draw[->] (rdo11) to node[above] {\scalebox{0.5}{$\nfr{m}$}} ($(pos2.north west)!0.5!(pos2.west)$);
              \draw[<-] (0.55cm, -0.125cm) arc [radius=0.15, start angle = 90, end angle=270];
              \draw[->] ($(pos2.north east)!0.5!(pos2.east)$) to node[above] {\scalebox{0.5}{$\nfr{n}$}} (rcdo11);
              \draw[-] (0.45cm, -0.425cm) to node[at end, above, xshift=-0.2cm] {\scalebox{0.5}{$\nfl{n}$}} (lcdo11);
              \draw[-] (1.05cm, -0.425cm) arc [radius=0.15, start angle = 270, end angle=450];
              \draw[->] (1.05cm, -0.425cm) to node[at end, above, xshift=0.2cm] {\scalebox{0.5}{$\nfl{m}$}} (ldo11);
              \node[interface] (maps1) at (2cm, 0.1cm) {$\overset{\Psi}{\reflectbox{$\longmapsto$}}$};
              \node[interface] (rdo1) at (2.4cm, 0.125cm) {};
              \node[innode, fill=white] (pos1) at (3.2cm, 0cm) {$\mdp{A}$};
              \node[interface] (ldo1) at (2.4cm, -0.125cm) {};
              \draw[->] (rdo1) to node[above] {\scalebox{0.5}{$\nfr{m}$}} ($(pos1.north west)!0.5!(pos1.west)$);
              \draw[->] (ldo1) to node[above,yshift=-0.05cm] {\scalebox{0.5}{$\nfl{n}$}} ($(pos1.south west)!0.5!(pos1.west)$);
              \node[interface] (rcdo1) at (4cm, 0.125cm) {};
              \node[interface] (lcdo1) at (4cm, -0.125cm) {};
              \draw[<-] (rcdo1) to node[above] {\scalebox{0.5}{$\nfr{n}$}} ($(pos1.north east)!0.5!(pos1.east)$);
              \draw[<-] (lcdo1) to node[above,yshift=-0.05cm] {\scalebox{0.5}{$\nfl{m}$}} ($(pos1.south east)!0.5!(pos1.east)$);
        \end{tikzpicture}
  \end{aligned}
 \end{equation}

 \end{proposition}
 \begin{proof}
  Straightforward.
 \vspace*{-1.2em} 
 \qed
 \end{proof}
\subsection{Decomposition Equalities for Sum $\oplus$ of $\roMC$}
\label{prop:mc_parallel}
\begin{proposition}
    Let $\mc{C}:m_1\rightarrow n_1$ and $\mc{D}:m_2\rightarrow n_2$ be rightward open MCs. The following equalities of matrices hold:
    \begin{align*}
   &     \bigl[\,\Reacha{i}{j}{\mc{C}\oplus \mc{D}}\,\bigr]_{i\in\nset{m_1+m_2},j\in\nset{n_1+n_2}}  \;=\; 
   \begin{pmatrix}
 \bigl[\,\Reacha{i}{j}{\mc{C}}\,\bigr]_{i\in\nset{m_1},j\in\nset{n_1}}  &  \bigl[\,0\,\bigr]_{i\in\nset{m_1},j\in\nset{n_2}} \\
   \bigl[\,
    0
   \,\bigr]_{i\in\nset{m_2},j\in\nset{n_1}} & \bigl[\,
    \Reacha{i}{j}{\mc{D}} 
   \,\bigr]_{i\in\nset{m_2},j\in\nset{n_2}}\\
\end{pmatrix},
   \\
   &
   \bigl[\,\ETRa{i}{j}{\mc{C}\oplus \mc{D}}\,\bigr]_{i\in\nset{m_1+m_2},j\in\nset{n_1+n_2}} 
    \;=\;
    \begin{pmatrix}
 \bigl[\,\ETRa{i}{j}{\mc{C}}\,\bigr]_{i\in\nset{m_1},j\in\nset{n_1}}  &  \bigl[\,0\,\bigr]_{i\in\nset{m_1},j\in\nset{n_2}} \\
   \bigl[\,
    0
   \,\bigr]_{i\in\nset{m_2},j\in\nset{n_1}} & \bigl[\,
    \ETRa{i}{j}{\mc{D}} 
   \,\bigr]_{i\in\nset{m_2},j\in\nset{n_2}}\\
\end{pmatrix}.
    \end{align*}
\end{proposition}

\subsection{Complete Definitions of $\semCatr$}
\label{apnd:completeDefOfSemCatr}

\begin{definition}[sum $\oplus$ in $\semCatr$]
\label{def:sum_in_sr}
Let $F:m\rightarrow n$, $G:k\rightarrow l$ be arrows in $\semCatr$.
The \emph{sum} $F\oplus G$ of $F,G$ is given by $F\oplus G\defeq \{ f\oplus g\mid f\in F,\  g\in G\}$ where $f\oplus g$ is the sum of $f, g$ in $\semCatrmc$.
    
\end{definition}

\begin{definition}[identity and swap]
\label{def:identity_and_swap} 
Let $m$ be a natural number. The \emph{identity} on $m$ in $\semCatr$ is given by the singleton set $\{ \id_m \}$ where $\id_m$ is the identity on $m$ in $\semCatrmc$. Let $m, n$ be natural numbers. The \emph{swap} on $m, n$ in $\semCatr$ is given by $\{\swap{m}{n}\}$ where $\swap{m}{n}$ is the swap on $m$ and $n$ in $\semCatrmc$.
\end{definition}

\section{Proof of~\cref{thm:roMDPsTSMCEqAx}}
\label{sec:proof_roMDPsTSMCEqAx}

We shall first prove the well-definedness of~\cref{def:trace_tmdp}, i.e., that the trace $\trace{l}{m}{n}{}(\mdp{A})$ is indeed an roMDP. The following  lemma is used in the proof.

\begin{lemma}
\label{lem:unique_precedent}
In the setting of \cref{def:trace_tmdp},
let $\mdp{A}:l+m\rightarrow l+n$, $i\in \nset{l}$, and $s, s'\in [l+1, l+n] + Q^{\mdp{A}}$. If $i\in \precedent(s)$ and $i\in \precedent(s')$, then $s = s'$.
\end{lemma}
\begin{proof}
    By definition.
\end{proof}

\begin{proof}[of well-definedness of~\cref{def:trace_tmdp}]
Let $\mdp{A}:l+m\rightarrow l+n$.
For $s\in Q^{\mdp{A}}$ and $a\in A$, the following equalities hold by~\cref{lem:unique_precedent}:
\begin{align*}
    \sum_{s'\in Q^{\mdp{A}}+\nset{n}}P^{\mdp{A}}(s, a, s') &= \big(\sum_{s'\in Q^{\mdp{A}}}P^{\mdp{A}}(s, a, s') + \sum_{i\in \precedent(s')}P^{\mdp{A}}(s, a, i)\big)\\
    &\qquad+\big(\sum_{i\in\nset{n}}P^{\mdp{A}}(s, a, i+l) + \sum_{j\in \precedent(i+l)}P^{\mdp{A}}(s, a, j)\big)\\
    &= \sum_{s'\in Q^{\mdp{A}}}P^{\mdp{A}}(s, a, s')+\sum_{i\in\nset{n}}P^{\mdp{A}}(s, a, i+l) + \sum_{i\in\nset{l}}P^{\mdp{A}}(s, a, i).
\end{align*}
The last value is clearly either $0$ or $1$. 

The ``unique access to each exit'' condition in \cref{def:openMDP} can also be easily proved. \qed
\end{proof}

Towards the proof of \cref{thm:roMDPsTSMCEqAx}, 
 we prove (Naturality1) in the rest of the section.
The proofs of the other equational axioms of the trace operator are either
simpler or similar.

\begin{proof}[of (Naturality1) from \cref{fig:TSMCAxioms}]
  Assume $\mdp{A} \colon l + m \rightarrow  l + n$ and $\mdp{B}
  \colon m' \to m$.
  Firstly, we define a partial function $\Estar{l}{\mdp{A}} \colon \nset{l + m}
  \rightharpoonup Q^{\mdp{A}} + \nset{n}$ recursively by
  \[
    \Estar{l}{\mdp{A}}(i) =
    \begin{cases}
      E^{\mdp{A}}(i) & \text{if $E^{\mdp{A}}(i) \in Q^{\mdp{A}}$} \\
      E^{\mdp{A}}(i) - l & \text{if $E^{\mdp{A}}(i) \in \nset{l
        + n} \setminus \nset{l}$} \\
      \Estar{l}{\mdp{A}}(E^{\mdp{A}}(i)) & \text{if $E^{\mdp{A}}(i)
        \in \nset{l}$,}
    \end{cases}
  \]
  where $\Estar{l}{\mdp{A}}(i)$ is undefined if the recursion never
  reaches the base case.

  By 
the ``unique access to each exit'' condition in \cref{def:openMDP},
 $\Estar{l}{\mdp{A}}(i)$ is defined for all $i \in
  \nset{l + m} \setminus \nset{l}$.

  We note the following facts about $\Estar{l}{\mdp{A}}$.
  \begin{itemize}
    \item 
      For all $i \in \nset{m}$,
      \begin{equation}
        E^{\trace{l}{m}{n}{}(\mdp{A})}(i) = \Estar{l}{\mdp{A}}(i+l)
        \rlap{.}
        \label{eq:proof_trace_nat_X:star_vs_E}
      \end{equation}
    \item 
      For all $i \in \nset{l + m}$ and $q \in Q^{\mdp{A}} +
      \nset{n}$,
      \begin{equation}
          \Estar{l}{\mdp{A}}(i) = q \quad\Longleftrightarrow\quad
          E^{\mdp{A}}(i) = q \lor \exists j \in \nset{l}.
          (E^{\mdp{A}}(i) = j \land \Estar{l}{\mdp{A}}(j) = q)
          \rlap{.}
        \label{eq:proof_trace_nat_X:star_vs_sum}
      \end{equation}
  \end{itemize}

 Towards the proof of (Naturality1), we now establish a connection between $\Estar{l}{(\ID_l\oplus 
  \mdp{B}); \mdp{A}}$ and $\Estar{l}{\mdp{A}}$.

  Firstly, when we take $i \in \nset{l}$, we have
  \begin{align*}
    E^{(\ID_l \oplus \mdp{B}); \mdp{A}}(i)
    &=
    \begin{cases}
      E^{\ID_l \oplus \mdp{B}}(i) & \text{if $E^{\ID_l \oplus \mdp{B}}(i) \in Q^{\ID_l \oplus \mdp{B}}$} \\
      E^{\mdp{A}}(E^{\ID_l \oplus \mdp{B}}(i)) & \text{if $E^{\ID_l \oplus \mdp{B}}(i) \in \nset{l + m}$}
    \end{cases} \\
    &=
    \begin{cases}
      E^{\ID_l}(i) & \text{if $E^{\ID_l}(i) \in Q^{\ID_l}$} \\
      E^{\mdp{A}}(E^{\ID_l}(i)) & \text{if $E^{\ID_l}(i) \in \nset{l}$}
    \end{cases} \\
    &= E^{\mdp{A}}(i) \rlap{,}
  \end{align*}
  where $E^{\ID_l \oplus \mdp{B}}(i) = E^{\ID_l}(i)$ because $i \in
  \nset{l}$.

  Similarly, when we take $i \in \nset{l + m'} \setminus \nset{l}$
  such that $E^{\mdp{B}}(i-l) \in \nset{m}$, we have
  \begin{align*}
    E^{(\ID_l \oplus \mdp{B}); \mdp{A}}(i)
    &=
    \begin{cases}
      E^{\ID_l \oplus\mdp{B}}(i) & \text{if $E^{\ID_l \oplus \mdp{B}}(i) \in Q^{\ID_l \oplus \mdp{B}}$} \\
      E^{\mdp{A}}(E^{\ID_l \oplus \mdp{B}}(i)) & \text{if $E^{\ID_l \oplus \mdp{B}}(i) \in \nset{l + m}$}
    \end{cases} \\
    &=
    \begin{cases}
      E^{\mdp{B}}(i-l)+l & \text{if $E^{\mdp{B}}(i-l)+l \in Q^{\mdp{B}}$} \\
      E^{\mdp{A}}(E^{\mdp{B}}(i-l)+l) & \text{if $E^{\ID_l}(i-l)+l \in \nset{l + m}$}
    \end{cases} \\
    &= E^{\mdp{A}}(E^{\mdp{B}}(i-l)+l) \rlap{,}
  \end{align*}
  since the first case contradicts $E^{\mdp{B}}(i-l) \in \nset{m}$.

  Combining the last two facts (and a simple computation for the last
  one), we obtain the following.
  \begin{itemize}
    \item For all $i \in \nset{l}$, $\Estar{l}{(\ID_l \oplus \mdp{B});
      \mdp{A}}(i) = \Estar{l}{\mdp{A}}(i)$ (by which we mean that, if
      one side is defined, then so is the other, and they are equal).
    \item For all $i \in \nset{l + m'} \setminus \nset{l}$ such
      that $E^{\mdp{B}}(i-l) \in \nset{m}$, $\Estar{l}{(\ID_l \oplus
      \mdp{B}); \mdp{A}}(i) = \Estar{l}{\mdp{A}}(E^{\mdp{B}}(i-l)+l)$.
    \item For all $i \in \nset{l + m'} \setminus \nset{l}$ such
    that $E^{\mdp{B}}(i-l) \in Q^{\mdp{B}}$, $\Estar{l}{(\ID_l \oplus
    \mdp{B}); \mdp{A}}(i) = E^{\mdp{B}}(i-l)$.
  \end{itemize}

  By~\eqref{eq:proof_trace_nat_X:star_vs_E} we have, for all $i \in
  \nset{m'}$,
  \begin{align*}
    E^{\trace{l}{m'}{n}{}((\ID_l \oplus \mdp{B}) ; \mdp{A})}(i)
    &= \Estar{l}{(\ID_l \oplus \mdp{B}) ; \mdp{A}}(i+l) \\
    &=
    \begin{cases}
      \Estar{l}{\mdp{A}}(E^{\mdp{B}}(i)+l)
      & \text{if $E^{\mdp{B}}(i) \in \nset{m}$} \\
      E^{\mdp{B}}(i)
      & \text{if $E^{\mdp{B}}(i) \in Q^{\mdp{B}}$}
    \end{cases} \\
    &=
    \begin{cases}
      E^{\trace{l}{m}{n}{}(\mdp{A})}(E^{\mdp{B}}(i))
      & \text{if $E^{\mdp{B}}(i) \in \nset{m}$} \\
      E^{\mdp{B}}(i)
      & \text{if $E^{\mdp{B}}(i) \in Q^{\mdp{B}}$}
    \end{cases} \\
    &= E^{\mdp{B} ; \trace{l}{m'}{n}{}(\mdp{A})}(i)
  \end{align*}
  as desired.

  For $P^{\trace{l}{m'}{n}{}((\ID_l \oplus \mdp{B}); \mdp{A})}$, we use
  the following equivalent definition of
  $P^{\trace{l}{m}{n}{}(\mdp{A})}$:
  \[
    P^{\trace{l}{m}{n}{}(\mdp{A})}(s,a,s') =
      P^{\mdp{A}}(s,a,s'_{l}) + \sum_{i \in \nset{l}} P^{\mdp{A}}(s,a,i)
      \cdot \kron{\Estar{l}{\mdp{A}}(i)}{s'_l}
  \]
  where $s'_l$ is $s'$ if $s' \in Q^{\mdp{A}}$ and $s'+l$ otherwise
  (i.e.\ if $s' \in \nset{n}$).
  In particular, we have
  that
  \begin{align*}
    P^{\trace{l}{m'}{n}{}((\ID_l \oplus \mdp{B}); \mdp{A})}(s,a,s')
    &= P^{(\ID_l \oplus \mdp{B}); \mdp{A}}(s,a,s'_l) \\
    &\phantom{={}} + \sum_{i \in \nset{l}} P^{(\ID_l \oplus \mdp{B});
      \mdp{A}}(s,a,i) \cdot \kron{\Estar{l}{(\ID_l \oplus \mdp{B});
      \mdp{A}}(i)}{s'_l} \\
    &= P^{(\ID_l \oplus \mdp{B}); \mdp{A}}(s,a,s'_l) \\
    &\phantom{={}} + \sum_{i \in \nset{l}} P^{(\ID_l \oplus \mdp{B});
      \mdp{A}}(s,a,i) \cdot \kron{\Estar{l}{\mdp{A}}(i)}{s'_l}.
  \end{align*}

  Moreover, we can compute:
  \begin{align*}
    P^{\mdp{B}; \trace{l}{m}{n}{}(\mdp{A})}(s,a,s')
    &=
    \begin{cases}
      P^{\mdp{B}}(s,a,s') & \text{if $s,s' \in Q^{\mdp{B}}$} \\
      \sum_{i \in \nset{m}} P^{\mdp{B}}(s,a,i) \cdot
        \kron{E^{\trace{l}{m}{n}{}(\mdp{A})}(i)}{s'_l} & \text{if $s
        \in Q^{\mdp{B}}$ and $s' \in Q^{\mdp{A}} + \nset{n}$} \\
      0 & \text{if $s \in Q^{\mdp{A}}$ and $s' \in Q^{\mdp{B}}$} \\
      P^{\trace{l}{m}{n}{}(\mdp{A})}(s,a,s') & \text{if $s \in
        Q^{\mdp{A}}$ and $s' \in Q^{\mdp{A}} + \nset{n}$}
    \end{cases} \\
    &=
    \begin{cases}
      P^{\mdp{B}}(s,a,s') & \text{if $s,s' \in Q^{\mdp{B}}$} \\
      \sum_{i \in \nset{m}} P^{\mdp{B}}(s,a,i) \cdot
        \kron{\Estar{l}{\mdp{A}}(i+l)}{s'_l} & \text{if $s
        \in Q^{\mdp{B}}$ and $s' \in Q^{\mdp{A}} + \nset{n}$} \\
      0 & \text{if $s \in Q^{\mdp{A}}$ and $s' \in Q^{\mdp{B}}$} \\
      P^{\mdp{A}}(s,a,s'_l) + \sum_{i \in \nset{l}} P^{\mdp{A}}(s,a,i)
        \cdot \kron{\Estar{l}{\mdp{A}}(i)}{s'_l} & \text{if $s \in
        Q^{\mdp{A}}$ and $s' \in Q^{\mdp{A}} + \nset{n}$}
    \end{cases}
  \end{align*}

 We distinguish all possible cases.
  \begin{itemize}
    \item If $s,s' \in Q^{\mdp{B}}$, then in particular $s'_l = s'$,
      and
      \begin{align*}
        P^{\trace{l}{m'}{n}{}((\ID_l \oplus \mdp{B}); \mdp{A})}(s,a,s')
        &= P^{\mdp{B}}(s,a,s') + \sum_{i \in \nset{l}}
          P^{(\ID_l \oplus \mdp{B}); \mdp{A}}(s,a,i) \cdot
          \kron{\Estar{l}{\mdp{A}}(i)}{s'} \\
        &= P^{\mdp{B}}(s,a,s') \rlap{,}
      \end{align*}
      since $s' \notin Q^{\mdp{A}} + \nset{n}$, which is the codomain
      of $\Estar{l}{\mdp{A}}$.
    \item If $s \in Q^{\mdp{B}}$ and $s' \in Q^{\mdp{A}} + \nset{n}$,
      then
      \begin{align*}
        P^{\trace{l}{m'}{n}{}((\ID_l \otimes \mdp{B}); \mdp{A})}(s,a,s')
        &= \left( \sum_{j \in \nset{l + m}} P^{(\ID_l \oplus
          \mdp{B})}(s,a,j) \cdot \kron{E^{\mdp{A}}(j)}{s'_l} \right)
          \\
        &\phantom{={}} + \sum_{i \in \nset{l}} \left( \sum_{j \in
          \nset{l + m}} P^{(\ID_l \oplus \mdp{B})}(s,a,j) \cdot
          \kron{E^{\mdp{A}}(j)}{i} \right) \cdot
          \kron{\Estar{l}{\mdp{A}}(i)}{s'_l} \\
        &= \sum_{j \in \nset{m}} P^{\mdp{B}}(s,a,j) \left(
          \kron{E^{\mdp{A}}(j+l)}{s'_l} + \sum_{i \in \nset{l}}
          \kron{E^{\mdp{A}}(j+l)}{i} \cdot
          \kron{\Estar{l}{\mdp{A}}(i)}{s'_l} \right) \\
        &= \sum_{j \in \nset{m}} P^{\mdp{B}}(s,a,j) \cdot
          \kron{\Estar{l}{\mdp{A}}(j+l)}{s'_l} \rlap{,}
      \end{align*}
      where indices $j \in \nset{l}$ are discarded in the second step
      since $s \in Q^{\mdp{B}}$. The last equality follows from~\eqref{eq:proof_trace_nat_X:star_vs_sum}.
    \item If $s \in Q^{\mdp{A}}$ and $s' \in Q^{\mdp{B}}$, then
      \[
        P^{\trace{l}{m'}{n}{}((\ID_l\oplus \mdp{B}); \mdp{A})}(s,a,s')
        = 0 + \sum_{i \in \nset{l}} P^{(\ID_l \oplus \mdp{B});
        \mdp{A}}(s,a,i) \cdot \kron{\Estar{l}{\mdp{A}}(i)}{s'_l} = 0
        \rlap{,}
      \]
      since $s'_l = s' \notin Q^{\mdp{A}} + \nset{n}$, which is the
      codomain of $\Estar{l}{\mdp{A}}$.
    \item If $s \in Q^{\mdp{A}}$ and $s' \in Q^{\mdp{A}} + \nset{n}$,
      then
      \begin{align*}
        P^{\trace{l}{m'}{n}{}((\ID_l \oplus \mdp{B}); \mdp{A})}(s,a,s')
        &= P^{\mdp{A}}(s,a,s'_l) + \sum_{i \in \nset{l}} P^{\mdp{A}}(s,a,i)
          \cdot \kron{\Estar{l}{\mdp{A}}(i)}{s'_l}
      \end{align*}
      as desired.
  \end{itemize}
  Therefore $P^{\trace{l}{m'}{n}{}((\ID_l \oplus \mdp{B}); \mdp{A})} =
  P^{\mdp{B}; \trace{l}{m}{n}{}(\mdp{A})}$, which concludes the proof.
\end{proof}

\section{Proof for~\cref{prop:MCERweq} }
\label{sec:proof_MCERweq}
\begin{proof}
We prove the first decomposition equality for  sequential composition. For each entrance $i\in \nset{m}$ and exit $j\in \nset{n}$, the expected reward $\ETRa{i}{j}{\mc{C} \seqcomp \mc{D}}$ is given by 
\begin{equation}
\label{eq:seq_etrmc}
    \ETRa{i}{j}{\mc{C} \seqcomp \mc{D}}\defeq \sum_{\mypath{i}{j}\in  \Pathstra{\mc{C} \seqcomp \mc{D}}{\tau^{\mc{C}\seqcomp \mc{D}}}{i}{j}}\ProPathc{\mc{C}\seqcomp\mc{D}}(\mypath{i}{j})\cdot \PReward{\mc{C} \seqcomp \mc{D}}(\mypath{i}{j}).
\end{equation}
By the definition of the sequential composition $\mc{C} \seqcomp \mc{D}$, the right hand side of~\cref{eq:seq_etrmc} is equal to the following:
\begin{equation}
\label{eq:expand_seq_etrmc}
      \sum_{k\in \nset{l}}\sum_{\mypath{i}{k}\in  \Pathstra{\mc{C}}{\tau^{\mc{C}\seqcomp \mc{D}}}{i}{k}}\sum_{\mypath{k}{j}\in  \Pathstra{\mc{D}}{\tau^{\mc{C}\seqcomp \mc{D}}}{k}{j}} \ProPathc{\mc{C}}(\mypath{i}{k})\cdot \ProPathc{\mc{D}}(\mypath{k}{j})\cdot\big(\PReward{\mc{C}}(\mypath{i}{k}) + \PReward{\mc{D}}(\mypath{k}{j})  \big).
\end{equation}
By changing the order of the summation, we can see that~\cref{eq:expand_seq_etrmc} is equal to 
\begin{equation*}
\sum_{k\in \nset{l}} \Reacha{i}{k}{\mc{C}}\cdot \ETRa{k}{j}{\mc{D}} +  \ETRa{i}{k}{\mc{C}}\cdot \Reacha{k}{j}{\mc{D}}.
\end{equation*}

The latter decomposition equality for the trace operator can be easily obtained by proving that it is equivalent to the linear equations in~\cref{prop:linear_equations}. This is straightforward. 
\qed
\end{proof}

\section{Monads, Girard's Execution Formula, and Change of Base}
\label{sec:monads_execution_formula_change_of_base}
We 
introduce some basics of 
 \emph{monads} (a categorical notion~\cite{MacLane2} whose use in computer science is advocated e.g.\ in~\cite{Moggi91}), Girard's \emph{execution formula}~\cite{girard1989geometry,haghverdiS06,hoshino12}, and the \emph{change of base} construction~\cite{Eilenberg65,cruttwell2008normed}. We also discuss the relationship between these abstract machineries and our current compositional model checking framework.

In this section, for a category $\mathbb{C}$ and its objects $X, Y$, we write $\mathbb{C}(X, Y)$ for the set of arrows from $X$ to $Y$. 
\subsection{Monads}
\label{sec:monad}
The notion of monad is widely used not only in semantical studies but also in real-world programming languages such as Haskell. In the following definition, we use its equivalent presentation as a \emph{Kleisli triple}~\cite{Moggi91}. 

Note that, in the following, we let $\seqcomp$ denote composition of arrows in a category: given arrows $f\colon X\to Y$ and $g\colon Y\to Z$, the composite arrow of \emph{$f$ and then $g$} is denoted by $f\seqcomp g\colon X\to Z$. In the categorical literature, such composition is often denoted by $g\circ f$ (\emph{$g$ after $f$}; note the order). 

\begin{definition}[Kleisli triple]
A \emph{Kleisli triple} on the category $\mathbb{C}$ is a triple $(T, \eta,(\place)^{\star})$ consists of the following data: 
\begin{itemize}
    \item an assignment $T$, to each object $X$, of an object $T(X)$,
    \item a \emph{unit} $\eta_X: X\rightarrow T(X)$ for each object $X$,
    \item a \emph{Kleisli extension operator} $(\place)^{\star}$ that, given an arrow $f:X\rightarrow T(Y)$, returns an arrow  $f^{\star}:T(X)\rightarrow T(Y)$.
\end{itemize}
 The triple $(T, \eta,(\place)^{\star})$ is further required to satisfy the following conditions: for any objects $X, Y, Z$ and arrows $f: X\rightarrow T(Y), g:Y\rightarrow T(Z)$, 
\begin{itemize}
    \item $(\eta_X)^{\star} = \id_{T(X)}$,
    \item $\eta_X \seqcomp f^{\star}= f$,
    \item $f^{\star}\seqcomp g^{\star}  = (f\seqcomp g^{\star})^{\star}$.
\end{itemize}
    
\end{definition}

The \emph{Kleisli category} $\kleisli{\Sets}{T}$ for a monad $T$ accommodates \emph{$T$-effectful computation}~\cite{Moggi91}: an arrow $X\to Y$ in $\kleisli{\Sets}{T}$  is an  $X\to TY$ in $\mathbb{C}$; in common examples where $\mathbb{C}=\Sets$ (the category of sets and functions), such an arrow is  thought of as a function from $X$ to $Y$ \emph{with $T$-effect}.

\begin{definition}[Kleisli category]
Let $(T, \eta, (\place)^{\star})$ be a Kleisli triple on a category $\mathbb{C}$. The \emph{Kleisli category} is the category $\kleisli{\mathbb{C}}{T}$ whose objects are that of $\mathbb{C}$, and whose set $\kleisli{\mathbb{C}}{T}(X, Y)$  of arrows for arbitrary objects $X, Y$ is given by $\kleisli{\mathbb{C}}{T}(X, Y)\defeq \mathbb{C}(X, T(Y))$.  For each object $X$, the \emph{identity} $\id_X$ in  $\kleisli{\mathbb{C}}{T}$ is given by $\id_X \defeq \eta_X$, and for any arrows $f:X\rightarrow T(Y)$ and $g:Y\rightarrow T(Z)$ in $\mathbb{C}$ (identified with arrows $f\colon X\rightarrow Y$ and $g\colon Y\rightarrow Z$ in $\kleisli{\mathbb{C}}{T}$), their \emph{sequential composition} $f\seqcomp g$ in  $\kleisli{\mathbb{C}}{T}$ is given by $f\seqcomp g \defeq f\seqcomp g^{\star}$, where the latter $\seqcomp$ is in $\mathbb{C}$.
\end{definition}

We introduce 
the \emph{probabilistic reward monad}. It encapsulates the notion of effect given by  \emph{reachability probabilities} and \emph{expected rewards}.

\begin{definition}[the probabilistic reward monad]\label{def:prReMonad}
The \emph{probabilistic reward monad (PR monad)}  $(\mcMnd, \eta,  (\place)^{\star})$ on $\Sets$ consists of the following data:
 \begin{itemize}
 \item an assignment $\mcMnd$, to each set $X$, of the set $\mcMnd X \defeq \{(p_x, r_x)_{x\in X}\mid p_x,r_x \in\Real_{\ge 0}, \sum_{x}p_{x}\le 1, p_x=0\Rightarrow r_x=0\}$;
 \item the \emph{unit} function $\eta_{X}\colon X\to \mcMnd X$ given by $\eta_{X}(x) \defeq  (p_{x'},r_{x'})_{x'\in X}$, $p_{x'}=1$ if $x'=x$ and $p_{x'}=0$ otherwise, and $r_{x'}=0$ for each $x'$; and
 \item the \emph{Kleisli extension} operator $(\place)^{\star}$ that, given a function $f\colon X\to \mcMnd Y$, returns the function $f^{\star}\colon \mcMnd X\to \mcMnd Y$ defined as follows. Let $f(x)=(p^{f(x)}_{y},r^{f(x)}_{y})_{y\in Y}\in \mcMnd Y$ for each $x\in X$, and $t=(p^{t}_{x}, r^{t}_{x})_{x\in X}\in \mcMnd X$. Then $f^{\star}(t) \defeq (p^{f^{\star}(t)}_{y},r^{f^{\star}(t)}_{y})_{y\in Y}$, where 
\begin{equation}\label{eq:defprReMonad}
 \begin{array}{l}
  p^{f^{\star}(t)}_{y}\;\defeq \;  \sum_{x\in X} p^{t}_{x} \cdot
 p^{f(x)}_{y} 
,
\quad
  r^{f^{\star}(t)}_{y}\defeq\sum_{x\in X} 
  \bigl(\,
  p^{t}_{x}\cdot  r^{f(x)}_{y}
    + r^{t}_{x}\cdot p^{f(x)}_{y}\,\bigr).
 \end{array}
\end{equation}
\end{itemize}
\begin{auxproof}
 The definition  is in fact a  combination of the subdistribution monad (see e.g.~\cite{hasuo2015generic}) and the monoid $(\mathbb{R}_{\ge 0}, +, 0)$. The definitions in \cref{eq:defprReMonad} can be seen as an analogue of the first equalities of \cref{prop:MCeq,prop:MCERweq}, too.
\end{auxproof}
\begin{auxproof}
 Abstract construction of the monad?
\end{auxproof}
\end{definition}

\begin{proposition}
\begin{enumerate}
 \item 
  For the PR monad $\mcMnd$ in \cref{def:prReMonad}, consider the restriction $\rest{\kleisli{\Sets}{\mcMnd }}{\Nat}$ of the Kleisli category $\kleisli{\Sets}{\mcMnd }$ to  natural numbers as objects (i.e.\ the full subcategory for those objects). Then it is a TSMC, with respect to sum $\oplus$ and a naturally defined trace operator (using e.g.~\cite{hoshino12}, see~\cref{subsec:execution_formula}). 
 \item \label{item:propsemCatrmcAsKleisli}
  The restriction $\rest{\kleisli{\Sets}{\mcMnd }}{\Nat}$ has the same objects  as $\semCatrmc$ (\cref{def:semCatrmc}); their arrows are in a canonical bijective correspondence, too.
\qed
\end{enumerate}
\end{proposition}

\subsection{Girard's Execution Formula and Its Categorical Ramifications}
\label{subsec:execution_formula}
We review some basic theory~\cite{girard1989geometry,haghverdiS06,hoshino12} of Girard's execution formula, its categorical formulation, and strong unique decomposition categories. An important implication of this theory is that the semantic category  $\semCatrmc$ is a TSMC (cf.\ \cref{def:semCatrmcAsTSMC}); this is shown by establishing that  $\semCatrmc$ is a strong  unique decomposition category.

We write $e \kli e'$ for the \emph{Kleene inequality}: i.e.,  $e \kli e'$ holds  if
whenever \(e\) is defined, so is \(e'\) and \(e\) equals \(e'\).
We write \(e \kle e'\) if \(e \kli e'\) and \(e \klg e'\).
Note that \(e = e'\) means that both \(e\) and \(e'\) are always defined and they are the same. 
We write \(\fary{X}\) (\(\iary{X}\) resp.) for the sets of indexed families \((x_i)_{i \in I}\) of elements in \(X\) whose indexing sets \(I\) are finite (infinite resp.) subsets of \(\Nat\).
A \emph{countable partition} of $I\subseteq \Nat$ is
a family ${(I_j)}_{j\in J}$ of pairwise disjoint subsets of \(\Nat\) indexed by $J\subseteq \Nat$ such that \(\cup_{j \in J} I_j = I\).

\begin{definition}[$\Sigma$-monoid]
A \emph{$\Sigma$-monoid} is a non-empty set $X$ with a partial map $\Sigma: \cary{X} \pto X$ satisfying the following conditions:
\begin{itemize}
    \item If I is a singleton \(\{n\}\), then $\Sigma (x_i)_{i \in I} \kle x_n$,
    \item If ${(I_j)}_{j\in J}$ is a countable partition of $I\subseteq \Nat$,
then for every $(x_i)_{i\in I} \in \cary{X}$, $\Sigma (x_i)_{i\in I} \kle \Sigma \big(\Sigma (x_{i})_{i\in I_j}\big)_{j\in J}$.
\end{itemize}
\end{definition}

Note that the former condition above implies that the left-hand-side \(\Sigma (x_i)_{i \in I}\) is always defined since so is the right-hand-side \(x_n\).
A family \((x_i)_{i \in I}\) is called \emph{summable} if \(\Sigma (x_i)_{i \in I}\) is defined.
By the latter condition above, any subfamily of a summable family is summable.
Especially, the empty family \((x_i)_{i \in \emptyset}\) is summable, and we write \(0 \defeq \Sigma (x_i)_{i \in \emptyset}\), which we call the \emph{zero element}.
We also note that the summation is ``commutative'' in the sense that \(\Sigma (x_i)_{i \in I} \kle \Sigma (x_{\theta(j)})_{j \in J}\) for any \(J \subseteq \Nat\) and bijection \(\theta : J \to I\).
Hence, for any countable family \((x_i)_{i \in I}\) on \(X\) where \(I\) is an arbitrary countable set (and is not necessarily a subset of \(\Nat\)), we can define \(\Sigma (x_i)_{i \in I}\) by \(\Sigma (x_i)_{i \in I} \kle \Sigma (x_{\theta(j)})_{j \in J}\) where we choose \(J \subseteq \Nat\) and a bijection \(\theta: J \to I\); the definition is independent of the choice.

\begin{definition}[$\Sigma$-category]
A \emph{$\Sigma$-category} is a category $\mathbb{C}$ equipped with a $\Sigma$-monoid structure on the homset $\mathbb{C}(X, Y)$ for each objects $X$ and $Y$, and it is further required to  satisfy the following condition.
\begin{itemize}
    \item For any $(f_i:X\rightarrow Y)_{i\in I}$, $(g_j:Y\rightarrow Z)_{j\in J}$, $
\Sigma(f_i)_{i\in I}\seqcomp\Sigma(g_j)_{j\in J}
\kli
\Sigma (f_i\seqcomp g_j)_{(i, j)\in I\times J} 
$ holds.
\end{itemize}
\end{definition}

It is easy to prove that the category $\semCatrmc$ is a $\Sigma$-category by defining the sum of arrows in a pointwise manner: the sum is defined if and only if all the summations converge (which are necessarily absolute convergences) and also the resulting function satisfies the condition for arrows in $\semCatrmc$ (in \cref{def:semCatrmc}).

\begin{definition}[zero arrow]
    Let $\mathbb{C}$ be a $\Sigma$-category. For any objects $X, Y$, the \emph{zero arrow} $\mathbf{0}_{X, Y}$ in $\mathbb{C}(X, Y)$ is the zero element of the $\Sigma$-monoid \(\mathbb{C}(X, Y)\).
\end{definition}

The zero arrow $\mathbf{0}_{m, n}\in \semCatrmc(m, n)$ is given by $\mathbf{0}_{m, n}(i, j) = \big((0)_{i, j}, (0)_{i, j}\big)$. 

By a \emph{symmetric monoidal $\Sigma$-category}, we mean a $\Sigma$-category equipped with a symmetric monoidal structure on its underlying category. (We do not require that the symmetric monoidal structure is compatible with the $\Sigma$-monoid structure on homsets.)
\begin{definition}[strong unique decomposition category]
A \emph{strong unique decomposition category} is a symmetric monoidal $\Sigma$-category  $(\mathbb{C},\otimes, I)$ such that 
\begin{itemize}
    \item the identity on the monoidal unit $I$ is equal to the zero arrow $\mathbf{0}_{I, I}$, and
    \item for any objects $X, Y$, $\id_X \otimes \mathbf{0}_{Y, Y} + \mathbf{0}_{X, X} \otimes \id_Y = \id_{X\otimes Y}$ holds.
\end{itemize}
\end{definition}

The $\Sigma$-category $\semCatrmc$ is a (strict) symmetric monoidal category $(\semCatrmc, 0, \oplus)$, and it is easy to prove that $\semCatrmc$  is a strong unique decomposition category. 

Any strong unique decomposition category \(\mathbb{C}\) has the following ``decomposition'' structure:
For \(X, Y \in \mathbb{C}\), we have \emph{quasi projections}
\[
 \rho^{1}_{X,Y} \defeq 
 X \otimes Y \xrightarrow{\id_{X} \otimes \mathbf{0}_{Y,I}}
 X \otimes I \xrightarrow{\cong} X
\qquad
 \rho^{2}_{X,Y} \defeq 
 X \otimes Y \xrightarrow{\mathbf{0}_{X,I} \otimes \id_{Y}}
 I \otimes Y \xrightarrow{\cong} Y ,
\]
and \emph{quasi injections}
\[
 \iota^{1}_{X,Y} \defeq 
 X \xrightarrow{\cong}
 X \otimes I \xrightarrow{\id_{X} \otimes \mathbf{0}_{I,Y}}
 X \otimes Y
\qquad
 \iota^{2}_{X,Y} \defeq 
 Y \xrightarrow{\cong}
 I \otimes Y \xrightarrow{\mathbf{0}_{I,X} \otimes \id_{Y}}
 X \otimes Y
 .
\]
Then, for any \(f : X_1 \otimes X_2 \to Y_1 \otimes Y_2\), we have
\(f_{i,j} : X_i \to Y_j\) (\(i,j = 1,2\)) by
\[
 f_{X_i,Y_j} \defeq \iota^{i}_{X_1,X_2} \seqcomp f \seqcomp \rho^{j}_{Y_1,Y_2}.
\]
We write \(f_{i,j}\) also as \(f_{X_i,Y_j}\) if there is no ambiguity.

Now, by the following result, we can conclude that $\semCatrmc$ is a TSMC. 
The formula~\eqref{eq:categoricalExecFormula} below comes from Girard's \emph{execution formula}~\cite{girard1989geometry}. 

\begin{proposition}[{\cite[Corollary 5.4]{hoshino12}}]\label{prop:hoshino}
Let $\mathbb{C}$ be a strong unique decomposition category. If the sum
\begin{equation}\label{eq:categoricalExecFormula}
 \trace{X}{Y}{Z}{}(f) \defeq f_{X,Y} + \Sigma\bigl(\,f_{X,Z}\seqcomp (f_{Z, Z})^{d}\seqcomp f_{Z, Y} \,\bigr)_{d\in \Nat} 
\end{equation}
is defined for any objects $X, Y, Z$ and any arrow $f:Z\otimes X\rightarrow Z\otimes Y$, then $\tr$ is a trace operator on $\mathbb{C}$.
\qed 
\end{proposition}

The arrows $f_{X,Y}, f_{X, Z}, f_{Z, Z}, f_{Z, Y}$ have the following concrete descriptions in our semantic category $\semCatrmc$. 
For an arrow $f:l+m\rightarrow l+n
$, let $f(i, j) \defeq (p_{i, j}, r_{i, j})$ for each $i\in \nset{l+m}$ and $j\in \nset{l+n}$, then $f_{m, n}:m\rightarrow n$ is given by $f_{m, n}(i', j') = (p_{l+i',l+j'}, r_{l+i', l+j'})$ for each $i'\in \nset{m}$ and $j'\in \nset{n}$. The others are described similarly. 

The sum on the right-hand-side of~\eqref{eq:categoricalExecFormula} is  always defined in the semantic category $\semCatrmc$, as expected. This is because, otherwise, we would have either some reachability probability or expected reward diverge to infinity; but reachability probabilities are easily seen to be bounded by $1$;  and expected rewards do not diverge under our current setting (in particular \cref{def:openMDPReachProbExpCumRew}). Therefore, by \cref{prop:hoshino}, we conclude that the category $\semCatrmc$ has the trace operator.

\subsection{The \emph{Change of Base} Construction}
\label{subsec:change_of_base}
The 
 general theory of change of base is built on the theory of 
\emph{enriched categories}~\cite{kelly1982basic} (see e.g.~\cite{cruttwell2008normed}); the latter is out of the scope of the paper. Instead, we exhibit an instance of the general theory, restricting to a particular functor (namely the powerset functor $\pMnd\colon \Sets\to\Sets$).

\begin{definition}[the powerset functor $\pMnd$]
The \emph{powerset functor} $\pMnd$ is an endofunctor on the category $\Sets$ of sets, mapping a set $X$ to the set $\pMnd(X)$ of its subsets, and a function $f:X\rightarrow Y$ to the function $\pMnd(f): \pMnd(X) \rightarrow \pMnd(Y)$, where $\pMnd(f)(S) \defeq \{ f(s) \mid s\in S\}$ for each $S\in \pMnd(X)$. 
\end{definition}

The powerset functor satisfies the following property, which makes it a \emph{lax monoidal} functor. 

\begin{proposition}[$\pMnd$ is lax monoidal]
    For any sets $X, Y$, let $\mu_{X, Y}:\pMnd(X)\times \pMnd(Y) \rightarrow \pMnd(X\times Y)$ be the function given by $\mu_{X, Y}(S, T) \defeq \{(s, t)\mid s\in S,\, t\in T\}$.
    In addition, we define $\epsilon:\mathbf{1}\rightarrow \pMnd(\mathbf{1})$ by $\epsilon(\star) \defeq \{\star\}$, where $\mathbf{1}$ is the singleton set $\{\star\}$. 
    Then, the following diagrams commute.
    \begin{description}
        \item[The associative law] For any sets $X, Y, Z$,
       \begin{equation*}
        \begin{tikzcd}
        (\pMnd(X)\times \pMnd(Y))\times \pMnd(Z) \arrow[r, "\cong"] \arrow[d, "\mu_{X, Y}\times id_{\pMnd(Z)}"]&  \pMnd(X)\times (\pMnd(Y)\times \pMnd(Z)) \arrow[d, "\id_{\pMnd(X)}\times \mu_{Y, Z}"]\\
        \pMnd(X\times Y) \times \pMnd(Z) \arrow[d, "\mu_{X\times Y, Z}"]& \pMnd(X)\times \pMnd(Y\times Z) \arrow[d, "\mu_{X, Y\times Z}"]\\
        \pMnd((X\times Y)\times Z)\arrow[r, "\pMnd(\cong)"] & \pMnd(X\times (Y\times Z))
        \end{tikzcd}             
       \end{equation*} 
       \item[The unit law] For each set $X$, 
       \begin{equation*}
        \begin{tikzcd}
        \mathbf{1}\times\pMnd(X)\arrow[r, "\epsilon\times \id_{\pMnd(X)}"] \arrow[d, "\cong"]&  \pMnd(\mathbf{1})\times\pMnd(X)\arrow[d,"\mu_{\mathbf{1}, X}"]\\
        \pMnd(X) & \pMnd(\mathbf{1}\times X) \arrow[l, "\pMnd(\cong)"]
        \end{tikzcd}             
       \end{equation*}
       \begin{equation*}
        \begin{tikzcd}
        \pMnd(X)\times \mathbf{1}\arrow[r, "\id_{\pMnd(X)}\times \epsilon"] \arrow[d, "\cong"]&  \pMnd(X)\times \pMnd(\mathbf{1})\arrow[d,"\mu_{X,\mathbf{1}}"]\\
        \pMnd(X) & \pMnd(X\times \mathbf{1}) \arrow[l, "\pMnd(\cong)"]
        \end{tikzcd}             
       \end{equation*}

   Here $\cong$ represent canonical bijections.
    \end{description}
    \qed 
\end{proposition}

The following describes the change of the base construction~\cite{Eilenberg65}, restricting to  the powerset functor $\pMnd$. 
\begin{definition}[change of base]
\label{def:change_of_base}
Let $\mathbb{C}$ be a locally small (i.e.\ $\Sets$-enriched) category. The \emph{change of base} $\pMnd_{\star}(\mathbb{C})$ of $\mathbb{C}$ by $\pMnd$ is the category whose object is that of $\mathbb{C}$, and whose arrow $F:X\rightarrow Y$ is a set of arrows of $\mathbb{C}$, i.e., $F\in \pMnd(\mathbb{C}(X, Y))$. For each object $X$, the identity $\id_{X}\colon X\to X$ in $\pMnd_{\star}(\mathbb{C})$ is given by $\{\id_X\}$, where $\id_X\colon X\to X$ is the identity in $\mathbb{C}$. For any arrows $F:X\rightarrow Y$ and $G:Y\rightarrow Z$, their (sequential) composition $F\seqcomp G$ in $\pMnd_{\star}(\mathbb{C})$ is given by $\{f\seqcomp g\mid f\in F,\, g\in G\}$, where $f\seqcomp g$ is the sequential composition in $\mathbb{C}$.
\end{definition}

Let $\mathbb{C}$ be a locally small category, and $X, Y, Z\in \mathbb{C}$ be objects. 
Then 
the identity $\id_X$  can be identified with a (set-theoretic) function $j_X: \mathbf{1}\rightarrow \mathbb{C}(X, X)$, via $j_X(\star) \defeq \id_X $. Similarly,  the sequential composition $\seqcomp_{X, Y, Z}$ of arrows of suitable types can be identified with a function  $\seqcomp_{X, Y, Z}:\mathbb{C}(X, Y)\times \mathbb{C}(Y, Z)\rightarrow \mathbb{C}(X, Z)$. 
Then, the change of base construction by the powerset functor $\pMnd$---the definitions of $\id$ and $\seqcomp$, in particular---can be described using the lax monoidal structure $\mu$ of $\pMnd$, as illustrated in the following diagrams.
\begin{equation*}
    \begin{tikzcd}
        \pMnd(\mathbb{C}(X, Y))\times \pMnd(\mathbb{C}(Y, Z))\arrow[d, "\pMnd(\mu_{\mathbb{C}(X, Y), \mathbb{C}(Y, Z)})",swap] \arrow[rd, "\seqcomp^{\pMnd_{\star}(\mathbb{C})}"]\\
        \pMnd(\mathbb{C}(X, Y)\times\mathbb{C}(Y, Z)) \arrow[r, "\pMnd(\seqcomp^{\mathbb{C}})"]& \pMnd(\mathbb{C}(X, Z))
    \end{tikzcd}
\end{equation*}

\begin{equation*}
    \begin{tikzcd}
        \mathbf{1}\arrow[d, "\epsilon",swap] \arrow[rd, "j^{\pMnd_{\star}(\mathbb{C})}_X"]\\
        \pMnd(\mathbf{1}) \arrow[r, "\pMnd(j^{\mathbb{C}}_X)"]& \pMnd(\mathbb{C}(X, X))
    \end{tikzcd}
\end{equation*}

Some properties of $\pMnd$ as a lax monoidal functor are used for proving the following proposition. 
\begin{proposition}
    The unit and sequential composition operators defined in~\cref{def:change_of_base} satisfy the equational axioms of categories. That is, the unit and associative laws hold. \qed
\end{proposition}

\else
\fi

\end{document}